\documentclass[tbtags,authorcolumns,numberwithinsect,footnotebackref]{no-lipics}

\usepackage{amssymb}
\usepackage{mathtools}
\usepackage{bm}
\usepackage{stmaryrd}
\usepackage[draft]{fixme}
\usepackage{booktabs}
\usepackage{tikz-3dplot}
\usepackage{caption}
\usepackage{subcaption}
\usepackage{cases}
\usepackage{colortbl}

\usetikzlibrary{shapes, positioning}

\SetKwComment{Comment}{$\triangleright$\ \rm}{}

\usetikzlibrary{arrows, arrows.meta,shapes.misc, decorations.pathreplacing,decorations.pathmorphing,calligraphy, patterns, patterns.meta, calc}
\tikzset{dotmark/.style={circle,fill,inner sep=1.5pt}}
\tikzset{emptymark/.style={circle,draw,fill=white,inner sep=1.5pt}}
\tikzset{crossmark/.style={thick,inner sep=1.5pt}}

\newcommand{\eps}{\varepsilon}
\newcommand{\A}{\mathcal{A}}

\SetKwInput{KwInput}{Input}
\SetKwInput{KwOutput}{Output}

\def\ShowAuthNotes{1}
\ifnum\ShowAuthNotes=1
\newcommand{\authnote}[3]{\textcolor{#3}{[{\bf #1:} { {#2}}]}}
\else
\newcommand{\authnote}[3]{}
\fi

\newcommand{\Oh}{\mathcal{O}}
\newcommand{\Ohtilde}{\tilde{\Oh}}

\DeclareMathOperator*{\poly}{poly}
\DeclareMathOperator*{\polylog}{polylog}

\def\fragmentco#1#2{\bm{[}\,#1\,\bm{.\,.}\,#2\,\bm{)}}
\def\fragmentoc#1#2{\bm{(}\,#1\,\bm{.\,.}\,#2\,\bm{]}}
\def\fragmentoo#1#2{\bm{(}\,#1\,\bm{.\,.}\,#2\,\bm{)}}
\def\fragment#1#2{\bm{[}\,#1\,\bm{.\,.}\,#2\,\bm{]}}

\newcommand{\yes}{{\sc yes}\xspace}
\newcommand{\no}{{\sc no}\xspace}

\newcommand{\APSP}{\textsf{APSP}\xspace}
\newcommand{\RP}{\textsf{RP}\xspace}
\newcommand{\SSRP}{\textsf{SSRP}\xspace}
\newcommand{\FRP}{\textsf{FRP}\xspace}

\newcommand{\hi}{\hat{\imath}}
\newcommand{\hj}{\hat{\jmath}}

\newcommand{\sT}{\mathsf{T}}

\newcommand{\bG}{\mathbf{G}}
\newcommand{\bH}{\mathbf{H}}

\newcommand{\bL}{\mathbf{L}}
\newcommand{\bR}{\mathbf{R}}
\newcommand{\bA}{\mathbf{A}}
\newcommand{\bB}{\mathbf{B}}

\newcommand{\bT}{\mathbf{T}}

\renewcommand{\epsilon}{\varepsilon}

\DeclarePairedDelimiter\abs{\lvert}{\rvert}
\DeclarePairedDelimiter\norm{\lVert}{\rVert}

\makeatletter
\let\oldabs\abs
\def\abs{\@ifstar{\oldabs}{\oldabs*}}
\let\oldnorm\norm
\def\norm{\@ifstar{\oldnorm}{\oldnorm*}}
\makeatother

\def\problembox#1{%
    \vspace{2mm}%
    \noindent\fbox{%
    \begin{minipage}{.985\linewidth}%
        #1
    \end{minipage}%
    }%
    \vspace{2mm}%
}
\newcommand{\defproblem}[3]{%
    \problembox{%
        \textbf{#1}\\
        {\bf{Input:}} #2 ~\\
        {\bf{Output:}} #3
    }%
}

\makeatletter
\renewenvironment{cases}{%
  \matrix@check\cases\env@cases
}{%
  \endarray\right.%
}
\def\env@cases{%
  \let\@ifnextchar\new@ifnextchar
  \left\lbrace
  \def\arraystretch{1.1}%
  \array{@{\;}c@{\quad}l@{}}%
}
\makeatother

\def\mid{\ensuremath :}

\def\emptyset{\varnothing}

\newcommand\thefont{\expandafter\string\the\font}

\title{Undirected Replacement Paths:\\Dual Fault Reduces to Single Source}

\author{Jakob Nogler}{Massachusetts Institute of Technology\\Cambridge, United States}{jnogler@mit.edu}{https://orcid.org/0009-0002-7028-2595}{Supported by the Akamai Presidential Fellowship}
\author{Virginia Vassilevska Williams}{Massachusetts Institute of Technology\\Cambridge, United States}{virgi@mit.edu}{https://orcid.org/0000-0003-4844-2863}{Supported by NSF Grant CCF-2330048, BSF Grant 2024233 and a Simons Investigator Award}

\authorrunning{J. Nogler and V. Vassilevska Williams}
\titlerunning{Undirected Replacement Paths: Dual Fault Reduces to Single Source}

\nolinenumbers

\begin{document}
\maketitle
\begin{abstract}
Given a graph and two fixed vertices $s$ and $t$, the \emph{Replacement Path} Problem (\RP) is to compute for every edge $e$, the distance between $s$ and $t$ when $e$ is removed. There are two natural extensions to \RP:

\begin{itemize}
    \item \textbf{Single Source Replacement Paths (\SSRP):} Given a graph $\bG$ and a source node $s$, compute for every vertex $v$ and every edge $e$ the $s$-$v$ distance in $\bG \setminus e$. That is, we do not fix the target anymore.
    \item \textbf{$2$-Fault Replacement Paths (2-\FRP):} Given a graph $\bG$ and two nodes $s$ and $t$,  compute for every pair of edges $e,e'$ the $s$-$t$ distance in $\bG \setminus e,e'$. That is, we consider two failures instead of one.
\end{itemize}

Previously, there was no known formal reduction between \SSRP and 2-\FRP. It seemed plausible that 2-\FRP would be computationally harder because there are no settings where 2-\FRP admits a faster algorithm than \SSRP. In directed unweighted graphs there is a provable gap in complexity, and in undirected graphs many of the known 2-\FRP algorithms in a variety of settings are much slower than those for \SSRP in the same setting. 

The main contribution of this paper is a tight reduction from undirected $2$-\FRP to undirected \SSRP, showing that contrary to prior intuition, 2-\FRP is not harder than \SSRP.
As our reduction is weight-preserving, we obtain the first algorithms for $2$-\FRP that match the best-known runtimes for \SSRP:
\begin{enumerate}[(a)]
    \item $\Ohtilde(M n^{\omega})$ for weights in $\fragment{1}{M}$ [Grandoni and Vassilevska Williams, FOCS 2012 \& TALG 2019], improving upon $\Oh(Mn^{2.87})$ [Chechik and Zhang, ICALP 2024];
    \item $n^3/2^{\Omega(\sqrt{\log n})}$ for weights in $\fragment{1}{\poly(n)}$ [Grandoni and Vassilevska Williams, FOCS 2012 \& TALG 2019], improving over the previous $n^3\polylog(n)$ running time [Vassilevska W., Woldeghebriel and Xu, FOCS 2022];
    \item $\Ohtilde(mn^{1/2}+n^{2})$ combinatorial time for unweighted graphs [Chechik and Cohen, SODA 2019], and more generally for rational weights in $[1,2]$ [Chechik and Magen, ICALP 2020], improving upon $\Ohtilde(n^{3-1/18})$ [Chechik and Zhang, ICALP 2024].
\end{enumerate}

We complement these upper bounds with tight lower bounds under established fine-grained hypotheses.
\end{abstract}

\clearpage

\section{Introduction}

In the \emph{Replacement Path} problem (\RP) one is given a graph $\bG$ and two fixed nodes $s$ and $t$ and is asked to compute for every possible failed edge $e$ of $\bG$, the distance between $s$ and $t$ in the subgraph $\bG\setminus e$. 
Beyond its clear application for distance computation in error-prone networks, the problem serves as a critical algorithmic primitive in diverse domains: it can be used to identify segments in biological sequence alignments \cite{BW84}, to calculate Vickrey Prices in an auction on a distributed network \cite{NR01, HS01}, and as a subproblem in an algorithm for computing the $k$th shortest simple $s$-$t$ paths in both directed and undirected graphs \cite{L72, Y71, R07, RZ12}.

For undirected graphs, a near-linear, $\Ohtilde(m)$\footnote{The $\Ohtilde(\cdot)$ notation suppresses factors in $\Oh(\polylog(n))$.} time  algorithm was first proposed by Malik et al. \cite{MMG89} (with corrections in \cite{SBK95}), and later improved by Nardelli, Proietti, and Widmayer to $\mathcal{O}(m\alpha(n,m))$ \cite{NPW01}. Interestingly, \RP is computationally harder in directed graphs: the general weighted version is equivalent to the All-Pairs Shortest Paths (\APSP) problem, with reductions existing in both directions \cite{GL09, VWW18}.

The popular \APSP Hypothesis from Fine-Grained Complexity (see e.g. \cite{vsurvey}) states that \APSP in $n$-node graphs requires $n^{3-o(1)}$ time in the word-RAM model of computation. Due to the equivalence, the \APSP Hypothesis implies an $n^{3-o(1)}$ time lower bound
for \RP in directed graphs as well. Truly subcubic\footnote{A running time is truly subcubic if it is of the form $\Oh(n^{3-\eps})$ for constant $\eps>0$.} time algorithms exist for graphs with small integer weights \cite{RZ12, VW11, WY13} and for the approximate version of the problem \cite{B10}.

The two most studied generalizations of \RP involve either increasing the number of failures or relaxing the target constraint. The first considers the failure of \emph{two} edges simultaneously, while the second generalizes the target by requiring the computation of the distances from a fixed source $s$ to all other vertices $v \in V(\bG)$ after the failure of each edge $e$.
More formally, the two obtained problems read as follows.

\defproblem
{Single Source Replacement Paths Problem (\SSRP)}
{A graph $\bG$ and a node $s \in V(\bG)$.}
{The distance $d_{\bG\setminus e}(s,v)$ for each edge $e \in E(\bG)$ and node $v \in V(\bG)$.}

\defproblem
{Dual-Fault Replacement Path Problem (2-\FRP)}
{A graph $\bG$ and nodes $s,t \in V(\bG)$.}
{The distance $d_{\bG\setminus e,e'}(s,t)$ for all $e,e' \in E$.}

We remark that for both problems the output size is actually quadratic.
In \SSRP this is because for each $v \in V(\bG)$ we only care about $d_{\bG\setminus e}(s,v)$ for $e$ on the shortest path $\pi_{\bG}(s,v)$ between $s$ and $v$.
Similarly, for $2$-\FRP we only need to consider $e,e'$ on the shortest $s$-$t$ path in $\bG$ and $\bG \setminus e$, respectively.

Next, we summarize our algorithmic understanding of \SSRP and 2-\FRP (see also \cref{table:algos}).

\subparagraph*{Previous results on \SSRP.}
In a general graph, \SSRP admits a simple $\Ohtilde(mn)$ time algorithm by solving SSSP in $\bG \setminus e$ for each $e$ in the shortest path tree of $s$ in $\bG$. This is optimal even in undirected graphs under the \APSP hypothesis \cite{CC19} (a previous $\Omega(nm)$ lower bound was already established in the path comparison model \cite{HSB07, KKP}).

The first to study \SSRP in depth were Grandoni and Vassilevska Williams \cite{GVW12, GV19}, who positioned the problem not only as a generalization of \RP but also as a fundamental component for constructing faster \emph{distance sensitivity oracles} (DSOs). They provided algorithms running in $\Ohtilde(M^{1/(4-\omega)}n^{2 + 1/(4-\omega)})$ time for directed graphs with weights in $\fragment{-M}{M}$ and in $\Ohtilde(Mn^{\omega})$ time for weights in $\fragment{1}{M}$; the latter matches the runtime for the \RP problem itself. Notably, \cite{GVW12, GV19} demonstrates that \SSRP in general graphs also reduces to \APSP, establishing a computational equivalence between the two problems and, by extension, directed \RP. Subsequent improvements for the $\fragment{-M}{M}$ weight domain were later provided by \cite{GPVWX21} and \cite{BCFS21}.

Combinatorial algorithms for \RP and \SSRP have also been intensively studied. For \RP in $m$-edge, $n$-node directed unweighted graphs, Roditty and Zwick \cite{RZ12} gave an $\Ohtilde(m\sqrt n)$ time algorithm and Vassilevska Williams and Williams \cite{VWW18} showed that this is optimal for ``combinatorial''\footnote{Not a well-defined notion but used in the literature anyway.} algorithms unless the so called BMM Hypothesis is violated; the latter asserts that Boolean Matrix Multiplication (BMM) does not admit an $\Oh(n^{3-\eps})$ time combinatorial algorithm for any $\eps>0$. Chechik and Cohen \cite{CC19} showed that roughly the same running time and conditional lower bounds hold for undirected unweighted \SSRP by giving an $\Ohtilde(\sqrt{n}m + n^{2})$ time combinatorial algorithm and a matching combinatorial lower bound under the BMM Hypothesis. Then \cite{CM20} obtained the same running time for directed unweighted graphs and even for directed graphs with rational weights upper bounded by a constant (derandomized in \cite{BCFS21}) and showed that the latter version requires $mn^{0.5-o(1)}$ time even under the \APSP hypothesis. Lastly, \SSRP has also been studied in specialized graph classes, such as planar graphs \cite{BLM12}, and in the approximate setting \cite{BK13, BCGLPP18, BCHL20, BCFS21, HKIM24}. (In the latter, to bypass the quadratic output lower bound, \SSRP is often framed as an oracle problem.) 

\newcommand{\citesub}[1]{\scriptsize{\cite{#1}}}

\begin{table}[t]
\small
\centering
\setlength{\tabcolsep}{5pt} 

\begin{tabular}{|c|c|cc|cc|}
\hline
\multirow{2}{*}{\textbf{Setting}} &
\multirow{2}{*}{\textbf{Weights}} 
  & \multicolumn{2}{c|}{\SSRP} 
  & \multicolumn{2}{c|}{2-\FRP} \\ 
& & Upper Bound & Lower Bound & Upper Bound & Lower Bound \\ \hline
undir. & $\fragment{1}{n^c}$ & \APSP-time \citesub{GV19} & \APSP-hard \citesub{CM20} & \cellcolor{yellow!40}$\Ohtilde(n^3)$ \citesub{VWWX22} & \APSP-hard \citesub{HLNVW17} \\ \hline
dir. & $\fragment{1}{n^c}$ & \APSP-time \citesub{GV19} & \APSP-hard \citesub{CM20} & $\Ohtilde(n^3)$ \citesub{VWWX22} & \APSP-hard \citesub{VWW18} \\ \hline
dir. & $\fragment{-M}{M}$ & $\Oh(M^{0.81}n^{2.49})$\citesub{GPVWX21} & \textbf{?} & $\Oh(Mn^{2.87})$ \citesub{CT24} & \textbf{?} \\ \hline
undir. & $\fragment{1}{M}$ & $\Ohtilde(Mn^{\omega})$ \citesub{GV19}  & \textbf{?} & \cellcolor{yellow!40}$\Oh(Mn^{2.87})$ \citesub{CT24} & \textbf{?} \\ \hline
dir. & $\fragment{1}{M}$ & $\Ohtilde(Mn^{\omega})$ \citesub{GV19}  & \textbf{?} & $\Oh(Mn^{2.87})$ \citesub{CT24} & \textbf{?} \\ \hline
undir. & $\mathbb{Q} \cap [1,2]$ & $\Ohtilde(mn^{1/2}+n^{2})$ \citesub{CM20}& $mn^{1/2-o(1)}$ \citesub{CM20} & \cellcolor{yellow!40}$\Ohtilde(n^3)$ \citesub{VWWX22} & \cellcolor{yellow!40}\textbf{?} \\ \hline
dir. & $\mathbb{Q} \cap[1,2]$ & $\Ohtilde(mn^{1/2}+n^{2})$ \citesub{CM20} & $mn^{1/2-o(1)}$ \citesub{CM20} & $\Ohtilde(n^3)$ \citesub{VWWX22} & \cellcolor{yellow!40}\textbf{?} \\ \hline
undir. & unweigh. & $\Ohtilde(mn^{1/2}+n^{2})$ \citesub{CC19} & $mn^{1/2-o(1)}$ \citesub{CC19} & \cellcolor{yellow!40}$\Oh(n^{3-1/18})$\citesub{CT24} & \cellcolor{yellow!40}\textbf{?} \\ \hline
dir.  & unweigh. & $\Ohtilde(mn^{1/2}+n^{2})$ \citesub{CM20} & $mn^{1/2-o(1)}$ \citesub{CM20}& $\Oh(n^{3-1/18})$\citesub{CT24} & $n^{8/3-o(1)}$ \citesub{VWWX22} \\ \hline
\end{tabular}
 \bigskip
\caption{Upper and lower bounds for \SSRP and 2-\FRP under different graph settings before this work. Under the current bounds \APSP can be solved in time $n^3/\exp({\Omega(\sqrt{\log n})})$ \cite{W18}. The last two rows refer to combinatorial algorithms.
The cells colored in yellow are improved by this paper.}
\label{table:algos}
\end{table}

\subparagraph*{Previous results on 2-\FRP.} 
Early on, it was conjectured that $2$-\FRP might admit an $\Ohtilde(n^3)$ time algorithm \cite{BG04}, matching the complexity of the single-failure case in directed graphs. The existence of such an algorithm was eventually established nearly twenty years later by Vassilevska Williams, Woldeghebriel, and Xu \cite{VWWX22}. Interestingly, in undirected graphs, $2$-\FRP reduces to \APSP \cite{HLNVW17,CDWX25}. Consequently, the two-failure case in both undirected and directed graphs shares the same cubic-time complexity. This is in stark contrast to the single-failure setting, where the problems exhibit significantly different computational profiles.

Similar to the developments in \SSRP, research has also focused on small-weight and approximate algorithms for $2$-\FRP. In \cite{VWWX22}, an (algebraic) algorithm with a running time of $\Oh(M^{2/3}n^{2.9153})$ was introduced for weights in $\fragment{-M}{M}$. The authors also obtained a conditional lower bound for combinatorial algorithms for unweighted directed graphs under the BMM Hypothesis. However, this lower bound of $n^{8/3-o(1)}$, was subcubic and thus left open the possibility that a subcubic combinatorial algorithm is possible. Subsequently, Chechik and Zhang \cite{CT24} were able to obtain the first truly subcubic combinatorial algorithm for unweighted $2$-\FRP, running in $\Oh(n^{3-1/18})$ time. They also obtained a slightly improved $\Oh(Mn^{2.8716})$ (algebraic) algorithm for weights in $\fragment{-M}{M}$.

Finally, regarding approximation schemes, Chechik and Zhang \cite{CZ24} also demonstrated that a $(1+\varepsilon)$-approximation for $2$-\FRP can be achieved in near-optimal $\Ohtilde_{\varepsilon}(n^2)$ time.

\subparagraph*{\SSRP vs 2-\FRP.} 
The preceding literature suggests that 2-\FRP is inherently more difficult than \SSRP. There are no settings where the former admits a faster algorithm than the latter; even in the undirected case—specifically for small edge weights and combinatorial 
algorithms in unweighted graphs—a clear complexity gap exists. 
In unweighted directed graphs, this gap is even provable for combinatorial algorithms: 
directed \SSRP admits an $\Ohtilde(mn^{1/2}+n^{2})$-time combinatorial algorithm \cite{CM20}, whereas directed 2-\FRP is subject to a conditional combinatorial lower bound of $n^{8/3-o(1)}$ \cite{VWWX22}.

We consider the following tantalizing question:
\begin{center}
{\em Is 2-\FRP computationally harder than \SSRP in undirected graphs? \\Is it harder than \APSP in arbitrary integer weighted graphs?}
\end{center}

\subsection*{Our Results}

In this paper, we demonstrate that for undirected graphs the opposite of our intuition holds: we show that the seemingly more difficult $2$-\FRP problem can, in fact, be reduced to the seemingly easier \SSRP problem. 

Our reduction considers two cases: whether both failed edges $e, e'$ lie on the $s$-$t$ shortest path or if only one does (note that at least one must always lie on this path). For the former case, we show that it can be handled in the same time as \SSRP on graphs with the same edge weight set.

\begin{restatable*}{mtheorem}{both} \label{thm:both}
   Let $\sT_{\SSRP}(n,m, W)$ be the time needed to compute \SSRP in an undirected graph with $n$ nodes, $m$ edges, and positive weight set $W$. If there is a constant $0 < \rho \leq 1$ such that
   \begin{center}
       $\sT_{\SSRP}(n_1,m_1,W) + \sT_{\SSRP}(n_2,m_2,W) \leq \rho \cdot \sT_{\SSRP}(n,m,W)$
   \end{center}
   for all $n_1+n_2\leq n+1$, $m_1+m_2 \leq m$ and $n_1,n_2 \leq 2n/3$,
   then we can compute all values $d_{\bG \setminus e,e'}(s,t)$ such that $e,e' \in \pi_{\bG}(s,t)$ in time $\Oh(\sT_{\SSRP}(n,m, W) + n^2)$ if $\rho < 1$, and $\Oh(\sT_{\SSRP}(n,m, W) \log^2 n + n^2)$ otherwise.
\end{restatable*}

Note that in \cref{thm:both}, the assumptions on $\sT_{\SSRP}$ are more than reasonable and align with standard algorithmic running times. For instance, any function of the form $n^{\alpha}m$ for $\alpha > 0$, or $n^{\alpha}m^{\beta}$ for $\alpha \geq 1$ and $\beta > 0$, satisfies these conditions with $\rho < 1$; similarly, any function of the form $m^{\alpha}$ for $\alpha \geq 1$ satisfies it with $\rho = 1$.

For the latter case when $e' \notin \pi_{\bG}(s,t)$, we obtain a reduction in a stronger sense. Specifically, we show that it suffices to compute \SSRP on the \emph{same graph} from both $s$ and $t$, followed by $\Ohtilde(n^2)$ additional work.

\begin{restatable*}{mtheorem}{single} \label{thm:single}
   Let $\bG$ be an undirected graph with $n$ nodes and positive weight set, and let $s,t \in \bG$ be two nodes. Further, let $\sT$ be the time needed to solve \SSRP from $s$ and $t$. Then, we can compute all values $d_{\bG \setminus e,e'}(s,t)$ such that $e \in \pi_{\bG}(s,t)$ and $e' \notin \pi_{\bG}(s,t)$ in time $\sT + \Oh(n^2 \log^2 n)$.
\end{restatable*}

Our reduction also yields faster algorithms; specifically, it allows us to take the upper bounds for \SSRP from each row for undirected graphs in \cref{table:algos} and apply them directly to 2-\FRP.
More specifically, combining \cref{thm:single}, \cref{thm:both} with \cite{GVW12, GV19} and \cite{CC19,CM20}\footnote{Works \cite{GVW12, GV19,CM20} are formulated on directed graphs, but it is possible to verify that the algorithms work on undirected graphs as well.} (and the fastest \APSP algorithm to date \cite{W18}), we obtain the following faster algorithms for 2-\FRP in undirected graphs.

\begin{corollary}\label{cor:algos}
    In undirected graphs there are the following algorithms for 2-\FRP:
    \begin{enumerate}[(a)]
        \item An $\Ohtilde(M n^{\omega})$-time algorithm for weights in $\fragment{1}{M}$;
        \label{it:algos:a}
        \item An $n^3/2^{\Omega(\sqrt{\log n})}$-time algorithm for weights in $\fragment{1}{\poly(n)}$; and
        \label{it:algos:b}
        \item An $\Ohtilde(mn^{1/2}+n^{2})$-time (combinatorial) algorithm for rational weights in $[1,2]$;  
        \label{it:algos:c}\lipicsEnd
    \end{enumerate}
\end{corollary}

As a side-result of our reduction, we establish that undirected 2-\FRP is \emph{computationally equivalent} to the following other problems on general weighted graphs: directed \RP, directed or undirected \SSRP, and directed or undirected \APSP. In other words, tolerating a single failure is computationally equivalent to tolerating two, provided we make the underlying network undirected.

\subparagraph*{Lower bounds.}
We complement the algorithmic results presented in \cref{cor:algos} with corresponding lower bounds. The tightness of \cref{cor:algos}\ref{it:algos:b} is immediate; by combining the results of \cite{HLNVW17, CDWX25} with our reduction, we establish the \APSP-equivalence of 2-\FRP for weights in $\fragment{1}{\poly(n)}$.

Regarding \cref{cor:algos}\ref{it:algos:a} and \cref{cor:algos}\ref{it:algos:c}, we provide lower bounds based on the BMM and \APSP hypotheses for the regime $m = \Omega(n^{3/2})$, where the output-size lower bound is non-dominant.

\begin{restatable*}{lemma}{lb} \label{thm:lb}
    Consider the problem of solving 2-\FRP in a graph with $n$ nodes and $m$ edges. 
    If $m = \Omega(n^{3/2})$, then for any $\varepsilon > 0$, there is no
    \begin{enumerate}[(i)]
        \item $\Oh(mn^{1/2-\varepsilon})$ combinatorial algorithm for unweighted graphs, unless the BMM hypothesis is false;
        \label{it:lb:i}
        \item $\Oh(mn^{1/2-\varepsilon})$ algorithm for graphs with rational weights in $[1,2]$, unless the \APSP hypothesis is false; 
        \label{it:lb:ii}\qedhere
    \end{enumerate}
\end{restatable*}

Interestingly, our reduction proceeds via a version of triangle detection in sparse graphs with unbalanced partition sizes; for certain values of $m$, this could imply a non-combinatorial lower bound even for the unweighted case. (See \cref{sec:lb} for a more detailed discussion.)

\subparagraph*{A note on the structural connection between \SSRP and 2-\FRP.}
From a \emph{structural} perspective, existing results already demonstrate a relationship between dual-failure fixed-endpoints and single-source single-failure instances in undirected graphs. To see this, we recall the Restoration Lemma, a classic result by Afek et al. \cite{BBAKCM01} (recently improved by \cite{BW25}).

\begin{theorem}[Restoration Lemma \cite{BBAKCM01}]
    Given an undirected graph $\bG$, any shortest path $\pi_{\bG \setminus F}(s,t)$ where $F \subseteq E(\bG)$ and $|F|=f$ can be decomposed into $f+1$ shortest paths in $\bG$ and $f$ edges in-between. \lipicsEnd
\end{theorem}

This lemma implies that any shortest path $\pi_{\bG \setminus e,e'}(s,t)$ can be decomposed into $\pi_{\bG \setminus e}(s,u) \circ \{u,v\} \circ \pi_{\bG \setminus e}(v,t)$ for some $u,v \in V(\bG)$. This establishes a clear structural link between dual-failure scenarios with source/target $s,t$ and single-failure instances with source $s$ and $t$ and arbitrary target.  However, it remains highly unclear how to use this lemma \emph{algorithmically}; the search for the intermediate vertices $u,v$ ranges over $n$ possible vertices each, and both must ensure that the subpaths $\pi_{\bG \setminus e}(s,u)$ and $\pi_{\bG \setminus e}(v,t)$ do not contain the second failed edge $e'$. In fact, our reduction does not rely on this lemma; instead, for most cases, we manage to use a simpler decomposition consisting of a shortest path in $\bG$, an individual edge, and a shortest path in $\bG \setminus e$.

\subparagraph*{Open questions.}

We pose the following two questions connected to this work:
\begin{itemize}
    \item Is a reduction from $2$-\FRP to \SSRP possible for directed graphs, at least in a weaker form?
    \item In the directed case, if no reduction from $2$-\FRP to \SSRP can be found, can we at least demonstrate that 2-\FRP reduces to \APSP, thereby establishing that one and two failures are computationally equivalent in general graphs?
\end{itemize}

\section{Overview}
\label{sec:to}

In this section, we give a brief overview of the reduction.
The overview is composed of two parts: one for the case where both failed edges $e$ and $e'$ are on the shortest path $\pi_{\bG}(s,t)$  (i.e., \cref{thm:both}) and one for the case when only $e$ is on  $\pi_{\bG}(s,t)$ (i.e., \cref{thm:single}). The notation should be clear from the context; if more clarity is needed, the reader can skip ahead to \cref{sec:prelims}.

Throughout this section, for simplicity, assume $\bG$ has unique shortest paths (with the essential property that subpaths of shortest paths are shortest paths themselves). Furthermore, assume that we are given an algorithm for computing \SSRP on any graph utilizing the same weight set as $\bG$.

\subsection*{Case $e',e \in \pi_{\bG}(s,t)$}

To compute $d_{\bG}(s,t)$ for all $e, e' \in \pi_{\bG}(s,t)$, say we adopt a divide-et-impera approach. Then, the most natural thing is to attempt the following strategy: first select a midpoint $c \in \pi_{\bG}(s,t)$ (see \cref{fig:both:a}) and compute $d_{\bG\setminus e,e'}(s,t)$ for all $e \in \pi_{\bG}(s,c)$ and $e' \in \pi_{\bG}(c,t)$ (here our \SSRP subroutine should prove helpful). We then construct two auxiliary graphs, $\bA$ and $\bB$ (sharing the same weight set as $\bG$), and recurse to compute  $d_{\bG \setminus e,e'}(s,t)$ for all $e,e' \in \pi_{\bG}(s,c)$ and $e,e' \in \pi_{\bG}(c,t)$, recursively.

\begin{figure*}[!t]
         \centering
         \begin{subfigure}[b]{0.30\textwidth}
             \centering
             \scalebox{0.8}{\begin{tikzpicture}[scale=0.5, dot/.style={
        draw, 
        circle, 
        fill, 
        minimum size=1mm, 
        inner sep=0pt
    }]

    \node[] at (1,2.5) {$\bG$};
    
    \node[dot, label={above:$s$}] (n1) at (1,0) {};
    \node[dot] (n2) at (2,0) {};
    \node[dot] (n3) at (3,0) {};
    \node[dot] (n4) at (4,0) {};
    \node[dot, label={above:$c$}] (n5) at (5,0) {};
    \node at (5,-0.7) {$\pi_{\bG}(s,t)$};
    \node[dot] (n6) at (6,0) {};
    \node[dot] (n7) at (7,0) {};
    \node[dot] (n8) at (8,0) {};
    \node[dot, label={above:$t$}] (n9) at (9,0) {};
    
    \foreach \i [evaluate=\i as \nexti using int(\i+1)] in {1,...,8} {
        \draw (n\i) -- (n\nexti);
    }

    \draw[dotted] (5,0) ellipse [x radius=5, y radius=2];
\end{tikzpicture}}
             \caption{}
              \label{fig:both:a}
         \end{subfigure}
         \hfill
         \begin{subfigure}[b]{0.30\textwidth}
             \centering
             \scalebox{0.8}{\begin{tikzpicture}[scale=0.5, dot/.style={
        draw, 
        circle, 
        fill, 
        minimum size=1mm, 
        inner sep=0pt
    }]

    \node[] at (1,2.5) {$\bL$};
    
    \node[dot, label={above:$s$}] (n1) at (1,0) {};
    \node[dot] (n2) at (2,0) {};
    \node[dot] (n3) at (3,0) {};
    \node[dot] (n4) at (4,0) {};
    \node[dot, label={above:$c$}] (n5) at (5,0) {};
    \node[dot] (n6) at (6,0) {};
    \node[dot] (n7) at (7,0) {};
    \node[dot] (n8) at (8,0) {};
    \node[dot, label={above:$t$}] (n9) at (9,0) {};
    
    \foreach \i [evaluate=\i as \nexti using int(\i+1)] in {1,...,4} {
        \draw (n\i) -- (n\nexti);
    }

    \draw[dotted] (5,0) ellipse [x radius=5, y radius=2];
\end{tikzpicture}}
             \caption{}
             \label{fig:both:b}
         \end{subfigure}
         \hfill
         \begin{subfigure}[b]{0.30\textwidth}
             \centering
             \scalebox{0.8}{\begin{tikzpicture}[scale=0.5, dot/.style={
        draw, 
        circle, 
        fill, 
        minimum size=1mm, 
        inner sep=0pt
    }]

    \node at (1,2.5) {$\bR$};
    
    \node[dot, label={above:$s$}] (n1) at (1,0) {};
    \node[dot] (n2) at (2,0) {};
    \node[dot] (n3) at (3,0) {};
    \node[dot] (n4) at (4,0) {};
    \node[dot, label={above:$c$}] (n5) at (5,0) {};
    \node[dot] (n6) at (6,0) {};
    \node[dot] (n7) at (7,0) {};
    \node[dot] (n8) at (8,0) {};
    \node[dot, label={above:$t$}] (n9) at (9,0) {};
    
    \foreach \i [evaluate=\i as \nexti using int(\i+1)] in {5,...,8} {
        \draw (n\i) -- (n\nexti);
    }

    \draw[dotted] (5,0) ellipse [x radius=5, y radius=2];
\end{tikzpicture}}
             \caption{}
             \label{fig:both:c}
         \end{subfigure}
         \\[15pt]
         \begin{subfigure}[b]{0.30\textwidth}
             \centering
             \scalebox{0.8}{\begin{tikzpicture}[scale=0.5, dot/.style={
        draw, 
        circle, 
        fill, 
        minimum size=1mm, 
        inner sep=0pt
    }]

    \node[] at (1,2.5) {$\bG$};
    
    \node[dot, label={above:$s$}] (n1) at (1,0) {};
    \node[dot, label={above:$z$}] (n2) at (2,0) {};
    \node[dot, label={below:$u$}] (n3) at (3,0) {};
    \node[dot, label={below:$v$}] (n4) at (4,0) {};
    \node[dot, label={above:$c$}] (n5) at (5,0) {};
    \node[dot] (n6) at (6,0) {};
    \node[dot, label={below:$u'$}] (n7) at (7,0) {};
    \node[dot, label={below:$v'$}] (n8) at (8,0) {};
    \node[dot, label={above:$t$}] (n9) at (9,0) {};
    
    \foreach \i [evaluate=\i as \nexti using int(\i+1)] in {1,...,8} {
        \draw (n\i) -- (n\nexti);
    }

    \draw[red, thick] (n3) --node[above] {$e$} (n4);
    \draw[red, thick] (n7) --node[above] {$e'$} (n8);
    \draw[green, opacity=0.5, line width=3pt] (n1) -- (n2);
    \draw[green, opacity=0.5, line width=3pt] (n2) to[out=45, in=180] (4,1.2);
    \draw[green, opacity=0.5, line width=3pt, dashed] (4,1.2) -- (5.5,1.2);

    \draw[dotted] (5,0) ellipse [x radius=5, y radius=2];
\end{tikzpicture}}
             \caption{}
             \label{fig:both:d}
         \end{subfigure}
         \hfill
         \begin{subfigure}[b]{0.30\textwidth}
             \centering
             \scalebox{0.8}{\begin{tikzpicture}[scale=0.5, dot/.style={
        draw, 
        circle, 
        fill, 
        minimum size=1mm, 
        inner sep=0pt
    }]

    \node[] at (1,2.5) {$\bG$};
    
    \node[dot, label={above:$s$}] (n1) at (1,0) {};
    \node[dot, label] (n2) at (2,0) {};
    \node[dot, label={below:$u$}] (n3) at (3,0) {};
    \node[dot, label={below:$v$}] (n4) at (4,0) {};
    \node[dot, label={above:$c$}] (n5) at (5,0) {};
    \node[dot] (n6) at (6,0) {};
    \node[dot, label={below:$u'$}] (n7) at (7,0) {};
    \node[dot, label={below:$v'$}] (n8) at (8,0) {};
    \node[dot, label={above:$t$}] (n9) at (9,0) {};
    
    \foreach \i [evaluate=\i as \nexti using int(\i+1)] in {1,...,8} {
        \draw (n\i) -- (n\nexti);
    }

    \draw[red, thick] (n3) --node[above] {$e$} (n4);
    \draw[red, thick] (n7) --node[above] {$e'$} (n8);
    \draw[green, opacity=0.5, line width=3pt] (n8) -- (n9);
    \draw[green, opacity=0.5, line width=3pt] (n8) to[out=45, in=0] (7,1) to[out=180, in=135] (n6);
    \draw[green, opacity=0.5, line width=3pt] (n6) -- (n4);
    \draw[green, opacity=0.5, line width=3pt] (n4) to[out=45, in=0] (3,1) to[out=180, in=135] (n2);
    \draw[green, opacity=0.5, line width=3pt] (n2) -- (n1);

    \draw[dotted] (5,0) ellipse [x radius=5, y radius=2];
\end{tikzpicture}}
             \caption{}
             \label{fig:both:e}
         \end{subfigure}
         \hfill
         \begin{subfigure}[b]{0.30\textwidth}
             \centering
             \scalebox{0.8}{\begin{tikzpicture}[scale=0.5, dot/.style={
        draw, 
        circle, 
        fill, 
        minimum size=1mm, 
        inner sep=0pt
    }]

    \node[] at (1,2.5) {$\bG$};
    
    \node[dot, label={above:$s$}] (n1) at (1,0) {};
    \node[dot, label] (n2) at (2,0) {};
    \node[dot, label={below:$u$}] (n3) at (3,0) {};
    \node[dot, label={below:$v$}] (n4) at (4,0) {};
    \node[dot, label={above:$c$}] (n5) at (5,0) {};
    \node[dot] (n6) at (6,0) {};
    \node[dot, label={below:$u'$}] (n7) at (7,0) {};
    \node[dot, label={below:$v'$}] (n8) at (8,0) {};
    \node[dot, label={above:$t$}] (n9) at (9,0) {};
    
    \foreach \i [evaluate=\i as \nexti using int(\i+1)] in {1,...,8} {
        \draw (n\i) -- (n\nexti);
    }

    \draw[red, thick] (n3) --node[above] {$e$} (n4);
    \draw[red, thick] (n7) --node[above] {$e'$} (n8);
    \draw[green, opacity=0.5, line width=3pt] (n8) -- (n9);
    \draw[green, opacity=0.5, line width=3pt] (n8) to[out=45, in=0] (6,1) to[out=180, in=135] (n4);
    \draw[green, opacity=0.5, line width=3pt] (n6) -- (n4);
    \draw[green, opacity=0.5, line width=3pt] (n6) to[out=45, in=0] (4,1) to[out=180, in=135] (n2);
    \draw[green, opacity=0.5, line width=3pt] (n2) -- (n1);

    \draw[dotted] (5,0) ellipse [x radius=5, y radius=2];
\end{tikzpicture}}
             \caption{}
             \label{fig:both:f}
         \end{subfigure}

        \caption{
           Some visualizations of the introduced notation when computing $d_{\bG \setminus e,e'}(s,t)$. 
         }
         \label{fig:both}
\end{figure*}
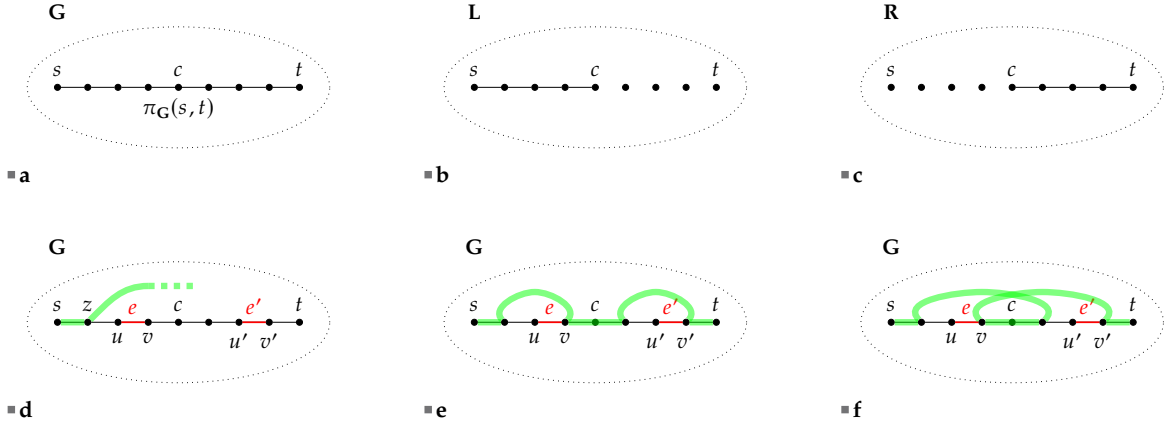

\subparagraph*{Computing $d_{\bG \setminus e,e'}(s,t)$ for all $e \in \pi_{\bG}(s,c)$ and $e' \in \pi_{\bG}(c,t)$.}
This turns out not to be the hard part.
Let $\bL$ and $\bR$ be $\bG$ where we remove all edges of $\pi_{\bG}(c,t)$ and $\pi_{\bG}(s,c)$, respectively (see \cref{fig:both:b} and \cref{fig:both:c}).
Suppose we compute \SSRP in $\bL$ and $\bR$ from $s$ and $t$, respectively.
Order the nodes on $\pi_{\bG}(s,t)$ from the closest to $s$ to the farthest, letting $\preceq$ be the preceding relation.
For edges $e=\{u,v\} \in \pi_{\bG}(s,c)$ and $e'=\{u',v'\} \in \pi_{\bG}(c,t)$ such that $u$ and $u'$ are closer to $s$ than $v$ and $v'$, respectively, we claim that the desired value $d_{\bG \setminus e,e'}(s,t)$ equals:
\begin{align}
    \min \Big\{ 
        \min_{v' \preceq z \preceq t} \Big\{ d_{\bL \setminus e}(s,z) + d_{\bG}(z,t) \Big\} \ , \
        \min_{s \preceq z \preceq u} \Big\{ d_{\bG}(s,z) + d_{\bR \setminus e'}(z,t) \Big\} \ , \
        d_{\bG \setminus e}(s,c) + d_{\bG \setminus e'}(c,t)
    \Big\}. \label{eq:to:1}
\end{align}

It turns out that by using some simple tricks, \cref{eq:to:1} can be computed in $\Oh(n^2)$ overall time for each $e \in \pi_{\bG}(s,c)$ and $e' \in \pi_{\bG}(c,t)$. We defer the discussion to the full proof.

To provide intuition for the correctness, we explain why \cref{eq:to:1} is upper bounded by $d_{\bG \setminus e,e'}(s,t)$. To see this, we perform a case distinction on $\pi_{\bG \setminus e,e'}(s,t)$:
\begin{itemize}
    \item $\pi_{\bG \setminus e,e'}(s,t)$ does not use any vertex $x \in \pi_{\bG}(s,t)$ such that $v \preceq x \preceq c$.
    Then, the last node $z \in \pi_{\bG}(s,t)$ visited by $\pi_{\bG \setminus e,e'}(s,t)$ must satisfy $z \preceq u$ (see \cref{fig:both:d}). In this case, $\pi_{\bG \setminus e,e'}(s,t)$ between $s$ and $z$ equals $\pi_{\bG}(s,z)$, and between $z$ and $t$, where no other edge of $\pi_{\bG}(s,c)$ is visited, it must coincide with $\pi_{\bR \setminus e'}(z,t)$. 
    We conclude that $\pi_{\bG \setminus e,e'}(s,t)$ is captured by the second term of \cref{eq:to:1}.
    
    \item $\pi_{\bG \setminus e,e'}(s,t)$ does not use any vertex $y \in \pi_{\bG}(s,t)$ such that $c \preceq y \preceq u'$.
    Using a symmetric argument to the one above, this case is captured by the first term of \cref{eq:to:1}.
    
    \item $\pi_{\bG \setminus e,e'}(s,t)$ uses a vertex $x \in \pi_{\bG}(s,t)$ such that $v \preceq x \preceq c$ and a vertex $y \in \pi_{\bG}(s,t)$ such that $c \preceq y \preceq u'$ (refer to \cref{fig:both:e} and \cref{fig:both:f} to see the two possible ways $\pi_{\bG \setminus e,e'}(s,t)$ might look in this case). 
    Since no edge on $\pi_{\bG}(s,t)$ between $x$ and $y$ equals $e$ or $e'$, the path $\pi_{\bG \setminus e,e'}(s,t)$ between $x$ and $y$ equals $\pi_{\bG}(x,y)$ and, in particular, passes through $c$. 
    Since $d_{\bG \setminus e}(s,c)$ and $d_{\bG \setminus e'}(c,t)$ are always less than or equal to $d_{\bG \setminus e,e'}(s,c)$ and $d_{\bG \setminus e,e'}(c,t)$, the third term of \cref{eq:to:1} is upper bounded by $d_{\bG \setminus e,e'}(s,t)$.
\end{itemize}
It remains to show why \cref{eq:to:1} is lower bounded by $d_{\bG \setminus e,e'}(s,t)$. It is not difficult to see that the first and second terms are lower bounded by $d_{\bG \setminus e,e'}(s,t)$ because the described paths never use $e$ or $e'$. For the third term, the complete argument for which we omit, we essentially need to argue that if $\pi_{\bG \setminus e}(s,c)$ or $\pi_{\bG \setminus e'}(c,t)$ uses $e'$ or $e$, then there exists a shorter option for $\pi_{\bG \setminus e,e'}(s,t)$ different from the path described by those terms and captured by the other first two terms.

\subparagraph*{Constructing $\bA$ and $\bB$.}

Ideally, we would like to pick $c$ and construct $\bA$ and $\bB$ such that each contains roughly half of the nodes of $\bG$. 
If we pick $c$ to be roughly in the middle of $\pi_{\bG}(s,t)$ (the exact definition of which we leave unspecified for now), a natural candidate for $\bA$ would be the subtree rooted at $c$ within the shortest path tree rooted at $t$. By picking $\bA$ in this manner, we obtain the useful property that for any $e, e' \in \pi_{\bG}(s,c)$, as soon as $\pi_{\bG \setminus e,e'}(s,t)$ visits the first vertex $z \notin \bA$, the path $\pi_{\bG \setminus e,e'}(s,t)$ from $z$ to $t$ coincides with $\pi_{\bG}(z,t)$, as this path contains neither $e$ nor $e'$. By further adding to $\bA$ an edge from $c$ to every node $x \in \bA$ with weight $\min_{z \notin \bA} \{w(x,z) + d_{\bG}(z,t)\} - d_{\bG}(c,t)$, we have $d_{\bA \setminus e,e'}(s,c) + d_{\bG}(c,t) = d_{\bG \setminus e,e'}(s,t)$ because such new edge captures the segment $\pi_{\bG}(z,t)$ and the edge $\{x,z\}$ before $z$ on $\pi_{\bG \setminus e,e'}(s,t)$.

This last construction is nice because it lets us condense $\bG$ into $\bA$ that acts well for the purpose of computing $d_{\bG \setminus e,e'}(s,t)$ for $e, e' \in \pi_{\bG}(s,c)$. However, the construction works against the objective of creating a universal reduction that preserves edge weights: The newly added edge might be of increased edge weight and $\bA$ would not have the same edge set as $\bG$! 
Generally, when condensing $\bG$ into $\bA$, we should always expect to encounter this behavior, regardless of how we choose the smaller subgraph $\bA$.

\subparagraph*{A technical contribution to \SSRP as workaround.}

As a workaround, we establish that the added weights do not, in fact, impact our capability to compute \SSRP on $\bA$. Formally, let $\bH$ be a graph with a positive weight set $W$ and let $s$ denote a vertex. We define \emph{short-cut weights $w_s$} as weights that satisfy $d_{\bH}(s,x) \leq w_s(x)$ for every $x$, and the \emph{augmented graph} as the result of adding edges $\{s,x\}$ with weights $w_s(x)$ to $\bH$. (Clearly, $\bA$ is captured by this construction because $\min_{z \notin \bA} \{w_{\bG}(x,z) + d_{\bG}(z,t)\} \geq d_{\bG}(x,c) + d_{\bG}(c,t)$ for each $x \in \bA$.) We show that \SSRP can be solved in the same time complexity even with such weights.

\begin{restatable*}{theorem}{wssrp} \label{thm:wssrp}
    Let $\bG$ be an undirected (resp. directed) $n$-node graph with a set of positive weights $W$. Let $s \in \bG$ be a node, and let $w_s$ denote shortcut weights w.r.t. $s$ (where $w_s(x)$ is not necessarily in $W$). Furthermore, let $\tilde{\bG}$ be the graph $\bG$ augmented by $w_s$.
    Then, for $\sT_{\SSRP}(n,m,W)$ and $\rho$ as in \cref{thm:both}, we can solve \SSRP in $\tilde{\bG}$ with source $s$ in time $\Oh(\sT_{\SSRP}(n,m, W) + n^2)$ if $\rho < 1$, and $\Oh(\sT_{\SSRP}(n,m, W) \log n + n^2)$ otherwise.
\end{restatable*}

\Cref{thm:wssrp} serves not only as an essential component of our reduction but also as a mean to simplify existing \SSRP results. Notably, these weights have appeared (either implicitly or explicitly) in some prior algorithms \cite{GVW12, GV19, CM20}. Consequently, we consider \cref{thm:wssrp} to be a technical contribution of our work.

\subparagraph*{Getting back to the construction of $\bA$ and $\bB$.}
\Cref{thm:wssrp} is the key to devise our divide-et-impera approach, yet the implementation remains non-trivial. Several technical details must still be addressed:
\begin{itemize}
    \item The recursive scheme must operate on graphs $\bG$ augmented by short-cut weights $w_s$ and $w_t$ (relative to $s$ and $t$ in $\bG$), which are accumulated as the recursion descends.
    \item Defining $\bA$ as previously established is insufficient; we must ensure that the shortest path from $s$ to any $x \in \bA$ is contained within $\bA$ itself. Otherwise, shortcut edges might become shorter than the true distances from $s$. The current definition of $\bA$ does not guarantee this property, but some similar variant of it will.
    \item We must select both $\bA$ and $\bB$ such that they are disjoint and partition the nodes roughly in half.
    \item We must prove we can compute \SSRP with shortcut weights as fast as without shortcut weights not only within $\bA$ (resp. $\bB$) but also within $\bR$ (resp. $\bL$). Indeed, in $\bR$ (resp. $\bL$), upon removing a specific set of edges in $\pi_{\bG}(c,t)$ (resp. $\pi_{\bG}(s,c)$), the short-cut property $d_{\bH}(s,x) \leq w_s(x)$ might not hold anymore; however, the removed edges are sufficiently structured that we can show that this is possible.
\end{itemize}

\subsection*{Case $e' \notin \pi_{\bG}(s,t)$}

While the case where $e' \in \pi_{\bG}(s,t)$ was handled via a divide-et-impera approach, the case $e' \notin \pi_{\bG}(s,t)$ requires extensive structural analysis and several case distinctions. 

The underlying strategy is to characterize the path $\pi_{\bG \setminus e,e'}(s,t)$ as the concatenation of a shortest path $\pi_{\bG}(s,x)$ in $\bG$, a single edge $\{x,y\}$, and a path of the form $\pi_{\bG \setminus e}(y,t)$ or $\pi_{\bG \setminus e'}(y,t)$. While it is challenging to provide a holistic treatment of every case, we limit our discussion to describing the primary tools used to simplify $\pi_{\bG \setminus e,e'}(s,t)$ to the desired form and provide one example of such simplification.

To this end, let us introduce some notation.
Let $T_{\bG}(s)$ denote the shortest path tree rooted at $s$ in $\bG$. For any node $x \in V(\bG)$, let $T_{\bG}(s, x)$ denote the subtree of $T_{\bG}(s)$ rooted at $x$. Lastly, given an edge $e \in T_{\bG}(s)$, we let $T_{\bG}(s,e)$ be the subtree of $T_{\bG}(s)$ below $e$.

The first (well-known) observation (which also highlights one of the primary reasons our reduction is restricted to undirected graphs) that allows us to simplify ways paths are written is the following:

\begin{proposition}\label{prp:two_paths}\label{prp:same_subtree}
    Let $\bG$ be a graph, and let $s,t \in V$ be nodes.
    Then, for any edge $e = \{u,v\} \in \pi_{\bG}(s,t)$ such that $u$ is closer to $s$ than $v$ we have $T_{\bG}(s,v) \cap T_{\bG}(t,u) = \emptyset$, which implies $\pi_{\bG \setminus e}(x,t) = \pi_{\bG}(x,t)$ for any $x \in T_{\bG}(s,v)$.
\end{proposition}
\begin{proof}
    For the sake of contradiction, assume that there is $x \in T_{\bG}(s,v) \cap T_{\bG}(t,u)$.
    From $x \in T_{\bG}(s,v)$, it follows that $e \notin \pi_{\bG} (v,x)$. 
    Moreover, from $x \in T_{\bG}(t,u)$ follows $v \in \pi_{\bG}(t,x)$ and $e \in \pi_{\bG}(t,x)$,
    meaning $\pi_{\bG}(t,x) = \pi_{\bG} (t,v) \circ \pi_{\bG} (v,x)$. But neither $\pi_{\bG} (t,v)$ nor $\pi_{\bG} (v,x)$ uses $e$, a contradiction.
\end{proof}

Defining $E_{\bG}(A,B) \coloneqq E(\bG) \cap (A \times B)$ for $A,B \subseteq V(\bG)$,
the second observation is as follows:

\begin{proposition}\label{prp:path_decomp}
    Let $\bG$ be a graph, and consider two vertices $s,t \in \bG$ and a set of edges $F \subseteq E(T_{\bG}(s))$ such that all $T_{\bG}(s,e)$ for different $e \in F$ are pairwise disjoint. Let $S \coloneqq \bigcup_{e \in F} V(T_{\bG}(s,e))$ be the set of vertices below any edge of $F$.

    Then, provided that $t \in S$, there is $\{x,y\} \in E_{\bG}(V(\bG) \setminus S, S)$ with $\{x,y\} \notin F$ such that we can write
    \[
        \pi_{\bG \setminus F}(s,t) = \pi_{\bG}(s,x) \circ \{x,y\} \circ \pi_{\bG \setminus F}(y,t).
    \]
    Moreover, for such $x,y$ we have that $\pi_{\bG \setminus F}(y,t)$ only contains vertices in $S$.
\end{proposition}
\begin{proof}
    Let $x$ be the last node visited by $\pi_{\bG \setminus F}(s,t)$ that is contained in $V(\bG) \setminus S$, and let $y$ be the node immediately after.
    This means that $\pi_{\bG \setminus F}(s,t) = \pi_{\bG \setminus F}(s,x) \circ \{x,y\} \circ \pi_{\bG \setminus F}(y,t)$, where  $\pi_{\bG \setminus F}(y,t)$ only contains vertices in $S$.
    From $x \in V(\bG) \setminus S$ and from the assumption on $V(\bG) \setminus S$ we get that $\pi_{\bG \setminus F}(s,x) = \pi_{\bG}(s,x)$, which concludes the proof. 
\end{proof}

\subparagraph{A simple showcase.}
As a (not too involved) example of how these observations are applied, we showcase the computations performed to determine $d_{\bG \setminus e,e'}(s,t)$ for all $e,e'$ such that $e \in \pi_{\bG}(s,t)$, $e' \notin T_{\bG}(s)$, and $e' \notin E_{\bG}(T_{\bG}(s,e), V(\bG) \setminus T_{\bG}(s,e))$.

Consider such arbitrary $e,e'$.
Note that since $e' \notin T_{\bG}(s)$, we have $T_{\bG \setminus e'}(s) = T_{\bG}(s)$.
In particular, $T_{\bG \setminus e'}(s,e) = T_{\bG}(s,e)$.
This allows us to use \cref{prp:path_decomp} on $\bG \setminus e', s,t$ and $F = \{e\}$ to get that there is $\{x,y\} \in E(\bG) \setminus e'$ with $x \notin T_{\bG}(s,e)$ and $y \in T_{\bG}(s,e)$ such that:
\begin{equation}
    \pi_{\bG \setminus e,e'} (s,t)
    = \pi_{\bG \setminus e'}(s,x) \circ \{x,y\} \circ \pi_{\bG \setminus e,e'} (y,t). \label{eq:to:2}
\end{equation}
Since $e' \notin T_{\bG}(s)=T_{\bG \setminus e}(s)$, we can simplify $\pi_{\bG \setminus e'}(s,x) = \pi_{\bG}(s,x)$.
Moreover, from \cref{prp:two_paths} follows that $\pi_{\bG \setminus e,e'} (y,t) = \pi_{\bG \setminus e'} (y,t)$. 
Altogether, \cref{eq:to:2} simplifies to
\begin{equation}
    \pi_{\bG \setminus e,e'} (s,t)
    = \pi_{\bG}(s,x) \circ \{x,y\} \circ \pi_{\bG \setminus e'} (y,t). \label{eq:to:3}
\end{equation}
Now that the three components of \cref{eq:to:3} are simpler, we can focus on computing $d_{\bG \setminus e,e'}(s,t)$ as 
\begin{equation}
    d_{\bG \setminus e,e'} (s,t)
    = \min_{x \notin T_{\bG}(s,e), y \in T_{\bG}(s,e)} \Big \{ d_{\bG}(s,x) + w_{\bG}(x,y) + d_{\bG \setminus e'}(y,t) \Big\}. \label{eq:to:4}
\end{equation}

\subparagraph*{Some additional (algorithmic) tools.}
To this end, similarly to many works on replacement paths, we employ what is commonly known as \emph{centroid decomposition}.

\begin{lemma}[Folklore, Separator Lemma]\label{lem:separator}
    Given a tree $\bT$ with $n$ nodes rooted at node $s$, one can find in time a node $c$ (called \emph{centroid}) that separates $\bT$ into $2$ edge disjoint subtrees $\bT_1, \bT_2$ such that $E(\bT_1) \cup E(\bT_2) = E(\bT)$, $V(\bT_1) \cap V(\bT_2) = \{c\}$ and $n/3 \leq \abs{V(\bT_1)}, \abs{V(\bT_2)} \leq 2n/3$. (Refer to \cref{fig:centroid} for a visualization.)
    \lipicsEnd
\end{lemma}

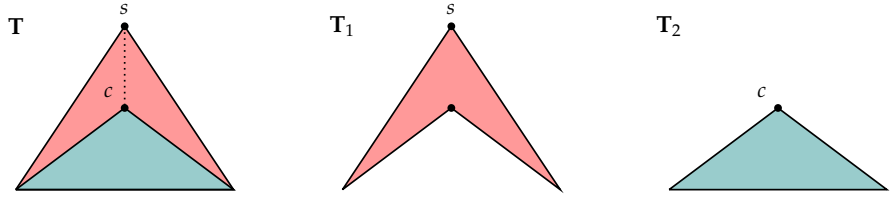
\begin{figure*}[t]
   \centering
   \scalebox{0.8}{\begin{tikzpicture}[scale=1, dot/.style={
        draw, 
        circle, 
        fill, 
        minimum size=1mm, 
        inner sep=0pt
    }]
    \begin{scope}[scale=0.9]
        \fill[red, opacity=0.4] (0,0) -- (2,1.5) -- (4,0) -- (2,3) -- (0,0);
        \draw[thick] (0,0) -- (4,0) -- (2,3) node[dot, label={above:$s$}] {} -- (0,0);
        \fill[teal, opacity=0.4] (0,0) -- (4,0) -- (2,1.5) -- (0,0);
        \draw[thick] (0,0) -- (4,0) -- (2,1.5)  node[dot, label={above left:$c$}] {} -- (0,0);
        \draw[thick, dotted] (2,3) -- (2,1.5);
        \node at (0, 3) {\large $\bT$};
    \end{scope}

    \begin{scope}[scale=0.9, shift={(6,0)}]
        \fill[red, opacity=0.4] (0,0) -- (2,1.5) -- (4,0) -- (2,3) -- (0,0);
        \draw[thick] (0,0) -- (2,1.5) node[dot] {} -- (4,0) -- (2,3) node[dot, label={above:$s$}] {} -- (0,0);
        \node at (0, 3) {\large $\bT_1$};
    \end{scope}

    \begin{scope}[scale=0.9, shift={(12,0)}]
        \fill[teal, opacity=0.4] (0,0) -- (4,0) -- (2,1.5) -- (0,0);
        \draw[thick] (0,0) -- (4,0) -- (2,1.5)  node[dot, label={above left:$c$}] {} -- (0,0);
        \node at (0, 3) {\large $\bT_2$};
    \end{scope}
\end{tikzpicture}}
   \caption{
        This figure shows how $\bT$ is split into $\bT_1$ and $\bT_2$.
    }
   \label{fig:centroid}
\end{figure*}

Another tool we need are common range query data structures.

\begin{lemma}[See \cite{D91,BV94,BFC00,BFC04}, Range Minimum Query]\label{prp:range_query_or}
    Let $q_0, \ldots, q_{n-1}$ be $n$ values.
    Then, we can construct in $\Oh(n)$ time a data structure that answers queries of the type "return $\min_{i \in \fragmentco{a}{b}} q_{x_i}$" in $\Oh(1)$ time, for any $a,b \in \fragment{0}{n}$ with $a < b$. \lipicsEnd
\end{lemma}

\begin{corollary}[Range Minimum Query on Trees]\label{prp:range_query}
    Suppose for each vertex $x$ of a $n$-node tree we are given a value $q_x$. Then, there is a data structure that can be constructed in time $\Oh(n)$ and answers in time $\Oh(1)$ the following queries: return $\min_{x \in X} q_x$, where $X$ is the union and/or intersection of constantly many subtrees of $\bT$.
\end{corollary}
\begin{proof}
    We order the nodes of $\bT$ according to a pre-order traversal, obtaining the sequence $x_0, \ldots, x_{n-1}$. The nodes of any subtree in $\bT$ correspond precisely to a set $\{x_i \mid i \in \fragmentco{a}{b}\}$ for some interval $\fragmentco{a}{b}$. This implies that any set formed by the union and intersection of a constant number of subtrees in $\bT$ also corresponds to the union of a constant number of such intervals. Thus, any query for the data structure of \cref{prp:range_query} can be converted into a constant number of queries to the aforementioned data structure. The final result is then the minimum of the values returned by these queries.
\end{proof}

\subparagraph*{Getting back to \cref{eq:to:4}.}

We want to compute $d_{\bG \setminus e,e'}(s,t)$ for all $e,e'$ as described above via \cref{eq:to:4}. In this presentation, we will use one fact that we prove later in the paper: there exists an (efficiently computable) set $F \subseteq E(\bG)$ of size $|F| = \Oh(n)$ that collects all $e'$ (across also different $e$) falling into this case.

We solve the problem recursively. In the beginning, $\bT \coloneqq T_{\bG}(s)$.
Our goal is to compute the following expression for every $e \in E(\bT) \cap \pi_{\bG}(s,t)$ and $e' \in F$:
\begin{align}
    \min_{x \in \bT \setminus T_{\bG}(s,e), y \in T_{\bG}(s,e)} \Big \{ d_{\bG}(s,x) + w_{\bG}(x,y) + d_{\bG \setminus e'}(y,t) \Big\}
    \label{eq:to:5}
\end{align}
In each level of the recursion, we split $\bT$ into $\bT_1$ and $\bT_2$ using \cref{lem:separator}, and find the corresponding centroid $c$. 
Let $n_{\bT},n_{\bT_1},n_{\bT_2}$ be number of nodes in the corresponding tree.

For $e \in E(\bT_1) \cap \pi_{\bG}(s,t)$ and $e' \in F$, we rewrite \cref{eq:to:5} as
\begin{align*}
    \min \Big\{ \min_{\substack{x \in \bT_1 \setminus T_{\bG}(s,e)\\ y \in T_{\bG}(s,e)}} \big\{ \ d_{\bG}(s,x) + w_{\bG}(x,y) + d_{\bG \setminus e'} (y,t) \ \big\},
    \min_{\substack{x \in \bT_2 \\y \in T_{\bG}(s,e)}} \big\{ \ d_{\bG}(s,x) + w_{\bG}(x,y) + d_{\bG \setminus e'} (y,t) \ \big\} \Big\},
\end{align*}
where in the second term we use that if $x \in \bT_2$ and $e \in \bT_1$, then $x \notin T_{\bG}(s,e)$.
(notice that this is only true when $e$ is not on the path from the root of $\bT$ to $c$, however if $e$ is on such path then the second term is non-existent, and we can ignore it.)

For the first term, we can recurse on $\bT_1$. To compute the second term, we do the following:
\begin{itemize}
    \item for each node $y \in \bG$ we compute $h_{y} \coloneqq \min_{x \in \bT_2} \{d_{\bG}(s,x) + w_{\bG}(x,y) \}$ in time $\Oh(n_\bT \cdot n)$;
    \item we setup for each $e' \in F$ a data structure from \cref{prp:range_query} where $y$ has the associated value $h_{y} + d_{\bG \setminus e'} (y,t)$ in time $\Oh(|F| \cdot n) = \Oh(n^2)$.
    \item for each $e \in \bT_1$ and $e' \in F$ we query $\max_{y \in T_{\bG}(s,e)} \{ h_{y} + d_{\bG \setminus e'} (y,t) \}$ in time $\Oh(|F| \cdot n_{\bT}) = \Oh(n n_{\bT})$.
\end{itemize} 

Similarly, when $e \in E(\bT_2) \cap \pi_{\bG}(s,t)$ we rewrite \cref{eq:to:4} as
\begin{align*}
    \min \Big\{ \min_{\substack{x \in \bT_2 \setminus T_{\bG}(s,e)\\ y \in T_{\bG}(s,e)}} \big\{ \ d_{\bG}(s,x) + w_{\bG}(x,y) + d_{\bG \setminus e'} (y,t) \ \big\},
    \min_{\substack{x \in \bT_1 \\y \in T_{\bG}(s,e)}} \big\{ \ d_{\bG}(s,x) + w_{\bG}(x,y) + d_{\bG \setminus e'} (y,t) \ \big\} \Big\},
\end{align*}
where we use again $x \in \bT_1$ and $e \in \bT_2$ implies $x \notin T_{\bG}(s,e)$. As before, we can compute the first term by recursing in $\bT_2$ and the first by performing similar computations to above.

For the running time notice that first we need to compute \SSRP from $t$ in $\bG$. Then, when the recursion arrives to a certain tree $\bT$, we use time $\Oh(n_{\bT}\cdot n)$ and we recurse on $\bT_1$ and $\bT_2$ that satisfy $(n_{\bT_1}-1)+(n_{\bT_2}-1)=(n_{\bT}-1)$. Since with every level of recursion $n_\bT$ shrinks by a constant fraction, we get that the overall running time is $\Oh(n^2 \log n)$.

\subparagraph*{Organization of the paper.} After introducing our notation more properly in \cref{sec:prelims}, we prove \cref{thm:single} in \cref{sec:single} (the case where $e,e' \in \pi_{\bG}(s,t)$) and \cref{thm:both} in \cref{sec:both} (the case where $e' \notin \pi_{\bG}(s,t)$). Finally, in \cref{sec:lb}, we prove the lower bounds of \cref{thm:lb}.

\section{Preliminaries}
\label{sec:prelims}

\subparagraph*{Set notation.}
For integers \(i, j \in \mathbb{Z}\), we write \(\fragment{i}{j}\) to represent the set \(\{i, \dots, j\}\), and \(\fragmentco{i}{j}\) to denote the set \(\{i, \dots, j - 1\}\).
We define \(\fragmentoc{i}{j}\) and \(\fragmentoo{i}{j}\) similarly.

\subparagraph*{Graph notation.}

Throughout this paper, unless explicitly stated otherwise, we consider only undirected graphs (specifically, the only section where we consider directed graphs is \cref{subsec:both:1}). For a graph $\bG$, we denote its vertex set by $V(\bG)$ and its edge set by $E(\bG)$. Let $w_{\bG}(\cdot, \cdot)$ be the weight function of $\bG$, which we assume to be positive. For convenience, we define $w_{\bG}(x,y) = +\infty$ whenever $\{x,y\} \notin E(\bG)$. To maintain brevity, we may use $\bG$ to refer to its vertex or edge sets when the meaning is clear from the context (we generally use the variables $e, f$ to denote edges and $x, y, z, u, v, a, b$ to denote nodes).
We also define $E_{\bG}(A,B) = E(\bG) \cap (A \times B)$ for $A,B \subseteq V(\bG)$.

Given a graph $\bG$ and a set $F \subseteq E(\bG)$, we denote with $\bG \setminus F$, the graph obtained by removing all edges in $F$ from $\bG$. Whenever $F = \{e\}$ and $F = \{e,e'\}$, we write $\bG \setminus e$ and $\bG \setminus e,e'$ instead of $\bG \setminus \{e\}$ and $\bG \setminus \{e,e'\}$.

\subparagraph*{On canonical shortest paths.}

Given a graph $\bG$ and nodes $x, y \in V(\bG)$, let $d_{\bG}(x, y)$ denote the distance between $x$ and $y$ in $\bG$. For every pair $x, y \in V(\bG)$, we choose and \textbf{fix a canonical shortest path $\pi_{\bG}(x, y)$}. We denote the set of vertices and edges belonging to this path by $V(\pi_{\bG}(x, y))$ and $E(\pi_{\bG}(x, y))$, respectively; as before, we may omit the $V(\cdot)$ and $E(\cdot)$ notation when the context is clear. Finally, we use $\circ$ to denote the concatenation of paths and edges.

When fixing shortest paths, we require that $\pi_{\bG}(\cdot,\cdot)$ is closed under subpaths; specifically, for any $x, y, z \in V(\bG)$ such that $z \in \pi_{\bG}(x, y)$, we have $\pi_{\bG}(x, y) = \pi_{\bG}(x, z) \circ \pi_{\bG}(z, y)$. Furthermore, we require that canonical shortest paths are consistent with edge removals: for any $e \in E(\bG)$ and $x, y \in V(\bG)$, if $e \notin \pi_{\bG}(x, y)$, then $\pi_{\bG \setminus e}(x, y) = \pi_{\bG}(x, y)$. While these canonical paths are primarily fixed for theoretical purposes, we must ensure they can be computed explicitly for at least two specific vertices, typically denoted $s$ and $t$. (This, combined with \cref{prp:two_paths} will provide runtime guarantees based on the splitting of trees.) 
We can achieve this easily in time $\Oh(m + n \log n)$ by adding a small perturbation to all edges and then using Dijkstra's algorithm from $s$ and $t$ (after this computation perturbations are dropped).

\subparagraph*{Shortest path trees.} Let $T_{\bG}(s)$ denote the shortest path tree of $s$ in $\bG$, constructed to be consistent with the previously fixed canonical paths. For any node $x \in V(\bG)$, we let $T_{\bG}(s, x)$ denote the subtree of $T_{\bG}(s)$ rooted at $x$.
Lastly, given an edge $e \in T_{\bG}(s)$, we let $T_{\bG}(s,e)$ be the subtree of $T_{\bG}(s)$ below $e$.
\section{Proof for $e' \in \pi_{\bG}(s,t)$}
\label{sec:both}

In this section, we prove \cref{thm:both}.

\both*

The proof is split up into \cref{subsec:both:1} and \cref{subsec:both:2}.
First, in \cref{subsec:both:1}, we introduce the notion of \emph{short-cut} weights (\cref{def:shortcut}), and we prove that \SSRP with short-cut weights takes the same time as without (\cref{thm:wssrp} and \cref{lem:wssrp2}).
The results in \cref{subsec:both:1} are general enough to hold also for directed graphs, so we phrase them in a way to reflect this.
Then, in \cref{subsec:both:1}, we show how \SSRP with short-cut weights can help us to compute $d_{\bG \setminus e,e'}(s,t)$ for all $e \in \pi_{\bG}(s,c)$ and $e' \in \pi_{\bG}(c,t)$ for an arbitrary node $c \in \pi_{\bG}(s,t)$. This allows us to finally prove \cref{thm:both} via a divide-et-impera approach.

\subsection{\SSRP with Short-cut Weights}
\label{subsec:both:1}

We start by stating the definition of \emph{short-cut weights}.

\begin{definition}[Short-cut Weights]\label{def:shortcut}
    Let $\bG$ be a directed or undirected and positively weighted graph, and let $s \in \bG$ be a node.
    We say that a function $w_s$ \emph{represents weights w.r.t. $s$} if for every node $x \in \bG$ we have $w_s(x) \geq 0$. Furthermore, we say that it \emph{represents short-cut weights w.r.t. $s$} if for every node $x \in \bG$ we have $w_s(x) \geq d_{\bG}(s,x)$.
\end{definition}

\begin{definition}[Augmented Graph]\label{def:augmented}
    Let $\bG$ be a directed or undirected graph with positive weights, and let $s \in \bG$ be a node.  Given a function $w_s$ representing weights w.r.t. $s$, the graph $\tilde{\bG}$ (referred to as $\bG$ \emph{augmented by} $w_s$) is defined as the graph $\bG$ extended by the set of edges $(s, x)$ with weights $w_s(x)$ for all $x \in \bG$ (these edges in $\tilde{\bG}$ are referred to as \emph{short-cut edges}).
\end{definition}

We note that whenever $w_s$ represents short-cut weights, then the distances from $s$ in $\bG$ and $\tilde{\bG}$ are preserved. Thus, we may assume that $\pi_{\bG}(s,x) = \pi_{\tilde{\bG}}(s,x)$ for all $x \in \bG$ and that $T_{\bG}(s) = \bT_{\tilde{\bG}}(s)$. Note, although the short-cut edges in $\tilde{\bG}$ are not used is shortest paths from $s$ to a node $x$, they might be still used in the shortest paths from $s$ in $\tilde{\bG} \setminus e$ for any edge $e$. We start by proving that \SSRP can be solved in the same time complexity even with such weights.

\wssrp

\begin{proof}
    We use a recursive approach.
    Whenever $\bG$ is small enough, i.e., $n,m=\Oh(1)$, we can use any trivial algorithm in time $\Oh(1)$.
    Otherwise, when we are not in the base case, we proceed as follows:
    \begin{enumerate}[(i)]
        \item We begin by using the Separator Lemma (\cref{lem:separator}) on $\bT \coloneqq T_{\bG}(s)$ and $s$ to find the centroid $c$ that divides $\bT$ into $\bT_1$ and $\bT_2$. Let $\bG_1$ and $\bG_2$ be the subgraphs of $\bG$ induced by the nodes in $\bT_1$ and $\bT_2$. 
        \item We proceed to construct the short-cut weights $w_1$ w.r.t. $s$ in $\bG_1$ and $w_2$ w.r.t. $c$ in $\bG_2$, respectively. 
        To this end, for every $x \in \bG_1$ and $z \in \bG_2$, we set
        \[
            w_1(x) \coloneqq \min \Big\{ w_s(x), d_{\tilde{\bG} \setminus \pi_{\bG}(s,c)}(s,x) \Big\}
            \quad \text{and} \quad
            w_2(z) \coloneqq w_s(z) - d_{\bG}(s,c).
        \]
        \item This allows us to  recurse on $\bG_1$ and $\bG_2$ with sources $s,c$ and short-cut weights $w_1$ and $w_2$, respectively,
        to compute \SSRP from $s$ and $c$ in $\tilde{\bG}_1$ and $\tilde{\bG}_2$.
        \item We also compute \SSRP from $s$ and $c$ in $\bG$ using the algorithm for \SSRP without short-cut weights that is guaranteed to exist from the problem statement.
        \item Next, define the relations $\prec$ and $\preceq$ for vertices on the path $\pi_{\bG}(s,c)$ to denote strict preceding and preceding, respectively.
        This allows us to compute for every $e = \{u,v\}\in \pi_{\bG}(s,c)$ (such that $u$ comes before $v$ in $\pi_{\bG}(s,c)$) and $x \in \bG$, the value
        \[
            h_{e,x} \coloneqq \min_{v \preceq y \preceq c} \{d_{\tilde{\bG} \setminus \pi_{\bG}(s,c)}(s,y) - d_{\bG}(s,y)\}\ + d_{\bG}(s,c) + d_{\bG \setminus e} (c,x).
        \]
        That is, $h_{e,x}$ is the length of the shortest path of the form $\pi_{\tilde{\bG} \setminus \pi_{\bG}(s,c)}(s,y) \circ \pi_{\bG}(y,c) \circ \pi_{\bG \setminus e}(c,x)$ for some $v \preceq y \preceq c$ of weight $d_{\tilde{\bG} \setminus \pi_{\bG}(s,c)}(s,y) + (d_{\bG}(s,c) - d_{\bG}(s,y)) + d_{\bG \setminus e}(c,x) = (d_{\tilde{\bG} \setminus \pi_{\bG}(s,c)}(s,y) - d_{\bG}(s,y)) + d_{\bG}(s,c) + d_{\bG \setminus e}(c,x)$.
        \item Finally, 
        we compute $d_{\tilde{\bG} \setminus e}(s,x)$ for every $e$ and $x$ using a case distinction (depicted in \cref{fig:ssrp_case}):
        \begin{subequations}
        \label{eq:wssrp}
        \begin{numcases}{d_{\tilde{\bG} \setminus e}(s,x) =}
                d_{\bG}(s,x) & \text{$x \in \bG_2 \setminus \{c\}$ and $e \notin \bG_2 \cup \pi_{\bG}(s,c)$} \label{eq:wssrp:a} \\
                \min \Big\{ \ d_{\bG \setminus e}(s,x) \ , \ d_{\bG}(s,c) + d_{\tilde{\bG}_2 \setminus e}(c,x) \Big\}& \text{$x \in \bG_2 \setminus \{c\}$ and $e \in \bG_2$}, \label{eq:wssrp:b} \\
                \min \Big\{ \ d_{\bG \setminus e}(s,x) \ , \ d_{\tilde{\bG} \setminus \pi_{\bG}(s,c)} (s,x) \ , \ h_{e,x} \ \Big\} & \text{$x \in \bG_2 \setminus \{c\}$ and $e \in \pi_{\bG}(s,c)$}, \label{eq:wssrp:c} \\
                \min \Big\{ \ d_{\bG \setminus e}(s,x) \ , d_{\tilde{\bG}_1 \setminus e}(s,x) \ \Big\} & \text{$x \in \bG_1$ and $e \notin \bG_2 \cup \pi_{\bG}(s,c)$}, \label{eq:wssrp:d} \\
                d_{\bG}(s,x) & \text{$x \in \bG_1$ and $e \in \bG_2$}, \label{eq:wssrp:e} \\
                \min \Big\{ \ d_{\bG \setminus e}(s,x) \ , d_{\tilde{\bG}_1}(s,x) \ , \ h_{e,x} \ \Big\} & \text{$x \in \bG_1$ and $e \in \pi_{\bG}(s,c)$}. \label{eq:wssrp:f}
        \end{numcases}
        \end{subequations}
    \end{enumerate}
    
    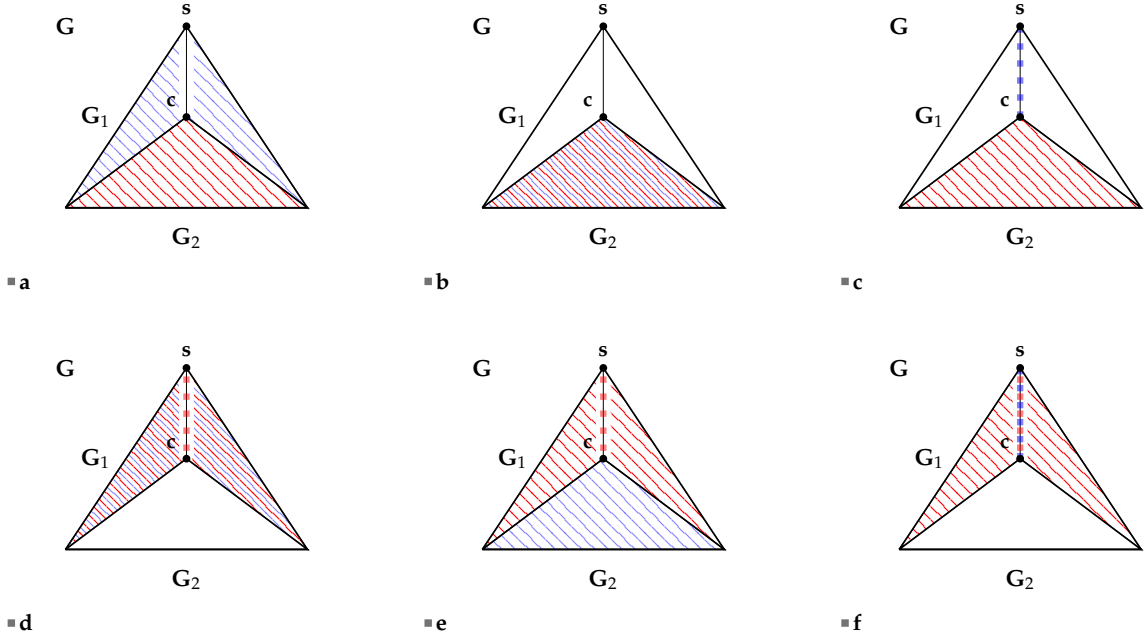
\begin{figure*}[!t]
         \centering
         \begin{subfigure}[b]{0.30\textwidth}
             \centering
             \scalebox{0.8}{\begin{tikzpicture}[scale=1, dot/.style={
        draw, 
        circle, 
        fill, 
        minimum size=1mm, 
        inner sep=0pt
    }]

    \draw[opacity=0.7, 
    pattern={Lines[angle=-45, distance={5pt/sqrt(2)}]}, 
    pattern color=blue] (0,0) -- (2,1.5) -- (4,0) -- (2,3) -- (0,0);
    \draw[line width=7pt, white] (2,3.0) -- (2,1.5);
    \draw (2,3.0) -- (2,1.5);
     
    \draw[thick] (0,0) -- (4,0) -- (2,3) node[dot, label={above:$\bf s$}] {} -- (0,0);
   
    \draw[opacity=0.7, 
    preaction={%
        pattern={Lines[angle=-45, distance={5pt/sqrt(2)}, xshift=2.5pt]}, 
        pattern color=red}
    ] (0,0) -- (4,0) -- (2,1.5) -- (0,0);
    \draw[thick] (0,0) -- (4,0) -- (2,1.5)  node[dot, label={above left:$\bf c$}] {} -- (0,0);

    \node at (0, 3) {\large $\bG$};
    \node at (0.5, 1.5) {\large $\bG_1$};
    \node at (2, -0.5) {\large $\bG_2$};
\end{tikzpicture}}
             \caption{}
              \label{fig:both:a}
         \end{subfigure}
         \hfill
         \begin{subfigure}[b]{0.30\textwidth}
             \centering
             \scalebox{0.8}{\begin{tikzpicture}[scale=1, dot/.style={
        draw, 
        circle, 
        fill, 
        minimum size=1mm, 
        inner sep=0pt
    }]

    \draw[line width=7pt, white] (2,3.0) -- (2,1.5);
    \draw (2,3.0) -- (2,1.5);
     
    \draw[thick] (0,0) -- (4,0) -- (2,3) node[dot, label={above:$\bf s$}] {} -- (0,0);
   
    \draw[opacity=0.7, 
    pattern={Lines[angle=-45, distance={5pt/sqrt(2)}]}, 
    pattern color=blue,
    preaction={%
        pattern={Lines[angle=-45, distance={5pt/sqrt(2)}, xshift=2.5pt]}, 
        pattern color=red}
    ] (0,0) -- (4,0) -- (2,1.5) -- (0,0);
    \draw[thick] (0,0) -- (4,0) -- (2,1.5)  node[dot, label={above left:$\bf c$}] {} -- (0,0);

    \node at (0, 3) {\large $\bG$};
    \node at (0.5, 1.5) {\large $\bG_1$};
    \node at (2, -0.5) {\large $\bG_2$};
\end{tikzpicture}}
             \caption{}
             \label{fig:both:b}
         \end{subfigure}
         \hfill
         \begin{subfigure}[b]{0.30\textwidth}
             \centering
             \scalebox{0.8}{\begin{tikzpicture}[scale=1, dot/.style={
        draw, 
        circle, 
        fill, 
        minimum size=1mm, 
        inner sep=0pt
    }]

    \draw[line width=7pt, white] (2,3.0) -- (2,1.5);
    \draw (2,3.0) -- (2,1.5);
    \draw[line width=3pt, blue, opacity=0.5, dash pattern= on 3pt off 5pt] (2,3.0) -- (2,1.5);
     
    \draw[thick] (0,0) -- (4,0) -- (2,3) node[dot, label={above:$\bf s$}] {} -- (0,0);
   
    \draw[opacity=0.7, 
    preaction={%
        pattern={Lines[angle=-45, distance={5pt/sqrt(2)}, xshift=2.5pt]}, 
        pattern color=red}
    ] (0,0) -- (4,0) -- (2,1.5) -- (0,0);
    \draw[thick] (0,0) -- (4,0) -- (2,1.5)  node[dot, label={above left:$\bf c$}] {} -- (0,0);

    \node at (0, 3) {\large $\bG$};
    \node at (0.5, 1.5) {\large $\bG_1$};
    \node at (2, -0.5) {\large $\bG_2$};
\end{tikzpicture}}
             \caption{}
             \label{fig:both:c}
         \end{subfigure}
         \\[15pt]
         \begin{subfigure}[b]{0.30\textwidth}
             \centering
             \scalebox{0.8}{\begin{tikzpicture}[scale=1, dot/.style={
        draw, 
        circle, 
        fill, 
        minimum size=1mm, 
        inner sep=0pt
    }]

    \draw[opacity=0.7, 
    pattern={Lines[angle=-45, distance={5pt/sqrt(2)},]}, 
    pattern color=blue,
    preaction={%
        pattern={Lines[angle=-45, distance={5pt/sqrt(2)}, xshift=2.5pt]}, 
        pattern color=red}
    ] (0,0) -- (2,1.5) -- (4,0) -- (2,3) -- (0,0);
    \draw[line width=7pt, white] (2,3.0) -- (2,1.5);
    \draw (2,3.0) -- (2,1.5);
    \draw[line width=3pt, red, opacity=0.5, dash pattern= on 3pt off 5pt, dash phase=4pt] (2,3.0) -- (2,1.5);
     
    \draw[thick] (0,0) -- (4,0) -- (2,3) node[dot, label={above:$\bf s$}] {} -- (0,0);
   
    \draw[thick] (0,0) -- (4,0) -- (2,1.5)  node[dot, label={above left:$\bf c$}] {} -- (0,0);

    \node at (0, 3) {\large $\bG$};
    \node at (0.5, 1.5) {\large $\bG_1$};
    \node at (2, -0.5) {\large $\bG_2$};
\end{tikzpicture}}
             \caption{}
             \label{fig:both:d}
         \end{subfigure}
         \hfill
         \begin{subfigure}[b]{0.30\textwidth}
             \centering
             \scalebox{0.8}{\begin{tikzpicture}[scale=1, dot/.style={
        draw, 
        circle, 
        fill, 
        minimum size=1mm, 
        inner sep=0pt
    }]

    \draw[opacity=0.7, 
    preaction={%
        pattern={Lines[angle=-45, distance={5pt/sqrt(2)}, xshift=2.5pt]}, 
        pattern color=red}
    ] (0,0) -- (2,1.5) -- (4,0) -- (2,3) -- (0,0);
    \draw[line width=7pt, white] (2,3.0) -- (2,1.5);
    \draw (2,3.0) -- (2,1.5);
    \draw[line width=3pt, red, opacity=0.5, dash pattern= on 3pt off 5pt, dash phase=4pt] (2,3.0) -- (2,1.5);
     
    \draw[thick] (0,0) -- (4,0) -- (2,3) node[dot, label={above:$\bf s$}] {} -- (0,0);
   
    \draw[opacity=0.7, 
    pattern={Lines[angle=-45, distance={5pt/sqrt(2)}]}, 
    pattern color=blue,
    ] (0,0) -- (4,0) -- (2,1.5) -- (0,0);
    \draw[thick] (0,0) -- (4,0) -- (2,1.5)  node[dot, label={above left:$\bf c$}] {} -- (0,0);

    \node at (0, 3) {\large $\bG$};
    \node at (0.5, 1.5) {\large $\bG_1$};
    \node at (2, -0.5) {\large $\bG_2$};
\end{tikzpicture}}
             \caption{}
             \label{fig:both:e}
         \end{subfigure}
         \hfill
         \begin{subfigure}[b]{0.30\textwidth}
             \centering
             \scalebox{0.8}{\begin{tikzpicture}[scale=1, dot/.style={
        draw, 
        circle, 
        fill, 
        minimum size=1mm, 
        inner sep=0pt
    }]

    \draw[opacity=0.7, 
    preaction={%
        pattern={Lines[angle=-45, distance={5pt/sqrt(2)}, xshift=2.5pt]}, 
        pattern color=red}
    ] (0,0) -- (2,1.5) -- (4,0) -- (2,3) -- (0,0);
    \draw[line width=7pt, white] (2,3.0) -- (2,1.5);
    \draw (2,3.0) -- (2,1.5);
    \draw[line width=3pt, blue, opacity=0.5, dash pattern= on 3pt off 5pt] (2,3.0) -- (2,1.5);
    \draw[line width=3pt, red, opacity=0.5, dash pattern= on 3pt off 5pt, dash phase=4pt] (2,3.0) -- (2,1.5);
     
    \draw[thick] (0,0) -- (4,0) -- (2,3) node[dot, label={above:$\bf s$}] {} -- (0,0);
   
    \draw[thick] (0,0) -- (4,0) -- (2,1.5)  node[dot, label={above left:$\bf c$}] {} -- (0,0);

    \node at (0, 3) {\large $\bG$};
    \node at (0.5, 1.5) {\large $\bG_1$};
    \node at (2, -0.5) {\large $\bG_2$};
\end{tikzpicture}}
             \caption{}
             \label{fig:both:f}
         \end{subfigure}

        \caption{
           Visual representation of the cases in \cref{eq:wssrp}. The red dashed line corresponds to the domain of $x$, and the blue dashed line corresponds to the domain of $e$.
         }
         \label{fig:ssrp_case}
    \end{figure*}
    
    \subparagraph*{Correctness:} 
    We begin by proving that $w_1$ and $w_2$ are indeed valid short-cut weights. To this end, note that for every $x \in \bG_1$ we have $d_{\tilde{\bG} \setminus \pi_{\bG}(s,c)} (s,x) \geq d_{\tilde{\bG} }(s,x) \geq d_{\bG}(s,x) = d_{\bG_1}(s,x)$. This, combined with $w_s(x) \geq d_{\bG}(s,x) = d_{\bG_1}(s,x)$ yields the validity of $w_1$. On the other hand, for all $y \in \bG_2$ we have $w_2(y) \geq d_{\bG}(s,x) = d_{\bG}(s,c) + d_{\bG_2}(c,x)$. Rearranging, we see that $w_2$ is valid as well.

    We proceed to prove that the values $d_{\tilde{\bG} \setminus e}(s,x)$ we compute are indeed correct. Clearly, \cref{eq:wssrp:a} and \cref{eq:wssrp:e} are correct, as for them $e \notin \pi_{\bG}(s,x)$ holds. Henceforth, we focus on the other cases.

    Observe that for the other cases in \cref{eq:wssrp}, we never underestimate the value of $d_{\tilde{\bG} \setminus e}(s,x)$. This is because, in each case, the terms used correspond to the weights of valid paths from $s$ to $x$ in $\tilde{\bG} \setminus e$.
    So, it remains to prove that we never overestimate. We observe that we can ignore the case where $\pi_{\tilde{\bG} \setminus e}(s,x)$ does not use any short-cut edge; in this case, $\pi_{\tilde{\bG} \setminus e}(s,x) = \pi_{\bG \setminus e}(s,x)$, and the distance $d_{\bG \setminus e}(s,x)$ is already included as a term in every case of \cref{eq:wssrp}.
    Henceforth, we assume that $\pi_{\tilde{\bG} \setminus e}(s,x)$ uses a short-cut edge.
    Note that such edge must be unique and it must be the first edge on $\pi_{\tilde{\bG} \setminus e}(s,x)$.
    
    \begin{claim}\label{claim:wssrp:1}
        \cref{eq:wssrp:c} and \cref{eq:wssrp:f} are correct.
    \end{claim}
    \begin{claimproof}
        Let $e \in \pi_{\bG}(s,c)$.
        First, we sort out the case when $\pi_{\tilde{\bG} \setminus e}(s,x)$ uses no edge on $\pi_{\bG}(s,c)$. If this is the case, then we have $\pi_{\tilde{\bG} \setminus e}(s,x) = \pi_{\tilde{\bG} \setminus \pi_{\bG}(s,c)}(s,x)$. In \cref{eq:wssrp:c} such terms appears directly in the minimum, and in \cref{eq:wssrp:f} it appears in $w_1$.

        Henceforth, assume $\pi_{\tilde{\bG} \setminus e}(s,x)$ uses at least one edge on $\pi_{\bG}(s,c)$ and let $y$ be the first visited node on $\pi_{\bG}(s,c)$ other than $s$. Note that we necessarily have that $v \preceq y \preceq c$,
        as $\pi_{\bG}(s,c)$ in-between $s$ and $y$ would coincide with $\pi_{\bG}(s,y)$ and not use a short-cut edge.
        Moreover, we observe that $\pi_{\tilde{\bG} \setminus e}(s,x)$ in-between $s$ and $y$ equals to $\pi_{\tilde{\bG} \setminus \pi_{\bG}(s,c)}(s,y)$.
        We further perform a case distinction. 
        \begin{itemize}
            \item The path $\pi_{\tilde{\bG} \setminus e}(s,x)$ contains $z \in \bG_2$ that comes after $y$ on $\pi_{\tilde{\bG} \setminus e}(s,x)$. Let $z$ be the last such node on $\pi_{\tilde{\bG} \setminus e}(s,x)$. Then, we can assume that $\pi_{\tilde{\bG} \setminus e}(s,x)$ in-between $y$ and $z$ coincides with $\pi_{\bG}(y,z) = \pi_{\bG}(y,c) \circ \pi_{\bG}(c, z)$. In particular, $\pi_{\bG}(y,z)$ must pass through $c$.
            All together
            \[
                \pi_{\bG \setminus e}(s,z) = \pi_{\tilde{\bG} \setminus \pi_{\bG}(s,c)}(s,y) \circ \pi_{\bG}(y,c) \circ \pi_{\bG \setminus e}(c, z),
            \]
            which implies 
            \[
                \pi_{\bG \setminus e}(s,x) = \pi_{\bG \setminus e}(s, c) \circ \pi_{\bG \setminus e}(z, x) = \pi_{\tilde{\bG} \setminus \pi_{\bG}(s,c)}(s,y) \circ \pi_{\bG}(y,c) \circ \pi_{\bG \setminus e}(c, x), 
            \]
            because $\pi_{\bG \setminus e}(c, z) \circ \pi_{\bG \setminus e}(z, x) = \pi_{\bG \setminus e}(c, x)$.
            So, this case is captured by the term $h_{e,x}$. 
            Note that for \cref{eq:wssrp:c} this already finishes the proof of \cref{claim:wssrp:1}, as we always have $z = x$.
            Thus, in the remaining part of the proof we may only consider \cref{eq:wssrp:f} and assume $x \in \bG_1$.
            
            \item The path $\pi_{\tilde{\bG} \setminus e}(s,x)$ contains nodes $z \in \bG_2$ only before $y$. Let $z \in \bG_2$ be the last such node in $\pi_{\tilde{\bG} \setminus e}(s,x)$ and $\bG_1 \in z'$ the node immediately after which need to exist because $x \in \bG_1$.
            We derive that $\pi_{\tilde{\bG} \setminus e}(s,x)$ between $s$ and $z'$ contains no edge from $\pi_{\bG}(s,c)$.
            Thus, we have $\pi_{\tilde{\bG} \setminus e}(s,x) = \pi_{\tilde{\bG} \setminus \pi_{\bG}(s,c)}(s,z') \circ \pi_{\bG_1 \setminus e}(s,z')$. Thus, we get $d_{\tilde{\bG} \setminus e}(s,x) = d_{\tilde{\bG}_1 \setminus e}(s,x)$ because of how we defined $w_1$.
        
            \item The path $\pi_{\tilde{\bG} \setminus e}(s,x)$ contains no $z \in \bG_2$. In this case, we have $d_{\tilde{\bG} \setminus e}(s,x) = d_{\tilde{\bG}_1 \setminus e}(s,x)$. \claimqedhere
        \end{itemize}
    \end{claimproof}

    \begin{claim}\label{claim:wssrp:1}
        \cref{eq:wssrp:b} and \cref{eq:wssrp:d} are correct.
    \end{claim}
    \begin{claimproof}
         First, we prove \cref{eq:wssrp:b}.
         Let $e \in \bG_2$ and $x \in \bG_2 \setminus \{c\}$.
         We note that $\pi_{\tilde{\bG} \setminus e}(s,x)$ uses no node in $\bG_1$ other than $s$.
         Otherwise, if $\pi_{\tilde{\bG} \setminus e}(s,x)$ contains $z \in \bG_1$ such that $z \neq s$, then $\pi_{\tilde{\bG} \setminus e}(s,x)$ in-between $s$ and $z$ coincides with $\pi_{\bG}(s,z)$ and no short-cut edge is used, a contradiction.
         Consequently, $\pi_{\tilde{\bG} \setminus e}(s,y) = \{s, y\} \circ \pi_{\bG_2 \setminus e}(y,x)$,
         for some $y \in \bG_2$ such that $\{s,y\}$ is a short-cut edge.
         Since $\tilde{\bG_2}$ contains a short-cut edge $\{s,y\}$ of weight $w_s(y) - d_{\bG \setminus e}(s,c)$,
         we conclude that $d_{\bG}(s,c) + d_{\tilde{\bG}_2 \setminus e}(c,x)$ gives $d_{\tilde{\bG} \setminus e}(s,x)$.

        The proof for \cref{eq:wssrp:d} is almost identical.  Let $x \in \bG_1$ and $e \notin \bG_2 \cup \pi_{\bG}(s,x)$.
        Using a similar argument, we note that $\pi_{\tilde{\bG} \setminus e}(s,x)$ uses no node in $\bG_2$.
        Consequently, $\pi_{\tilde{\bG} \setminus e}(s,y) = \{s, y\} \circ \pi_{\bG_1 \setminus e}(y,x)$,
        for some $y \in \bG_1$ such that $\{s,y\}$ is a short-cut edge.
        Since $\tilde{\bG}_1$ contains a short-cut edge $\{s,y\}$ of weight $w_s(y)$,
        we conclude that $d_{\tilde{\bG}_1 \setminus e}(s,x)$ gives $d_{\tilde{\bG} \setminus e}(s,x)$.
    \end{claimproof}
    
    \subparagraph*{Runtime:} 
    Apart from the recursive calls and the calls to \SSRP without short-cut weights, we execute the following subroutines. 
    In time $\Oh(n^2)$, we use Dijkstra's algorithm in $\bG$ and $\tilde{\bG} \setminus \pi_{\bG}(s,c)$ from $s$ to obtain access to distances of the form $d_{\bG}(s, \cdot)$ and $d_{\tilde{\bG} \setminus \pi_{\bG}(s,c)}(s, \cdot)$.
    Moreover, to compute $h_{e,x}$, we first 
    compute $\min_{v \preceq y \preceq c} \{d_{\bG \setminus \pi_{\bG}(s,c)}(s,y) - d_{\bG}(s,y)\}$ for each $v \in \pi_{\bG}(s,c)$ naively in time $\Oh(n)$. This allows us to compute each of the $\Oh(n^2)$ terms in $\Oh(1)$ time.

    Let $n_1,m_1$ and $n_2,m_2$ be the number of nodes and edges in $\bG_1$ and $\bG_2$, respectively. Note that $m_1 + m_2 \leq m$ and $n_1, n_2 \leq n + 1$ and $n_1,n_2 \leq 2n/3$.
    Together, this yields a recursion of the type:
    \begin{align*}
        \sT(n,m,W) \leq \sT(n_1,m_1,W) + \sT(n_2,m_1,W) + \Oh(n^2) + 2\sT_{\SSRP}(n,m,W).
    \end{align*}
    Using a standard induction proof and the properties of $\sT_{\SSRP}$, the runtime follows.
\end{proof}

In \cref{subsec:both:2}, when applying \cref{thm:wssrp}, we do not use it directly on a graph $\bG$ with valid shortcut weights $w_s$ with respect to $s$. Instead, we apply it to arbitrary weights $s$, but the instance is specific enough that it can be reduced to the former case. This is captured in the following proposition.

\begin{proposition} \label{lem:wssrp2}
    Let $\bG$ be an undirected (resp. directed) $n$-node graph with a set of positive weights $W$. Let $s \in \bG$ be a node, and let $w_s$ denote weights w.r.t. $s$ (where $w_s(x)$ is not necessarily in $W$ and $w_s$ are not necessarily short-cut weights). Furthermore, let $\tilde{\bG}$ be the graph $\bG$ augmented by $w_s$.

    Suppose we are given $X \subseteq V(\bG)$ and $F \subseteq E(\bG)$ such that for all $(x,e) \in X \times F$ such that 
    $\pi_{\tilde{\bG}}(s,x)$ uses a short-cut edge we have $e \notin \pi_{\tilde{\bG}}(s,x)$.
    Then, letting $g(n,m,W),\rho$ be as in \cref{thm:both}, we can compute $d_{\tilde{\bG} \setminus e}(s,x)$ for all $(x,e) \in X \times F$ in time $\Oh(g(n,m, W) + n^2)$ if $\rho < 1$, and $\Oh(g(n,m, W) \log n + n^2)$ otherwise.
\end{proposition}

\begin{proof} 
    Let $\tilde{\bT} \coloneqq \bT_{\tilde{\bG}}(s)$
    and let $H \subseteq V(\tilde{\bT})=V(\bG)$ be the set of nodes reachable in $\tilde{\bT}$ from $s$ using no short-cut edges (in particular, $s \in H$).
    We define $\bH$ as the subgraph of $\bG$ induced by the node set $H$.
    Moreover, we define the function $w_s'$ as $w_s'(y) \coloneqq \min \{w_s(y), \min_{z \notin H} d_{\tilde{\bG}}(s,z) + w_{\bG}(z,y)\}$ for $y \in H$.
    Let $\tilde{\bH}$ be $\bH$ augmented by $w_s'$.
    Via \cref{thm:wssrp}, we get the \SSRP values in $\tilde{\bH}$ from $s$.
    
    Finally, for every $x \in X$ and $e \in F$ we compute the desired values as:
    \begin{subnumcases}{d_{\tilde{\bG} \setminus e}(s,x) =}
        d_{\tilde{\bG}}(s,x) & \text{$x \notin H$} \label{eq:wssrp2:a} \\
         d_{\tilde{\bG}}(s,x) & \text{$x \in H$ and $e \notin \bH$}, \label{eq:wssrp2:b} \\
         d_{\tilde{\bH} \setminus e}(s,x) & \text{$x \in H$ and $e \in \bH$}.
         \label{eq:wssrp2:c}
    \end{subnumcases}
    \subparagraph*{Correctness:}
    First, we need to prove that we can indeed apply \cref{thm:wssrp}.
    To this end, we need to show that $w_s'$ are valid short-cut weights in $\bH$ w.r.t. $s$. To this end, notice that for every $y \in H$, we have $w_s(y) \geq d_{\bG}(s,y) = d_{\bH}(s,y)$ and $d_{\tilde{\bG}}(s,z) + w(z,y) \geq d_{\tilde{\bG}}(s,y) = d_{\bG}(s,y) = d_{\bH}(s,y)$ for every $z \notin H$. We conclude that $w_s'(y) \geq d_{\bH}(s,y)$ for every $y \in H$.
    
    We proceed to prove that $d_{\tilde{\bG} \setminus e}(s,x)$ is computed correctly.
    By our assumption, we get that \cref{eq:wssrp2:a} is correct.
    Moreover, notice that if $x \in H$ and $e \notin \bH$, then $e \notin \pi_{\tilde{\bG}}(s,x)$ so also \cref{eq:wssrp2:b} is correct.
    Lastly, for \cref{eq:wssrp2:c}, we perform a case distinction:
    \begin{itemize}
        \item If $\pi_{\tilde{\bG} \setminus e}(s,x)$ contains a node $z \notin H$, then let $z$ be the last such node and let $y$ be the node immediately after. From $e \in \bH$ follows that $\pi_{\tilde{\bG} \setminus e}(s,z) = \pi_{\tilde{\bG}}(s,z)$, meaning that $\pi_{\tilde{\bG} \setminus e}(s,x) = \pi_{\tilde{\bG}}(s,z) \circ (x,y) \circ \pi_{\bH \setminus e}(y,x)$. By our construction of $w'_s(x)$ follows that we compute \cref{eq:wssrp2:c} correctly.
        \item  If $\pi_{\tilde{\bG} \setminus e}(s,x)$ contains no node $z \notin H$, then we get the correctness of \cref{eq:wssrp2:c} directly.
    \end{itemize}    

    \subparagraph*{Runtime:} Apart from using \cref{thm:wssrp}, we need to use Dijkstra's algorithm once from $s$ in $\tilde{\bG}$.
\end{proof}

\subsection{Reducing the case $e \in \pi_{\bG}(s,t)$ to \SSRP with Short-cut Weights}
\label{subsec:both:2}

The proof of the divide-et-impera algorithm for the case $e \in \pi_{\bG}(s,t)$ is 
split up in \cref{lem:both_divide}, \cref{lem:both_conquer}, and \cref{thm:both}. 
For the divide step, we first select an edge $f = \{a,b\} \in \pi_{\bG}(s,t)$ such that $a$ is closer than $b$ to $s$ in $\pi_{\bG}(s,t)$. Then, we construct two subgraphs that preserve $\pi_{\bG \setminus e,e'}(s,t)$ for $e,e' \in \pi_{\bG}(s,a)$ and $e,e' \in \pi_{\bG}(b,t)$, respectively.

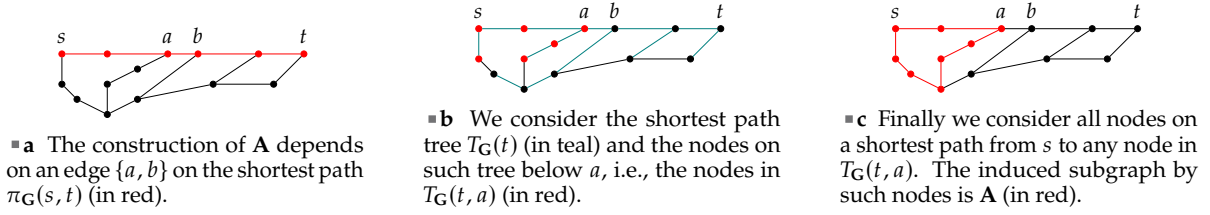
\begin{figure*}[!t]
         \centering
         \begin{subfigure}[b]{0.30\textwidth}
             \centering
             \scalebox{0.8}{\begin{tikzpicture}[scale=0.5, dot/.style={
        draw, 
        circle, 
        fill, 
        minimum size=1mm, 
        inner sep=0pt
    }]

    \node[dot, label={above:$s$}, red] (n1) at (1,0) {};
    \node[dot, label={above:$a$}, red] (n5) at (4.5,0) {};
    \node[dot, label={above:$b$}, red] (n6) at (5.5,0) {};
    \node[dot, red] (n11) at (7.5,0) {};
    \node[dot] (n12) at (8,-1) {};
    \node[dot] (n13) at (6,-1) {};
    \node[dot, label={above:$t$}, red] (n9) at (9,0) {};
    \node[dot, red] (n2) at (2.5,0) {};
    \node[dot] (n3) at (2.5,-1) {};
    \node[dot] (n4) at (2.5,-2) {};
    \node[dot] (n7) at (1,-1) {};
    \node[dot] (n7b) at (1.5,-1.5) {};
    \node[dot] (n8) at (3.5,-0.5) {};
    \node[dot] (n10) at (3.5,-1.5) {};
    
    \draw (n1) -- (n7) -- (n7b) -- (n4);
    \draw (n4) -- (n3);
    \draw (n3) -- (n8) -- (n5);
    \draw (n4) -- (n10) -- (n6);
    \draw (n9) -- (n12);
    \draw (n11) -- (n13) -- (n12);
    \draw (n13) -- (n10);
    
    \draw[red] (n1) -- (n2) -- (n5) -- (n6) -- (n11) -- (n9);
\end{tikzpicture}}
             \caption{The construction of $\bA$ depends on an edge $\{a,b\}$ on the shortest path $\pi_{\bG}(s,t)$ (in red).}
              \label{fig:aexample:a}
         \end{subfigure}
         \hfill
         \begin{subfigure}[b]{0.30\textwidth}
             \centering
             \scalebox{0.8}{\begin{tikzpicture}[scale=0.5, dot/.style={
        draw, 
        circle, 
        fill, 
        minimum size=1mm, 
        inner sep=0pt
    }]

    \node[dot, label={above:$s$}, red] (n1) at (1,0) {};
    \node[dot, label={above:$a$}, red] (n5) at (4.5,0) {};
    \node[dot, label={above:$b$}] (n6) at (5.5,0) {};
    \node[dot] (n11) at (7.5,0) {};
    \node[dot] (n12) at (8,-1) {};
    \node[dot] (n13) at (6,-1) {};
    \node[dot, label={above:$t$}] (n9) at (9,0) {};
    \node[dot, red] (n2) at (2.5,0) {};
    \node[dot, red] (n3) at (2.5,-1) {};
    \node[dot] (n4) at (2.5,-2) {};
    \node[dot, red] (n7) at (1,-1) {};
    \node[dot] (n7b) at (1.5,-1.5) {};
    \node[dot, red] (n8) at (3.5,-0.5) {};
    \node[dot] (n10) at (3.5,-1.5) {};
    
    \draw[teal] (n1) -- (n7);
    \draw (n7) -- (n7b);
    \draw[teal] (n7b) -- (n4);
    \draw (n4) -- (n3);
    \draw[teal] (n3) -- (n8) -- (n5);
    \draw[teal] (n4) -- (n10) -- (n6);
    \draw[teal] (n9) -- (n12);
    \draw[teal] (n11) -- (n13);
    \draw (n13) -- (n12);
    \draw (n13) -- (n10);
    
    \draw[teal] (n1) -- (n2) -- (n5) -- (n6) -- (n11) -- (n9);
\end{tikzpicture}}
             \caption{We consider the shortest path tree $T_{\bG}(t)$ (in teal) and the nodes on such tree below $a$, i.e., the nodes in $T_{\bG}(t,a)$ (in red).}
             \label{fig:aexample:b}
         \end{subfigure}
         \hfill
         \begin{subfigure}[b]{0.30\textwidth}
             \centering
             \scalebox{0.8}{\begin{tikzpicture}[scale=0.5, dot/.style={
        draw, 
        circle, 
        fill, 
        minimum size=1mm, 
        inner sep=0pt
    }]

    \node[dot, label={above:$s$}, red] (n1) at (1,0) {};
    \node[dot, label={above:$a$}, red] (n5) at (4.5,0) {};
    \node[dot, label={above:$b$}] (n6) at (5.5,0) {};
    \node[dot] (n11) at (7.5,0) {};
    \node[dot] (n12) at (8,-1) {};
    \node[dot] (n13) at (6,-1) {};
    \node[dot, label={above:$t$}] (n9) at (9,0) {};
    \node[dot, red] (n2) at (2.5,0) {};
    \node[dot, red] (n3) at (2.5,-1) {};
    \node[dot, red] (n4) at (2.5,-2) {};
    \node[dot, red] (n7) at (1,-1) {};
    \node[dot, red] (n7b) at (1.5,-1.5) {};
    \node[dot, red] (n8) at (3.5,-0.5) {};
    \node[dot] (n10) at (3.5,-1.5) {};
    
    \draw[red] (n1) -- (n7);
    \draw[red] (n7) -- (n7b);
    \draw[red] (n7b) -- (n4);
    \draw[red] (n4) -- (n3);
    \draw[red] (n3) -- (n8) -- (n5);
    \draw (n4) -- (n10) -- (n6);
    \draw (n9) -- (n12);
    \draw (n11) -- (n13);
    \draw (n13) -- (n12);
    \draw (n13) -- (n10);
    
    \draw[red] (n1) -- (n2) -- (n5);
    \draw (n5) -- (n6) -- (n11) -- (n9);
\end{tikzpicture}}
             \caption{Finally we consider all nodes on a shortest path from $s$ to any node in $T_{\bG}(t,a)$. The induced subgraph by such nodes is $\bA$ (in red).}
             \label{fig:aexample:c}
         \end{subfigure}
        
        \caption{
           The figure illustrates how the graph $\bA$ is constructed. 
           The construction of $\bB$ is symmetric w.r.t. $\{a,b\}$.
         }
         \label{fig:aexample}
\end{figure*}

\begin{lemma}\label{lem:both_divide}
    Let $\bG$ be an undirected $n$-node graph with a set of positive weights $W$. Let $s,t \in \bG$ be nodes, and let $w_s, w_t$ denote weights with respect to $s$ and $t$, respectively (where $w_s(x)$ and $w_t(x)$ are not necessarily in $W$). Furthermore, let $\tilde{\bG}$ be the graph $\bG$ augmented by $w_s$ and $w_t$.

     For a given edge $f = \{a,b\} \in \pi_{\bG}(s,t)$ such that $a$ is closer than $b$ to $s$ in $\pi_{\bG}(s,t)$, let $\bA$ and $\bB$ be the subgraphs of $\bG$ induced by the set of vertices
    \begin{align*}
         \{x \mid x \in \pi_{\bG}(s,y), y \in T_{\bG}(t,a)\} \text{ and } \{x \mid x \in \pi_{\bG}(t,y), y \in T_{\bG}(s,b)\}, \text{ respectively.}
    \end{align*}
    (\Cref{fig:aexample} illustrates an example of the construction of $\bA$.) 
    
    Then, all of the following hold:
    \begin{enumerate}[(a)]
        \item $V(\bA) \cap V(\bB) = \emptyset$; \label{lem:both_divide:a}
        \item In time $\Oh(n^2)$ we can find short-cut weights $w_{\bA,s}, w_{\bA,a}$ with respect to $s$ and $a$ in $\bA$, such that $d_{\tilde{\bG} \setminus e,e'}(s,t) = d_{\tilde{\bA} \setminus e,e'}(s,a) + d_{\bG}(a,t)$ for all $e,e' \in \pi_{\bG}(s,a)$, where $\tilde{\bA}$ is $\bA$ augmented by $w_{\bA,s}, w_{\bA,a}$; \label{lem:both_divide:b}
        \item In time $\Oh(n^2)$ we can find short-cut weights $w_{\bB,b}, w_{\bB,t}$ with respect to $b$ and $t$ in $\bB$, such that $d_{\tilde{\bG} \setminus e,e'}(s,t) = d_{\bG}(s,b) + d_{\tilde{\bB} \setminus e,e'}(b,t)$ for all $e,e' \in \pi_{\bG}(b,t)$, where $\tilde{\bB}$ is $\bB$ augmented by $w_{\bB,b}, w_{\bB,t}$. \label{lem:both_divide:c}
    \end{enumerate}
\end{lemma}

\begin{proof}
    We first show \ref{lem:both_divide:a}.
    
    \begin{claim}
        There is no $z$ such that $z \in \pi_{\bG}(s,y_A)$ and $z \in \pi_{\bG}(t,y_B)$ for some $y_A \in T_{\bG}(t,a)$ and $y_B \in T_{\bG}(s,b)$.
    \end{claim}
    \begin{claimproof}
        Assume for sake of contradiction that there is such $z$.
        First, we note that, by \cref{prp:same_subtree}, we have $T_{\bG}(t,a) \cap T_{\bG}(s,b) = \emptyset$.

        Next, we rule out the case $z \in T_{\bG}(t,a)$.
        Indeed, if $z \in T_{\bG}(t,a)$, then $b \in \pi_{\bG}(t,z)$
        and so $\pi_{\bG}(t,y_B) = \pi_{\bG}(t,b) \circ \pi_{\bG}(b,y_{B})$.
        However, all nodes of $\pi_{\bG}(t,b)$ and $\pi_{\bG}(b,y_{B})$ are fully contained in $T_{\bG}(s,b)$ because $t,y_{B} \in T_{\bG}(s,b)$. Since $T_{\bG}(t,a) \cap T_{\bG}(s,b) = \emptyset$, we have a contradiction with $z \in T_{\bG}(t,a)$.
        By symmetry, $z \notin T_{\bG}(s,b)$.
    
        The case that remains to handle is when $z \notin T_{\bG}(t,a)$ and $z \notin T_{\bG}(s,b)$. 
        Using $z \in \pi_{\bG}(s,y_A)$ and $z \in \pi_{\bG}(t,y_B)$ and the triangle inequality, we get:
        \begin{align}
            \label{eq:both_divide:1}
            &d_{\bG}(s,z) + d_{\bG}(z,y_A) = d_{\bG}(s,y_A) \leq d_{\bG}(s,a) + d_{\bG}(a,y_A) \quad \text{ and}\\
            &d_{\bG}(t,z) + d_{\bG}(z,y_B) = d_{\bG}(s,y_B) \leq d_{\bG}(t,b) + d_{\bG}(b,y_B).
            \label{eq:both_divide:2}
        \end{align}
        But because $y_A \in T_{\bG}(t,a)$ and $y_B \in T_{\bG}(s,b)$ we can also use the triangle inequality to derive
        \begin{align}
            \label{eq:both_divide:3}
            &d_{\bG}(s,a) + w_{\bG}(a,b) + d_{\bG}(b,y_B) = d_{\bG}(s,y_B) \leq d_{\bG}(s,z) + d_{\bG}(z,y_B) \quad \text{ and}\\
            &d_{\bG}(t,b) + w_{\bG}(a,b) + d_{\bG}(a,y_A) = d_{\bG}(t,y_A) \leq d_{\bG}(t,z) + d_{\bG}(z,y_A).
            \label{eq:both_divide:4}
        \end{align}
        Summing up \cref{eq:both_divide:1}, \cref{eq:both_divide:2} we get:
        \begin{align}
            &d_{\bG}(s,z) + d_{\bG}(z,y_A) + d_{\bG}(t,z) + d_{\bG}(z,y_B) \leq d_{\bG}(s,a) + d_{\bG}(a,y_A) + d_{\bG}(t,b) + d_{\bG}(b,y_B). \label{eq:both_divide:5}
        \end{align}
        On the other hand summing up \cref{eq:both_divide:3}, \cref{eq:both_divide:4}, we get:
        \begin{align}
            \MoveEqLeft d_{\bG}(s,a) + d_{\bG}(b,y_B) + d_{\bG}(t,b) + d_{\bG}(a,y_A) + 2w_{\bG}(a,b) \nonumber \\
            &\leq d_{\bG}(s,z) + d_{\bG}(z,y_B) + d_{\bG}(t,z) + d_{\bG}(z,y_A). \label{eq:both_divide:6}
        \end{align}
        We note that \cref{eq:both_divide:5} and \cref{eq:both_divide:6} are almost identical up to a factor $2w_{\bG}(a,b)$ and up to the direction of the inequality.
        This means that $w_{\bG}(a,b) = 0$,
        a contradiction with our assumption that weights are positive.
    \end{claimproof}
    
    We proceed to show \ref{lem:both_divide:b}. (By symmetry, \ref{lem:both_divide:c} follows as well.)
    For $x \in \bA$, we set $w_{\bA,s}(x)\coloneqq  w_s(x)$ and
    \[
        w_{\bA,a}(x) \coloneqq
        \begin{cases}
             \min \{ w_t(x), \min_{y \in \bG \setminus  T_{\bG}(t,a)} d_{\bG}(t,y) + w_{\bG}(y,x) \} - d_{\bG}(a,t) & \text{if $x \in T_{\bG}(t,a)$,}\\
             +\infty & \text{otherwise.}
        \end{cases}
    \] 
    We prove that $w_{\bA,s}, w_{\bA,a}$ are valid short-cut weights.
    For $w_{\bA,s}$, notice that $\bA$ contains all nodes of a connected subgraph of $T_{\bG}(s)$ that includes the root $s$ as well. This means that $d_{\bA}(s,x) = d_{\bG}(s,x) \leq w_s(x) = w_{\bA,s}(x)$ for $x \in \bA$. On the other hand, for $w_{\bA,a}(x)$ we only need to consider the case $x \in T_{\bG}(t,a)$. 
    Here, we use that $d_{\bG}(t,y) + w_{\bG}(y,x) \geq d_{\bG}(t,x)$ for any $x \in T_{\bG}(t,a)$ and $y$. Combining this with $w_t(x) \geq d_{\bG}(t,x)$, yields $\min \{ w_t(x), \min_{y \in \bG \setminus  T_{\bG}(t,a)} d_{\bG}(t,y) + w_{\bG}(y,x) \} \geq d_{\bG}(t,x) = d_{\bG}(t,a) + d_{\bG}(a,x)$. By rearranging this last inequality and noticing that $d_{\bG}(a,x) = d_{\bA}(a,x)$ we get the validity of $w_{\bA,a}(x)$.

    Next, we prove that $d_{\tilde{\bG} \setminus e,e'}(s,t) = d_{\tilde{\bA} \setminus e,e'}(s,a) + d_{\bG}(a,t)$ 
    for all $e,e' \in \pi_{\bG}(s,a)$. To this end, fix arbitrary $e,e' \in \pi_{\bG}(s,a)$. 
    It is evident that $d_{\tilde{\bG} \setminus e,e'}(s,t) \leq d_{\tilde{\bA} \setminus e,e'}(s,a) + d_{\bG}(a,t)$ 
    because any path $\pi_{\tilde{\bA} \setminus e,e'}(s,a)$ can always be extended to a path 
    in $\tilde{\bG}$ that avoids $e,e'$. This is achieved either by possibly expanding the 
    shortcut edge from $w_{\bA,a}(x)$ into a full path (if $\pi_{\tilde{\bA} \setminus e,e'}(s,a)$ 
    ends in a short-cut edge with weight that comes from the minimization over $y$) or by appending $\pi_{\bG}(a,t)$ otherwise.
  
    For the inequality in the opposite direction, we consider the first node $y$ visited by $\pi_{\tilde{\bG} \setminus e,e'}(s,t)$ in $V(\bG) \setminus V(T_{\bG}(t,a))$ and let $x \in T_{\bG}(t,a)$ be the node immediately before. Since $e,e' \notin \pi_{\bG}(y,t) = \pi_{\tilde{\bG}}(y,t)$, we have $\pi_{\tilde{\bG} \setminus e,e'}(s,t) = \pi_{\bA \setminus e,e'}(s,x) \circ \{x,y\} \circ \pi_{\tilde{\bG}}(y,t)$, where we use $V(T_{\bG}(t,a)) \subseteq V(\bA)$.
    Now, note that both cases when $\pi_{\tilde{\bG}}(s,t)$ ends with a short-cut weight (so $t=y$) and when it does not (so $\pi_{\tilde{\bG}}(y,t) = \pi_{\bG}(y,t)$) are captured by $w_{\bA,a}$.
\end{proof}

In \cref{lem:both_conquer} we give the proof for the impera step.
Here, we consider a node $c \in \pi_{\bG}(s,t)$ and then we compute $d_{\tilde{\bG} \setminus e,e'}(s,t)$ for each $e \in \pi_{\bG}(s,c)$ and $e' \in \pi_{\bG}(c,t)$. For this, we need \cref{lem:wssrp2}.

\begin{lemma}\label{lem:both_conquer}
    Let $\bG$ be an undirected $n$-node graph with a set of positive weights $W$. Let $s,t \in \bG$ be nodes, and let $w_s, w_t$ denote weights with respect to $s$ and $t$, respectively (where $w_s(x)$ and $w_t(x)$ are not necessarily in $W$). Furthermore, let $\tilde{\bG}$ be the graph $\bG$ augmented by $w_s$ and $w_t$.

    Let $c \in \pi_{\bG}(s,t)$ be an arbitrary node. Then, for $g(n,m,W)$ as defined in \cref{thm:both}, in time $\Oh(\sT_{\SSRP}(n,m,W) + n^2)$
     we can compute all values $d_{\tilde{\bG} \setminus e,e'}(s,t)$ such that $e \in \pi_{\bG}(s,c)$ and $e' \in \pi_{\bG}(c,t)$.
\end{lemma}

\begin{proof}
    We define $\bL \coloneqq \bG \setminus E(\pi_{\bG}(c,t)), \bR \coloneqq \bG \setminus E(\pi_{\bG}(s,c))$ and we let $\tilde{\bL}$ and $\tilde{\bR}$ be $\bL$ augmented by $w_s$ and $\bR$ augmented by $w_t$, respectively.
    Next, via \cref{lem:wssrp2} on $\tilde{\bL}$ and $\tilde{\bR}$, we compute $d_{\tilde{\bL} \setminus e}(s,x)$ for all $e \in \pi_{\bG}(s,c)$ and $x \in \pi_{\bG}(c,t)$ and $d_{\tilde{\bR} \setminus e}(t,x)$ for all $e \in \pi_{\bG}(c,t)$ and $x \in \pi_{\bG}(s,c)$. (Here, note that $w_s$ and $w_t$ are not necessarily short-cut weights w.r.t. $s$ and $t$ in $\bL$ and $\bR$; this is why we need to apply \cref{lem:wssrp2} instead of \cref{thm:wssrp}.)
    Via \cref{thm:wssrp} we also compute \SSRP in $\tilde{\bG}$ from $s$ and $t$, respectively.

    We define the relations $\prec$ and $\preceq$ for vertices on the path $\pi_{\bG}(s,t)$ to denote strict preceding and preceding, respectively.
    Let $e =\{u,v\} \in \pi_{\bG}(s,c)$ and  $e' = \{u',v'\} \in \pi_{\bG}(c,t)$ be such that $u \prec v$ and $u' \prec v'$.
    To get the desired value $d_{\tilde{\bG} \setminus e,e'}(s,t)$, we compute
    \begin{align}
        \min \Big\{ 
            \min_{v' \preceq x \preceq t} \Big\{ d_{\tilde{\bL} \setminus e}(s,x) + d_{\bG}(x,t) \Big\} \ , \
            \min_{s \preceq x \preceq u} \Big\{ d_{\bG}(s,x) + d_{\tilde{\bR} \setminus e'}(x,t) \Big\} \ , \
            d_{\tilde{\bG} \setminus e}(s,c) + d_{\tilde{\bG} \setminus e'}(c,t)
        \Big\}. \label{eq:both_rec:1}
    \end{align}

    \subparagraph*{Correctness:}
    We start by showing that $\bL$ and $\bR$ satisfy the conditions for 
    \cref{lem:wssrp2}. We show this only for $\bL$. By symmetry, the same holds for $\bR$. 

    \begin{claim}
        Let $x \in \pi_{\bG}(c,t)$ and $e \in \pi_{\bG}(s,c)$. 
        If $\pi_{\tilde{\bL}}(s,x)$ uses a short-cut edge, then $e \notin \pi_{\tilde{\bL}}(s,x)$.
    \end{claim}
    \begin{claimproof}
        For sake of contradiction, assume $\pi_{\tilde{\bL}}(s,x)$ uses a short-cut edge $\{s,z\}$ and also contains a node $y \in \pi_{\bG}(s,c)$. (Here, we use the fact that if any short-cut edge is used, then it must be the first one).
        By our guarantee on $w_s$, we get: 
        \[
            d_{\bG}(s,y) 
            \leq d_{\bG}(s,z) + d_{\bG}(z,y)
            \leq w_{s}(z) + d_{\bG}(z,y) \leq w_{s}(z) +  d_{\tilde{\bL}}(z,y)
            \leq d_{\tilde{\bL}}(s,y),
        \]
        meaning that $\pi_{\tilde{\bL}}(s,y)$ in-between $s$ and $y$ coincides with $\pi_{\bG}(s,y)$ and that $\pi_{\tilde{\bL}}(s,y)$ does actually not use $\{s,z\}$, a contradiction. Since $\pi_{\tilde{\bL}}(s,x)$ does not contain any $y \in \pi_{\bG}(s,c)$, we get $e \notin \pi_{\tilde{\bL}}(s,x)$.
    \end{claimproof}

    We proceed to prove that \cref{eq:both_rec:1} is correct.

    \begin{claim}
        For every $e \in \pi_{\bG}(s,c)$ and $e' \in \pi_{\bG}(c,t)$, we never underestimate $d_{\tilde{\bG} \setminus e,e'}(s,t)$ in \cref{eq:both_rec:1}.
    \end{claim}
    \begin{claimproof}
        For the first term of \cref{eq:both_rec:1}, we observe that neither $\pi_{\tilde{\bL} \setminus e}($ nor $\pi_{\bG}(x,t)$ use $e,e'$. Thus, $d_{\tilde{\bL} \setminus e}(s,x) + d_{\bG}(x,t)$ describes a $s$-$t$ path in $\tilde{\bG}$ that avoids $e,e'$, so this term is never less than $d_{\tilde{\bG} \setminus e,e'}(s,t)$.
        Similar holds for the second term of \cref{eq:both_rec:1}.

        For the second third, we first consider the case where $\pi_{\tilde{\bG} \setminus e}(s,c)$ uses $e'$. 
        We observe that $\pi_{\tilde{\bG} \setminus e}(s,c)$ need to go first through $v'$ then $u'$,
        as otherwise from $u'$ to $c$ the path would coincide with $\pi_{\bG}(c,u')$ and $e'$ would not be used.
        Let $z$ be the first node visited by $\pi_{\tilde{\bG} \setminus e}(s,c)$ such that $v' \preceq z \preceq t$. 
        Clearly, $d_{\tilde{\bG} \setminus e}(s,z) = d_{\tilde{\bL} \setminus e}(s,z)$.
        Since the path $\pi_{\tilde{\bL} \setminus e}(s,z) \circ \pi_{\bG}(z,t)$ avoids both $e,e'$, we get the desired lower bound from
        \begin{align*}
              d_{\tilde{\bG} \setminus e,e'}(s,t) &\leq d_{\tilde{\bL} \setminus e}(s,z) + d_{\bG}(z,t)
              = d_{\tilde{\bG} \setminus e}(s,z) + d_{\bG}(z,t)\\
              &\leq d_{\tilde{\bG} \setminus e}(s,z) + d_{\bG}(z,c) + d_{\bG}(c,t)
              \leq d_{\tilde{\bG} \setminus e}(s,c) + d_{\tilde{\bG} \setminus e}(c,t).
        \end{align*}
        Similarly, we can sort out the case when $\pi_{\tilde{\bG} \setminus e}(t,c)$ uses $e$.
        In the remaining case $\pi_{\tilde{\bG} \setminus e}(s,c) \circ \pi_{\tilde{\bR} \setminus e'}(c,t)$ avoids both $e,e'$,
        so we do not underestimate $d_{\tilde{\bG} \setminus e,e'}(s,t)$ in this case either.
    \end{claimproof}

    \begin{claim}
        For every $e \in \pi_{\bG}(s,c)$ and $e' \in \pi_{\bG}(c,t)$, we never overestimate $d_{\tilde{\bG} \setminus e,e'}(s,t)$ in \cref{eq:both_rec:1}.
    \end{claim}
    \begin{claimproof}
        We perform a case distinction.
        The first case is when $\pi_{\tilde{\bG} \setminus e,e'}(s,t)$ does not use any vertex $z \in \pi_{\bG}(s,t)$ such that $c \preceq z \preceq v'$.
        In such case, we consider the first node $x \in \pi_{\bG}(s,t)$ visited by $\pi_{\tilde{\bG} \setminus e,e'}(s,t)$ such that $v' \preceq x \preceq t$.
        In-between $x$ and $t$ the path $\pi_{\tilde{\bG} \setminus e,e'}(s,t)$ must coincide with $\pi_{\bG}(x,t)$. Moreover, $d_{\tilde{\bG} \setminus e,e'}(s,x) = d_{\tilde{\bL} \setminus e}(s,x)$.
        Thus, $d_{\tilde{\bG} \setminus e,e'}(s,t) = d_{\tilde{\bL} \setminus e}(s,x) + d_{\bG}(x,t)$ and \cref{eq:both_rec:1} captures this case through the first term.
        Similarly, the second term of \cref{eq:both_rec:1} captures the case when $\pi_{\tilde{\bG} \setminus e,e'}(s,t)$ does not use any vertex $z \in \pi_{\bG}(s,t)$ such that $u \preceq z \preceq c$.

        The remaining case is when $\pi_{\tilde{\bG} \setminus e,e'}(s,t)$ visits two vertices $x,y \in \pi_{\bG}(s,t)$ that satisfy $u \preceq x \preceq c$ and $c \preceq y \preceq v'$.
        We may assume that $\pi_{\tilde{\bG} \setminus e,e'}(s,t)$ in-between $x$ and $y$ equals to $\pi_{\bG}(x,y)$ (or alternatively $\pi_{\bG}(y,x)$, depending on whether $x$ comes before $y$ on $\pi_{\bG}(s,t)$). In any case, $\pi_{\tilde{\bG} \setminus e,e'}(s,t)$ passes through $c$. We conclude $d_{\tilde{\bG} \setminus e,e'}(s,t) = d_{\tilde{\bG} \setminus e,e'}(s,c) + d_{\tilde{\bG} \setminus e,e'}(c,t) \geq d_{\tilde{\bG} \setminus e}(s,c) + d_{\tilde{\bG} \setminus e'}(c,t)$.
    \end{claimproof}

    \subparagraph*{Runtime:}
    We only give the details how to compute efficiently the first two terms of \cref{eq:both_rec:1}. 
    The complexity of all other steps should not need any proof.

    Let us focus on the first term (the second term can be evaluated similarly).
    For each $e \in \pi_{\bG}(s,c)$, instead of minimizing $\min_{v' \preceq x \preceq t} \{ d_{\tilde{\bL} \setminus e}(s,x) + d_{\bG}(x,t) \}$, we minimize $\min_{v' \preceq x \preceq t} \{ d_{\tilde{\bL} \setminus e}(s,x) - d_{\bG}(s,x) \}$, as the two terms only differ by a constant offset $d_{\bG}(s,t)$.
    This permits us, for each $e \in \pi_{\bG}(s,c)$, to use the data structure of \cref{prp:range_query_or} over $|V(\pi_{\bG}(c,t))|$ values, where the $i$-th value equals to $d_{\tilde{\bL} \setminus e}(s,x) - d_{\bG}(s,x)$ assuming $x$ is the $i$-th node on $\pi_{\bG}(c,t)$.
    By doing so, for each $e \in \pi_{\bG}(s,c)$, we can query $\min_{v' \preceq x \preceq t} \{ d_{\tilde{\bL} \setminus e}(s,x) - d_{\bG}(s,x) \}$ in time $\Oh(1)$ after $\Oh(n)$ pre-processing.
    Over all $e\in \pi_{\bG}(s,c)$ this sums up to $\Oh(n^2)$ pre-processing time.
\end{proof}

Finally, we give the proof of \cref{thm:both}, where we put everything together.

\both

\begin{proof}
    We devise a divide-et-impera scheme for a more general problem, where we are also given two short-cut weights $w_s$ w.r.t. $s$ and $w_t$ w.r.t. $t$ and we want to compute $d_{\tilde{\bG} \setminus e,e'}(s,t)$ for all $e,e' \in \pi_{\bG}(s,t)$, where $\tilde{\bG}$ is $\bG$ augmented by $w_s$ and $w_t$. If $\pi_{\bG}(s,t)$ only contains a single edge between $s$ and $t$, then there is nothing left to compute, and we can return directly.

    Otherwise, we order the nodes on $\pi_{\bG}(s,t)$ (assuming there are $\ell \geq 3$ of them) obtaining $v_0 = s, v_1, \ldots, v_{\ell-1}=t$.
    For $i \in \fragmentco{0}{\ell}$, we let $\bA_i$ and $\bB_i$ be the subgraphs of $\bG$ induced by the sets
    \begin{align*}
         \{x \mid x \in \pi_{\bG}(s,y), y \in T_{\bG}(t,v_i)\} \text{ and } \{x \mid x \in \pi_{\bG}(t,y), y \in T_{\bG}(s,v_i)\}, \text{ respectively.}
    \end{align*}
    We observe that we have $V(\bA_{i-1}) \subseteq V(\bA_{i})$ for $i \in \fragmentoo{0}{\ell}$.
    This is because $V(T_{\bG}(t,v_{i-1})) \subseteq V(T_{\bG}(t,v_{i}))$, from which also follows $\{\pi_{\bG}(s,y) \mid y \in T_{\bG}(t,v_{i-1})\} \subseteq \{\pi_{\bG}(s,y) \mid y \in T_{\bG}(t,v_{i})\}$. By symmetry, we have $V(\bB_{i-1}) \supseteq V(\bB_{i})$ for $i \in \fragmentoo{0}{\ell}$.

    This allows us to proceed as follows. We choose the largest $k \in \fragmentoo{0}{\ell-1}$ such that $|V(\bA_{k-1})| \leq n/2$. Next, we use \cref{lem:both_conquer} three times, setting $c$ to $v_{k-1},v_{k},v_{k+1}$, respectively. Lastly, we recurse on $\bA_{k-1}$ and $\bB_{k+1}$
    (with short-cut weights given by \cref{lem:both_divide}, setting $f = \{v_{k-1},v_{k}\}$ for $\bA_{k-1}$  and $f = \{v_{k},v_{k+1}\}$ for $\bB_{k+1}$).

    \subparagraph*{Correctness:}
    Letting $e_i = \{v_{i-1}, v_{i}\}$ for $i \in \fragmentoo{0}{\ell}$,
    we observe that \cref{lem:both_divide} gives us $d_{\tilde{\bG} \setminus e_i,e_j}(s,t)$ for all $(i,j) \in (\fragmentco{0}{k-1} \times \fragmentco{k-1}{\ell}) \cup (\fragmentco{0}{k} \times \fragmentco{k}{\ell}) \cup (\fragmentco{0}{k+1} \times \fragmentco{k+1}{\ell}) \subseteq \fragmentco{0}{k} \times \fragmentoo{k}{\ell}$. 
    Recursing on $\bA_{k-1}$ and $\bB_{k+1}$, we compute the missing values $d_{\tilde{\bG} \setminus e_i,e_j}(s,t)$ for all $(i,j) \in \fragmentco{0}{k-1} \times \fragmentco{0}{k-1}$ and $(i,j) \in \fragmentoo{k}{\ell} \times \fragmentoo{k}{\ell}$.

    \subparagraph*{Runtime:} 
    By \cref{lem:both_divide}\ref{lem:both_divide:a}, we have $V(\bA_{k}) \cap V(\bB_{k+1}) = \emptyset$.
    Since $|V(\bA_k)| > n/2$, this implies  $|V(\bB_{k+1})| \leq n/2$.
    Moreover, from $V(\bA_{k-1}) \subseteq V(\bA_{k})$ also follows $V(\bA_{k-1}) \cap V(\bB_{k+1}) = \emptyset$ and $E(\bA_{k-1}) \cap E(\bB_{k+1}) = \emptyset$.
    Altogether, this yields a recursion of the type
    \begin{align*}
        \sT(n,m,W) \leq \sT(n_1,m_1,W) + \sT(n_2,m_2,W) + \Oh(n^2 + \sT_{\SSRP}(n,m,W)),
    \end{align*}
    for $n_1 + n_2 \leq n$, $m_1 + m_2 \leq m$ and $n_1, n_2 \leq n/2$. (Whenever $\rho \geq 1$, we need to add a log factor.)
    Using a similar analysis as in \cref{thm:wssrp}, we get the desired running time.
\end{proof}

\section{Proof for $e' \notin \pi_{\bG}(s,t)$}
\label{sec:single}

As a first step, we give the case distinction that we use to solve \cref{thm:single}.

\single

\begin{proof}
Let $e = \{u,v\}$ be such that $u$ is closer than $v$ to $s$ in $\pi_{\bG}(s, t)$.
We let $A = V(T_{\bG}(s,v))$ and $A'= V(T_{\bG}(t,u))$. (Technically, $A$ and $A'$ depend on $e$, but we hide this dependence.) Note that $A \cap A' = \emptyset$ because of \cref{prp:two_paths}. 

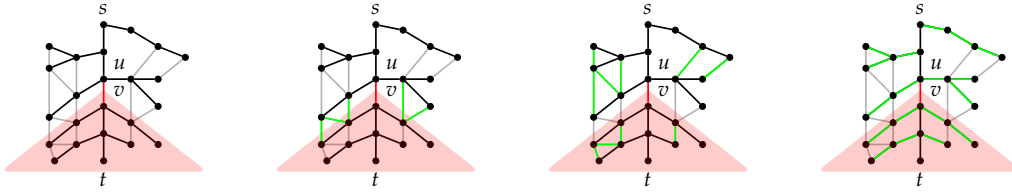
\begin{figure*}[thbp]
   \centering
   \scalebox{0.8}{\newcommand{\drawGraph}{
    \node[dot, label={above:$s$}] (root) at (0,0) {};
    \node[dot] (n1) at (0,-0.5) {};
    \node[dot, label={above right:$u$}] (n2) at (0,-1) {};
    \node[dot, label={above right:$v$}] (n3) at (0,-1.5) {};
    \node[dot] (n4) at (0,-2) {};
    \node[dot, label={below:$t$}] (n5) at (0,-2.5) {};
    \draw[thick] (root) -- (n1) -- (n2);
    \draw[thick, red] (n2) -- (n3);
    \draw[thick] (n3) -- (n4) -- (n5);

    \node[dot] (n6) at (-0.5,-1.8) {};
    \node[dot] (n7) at (-1,-2.2) {};
    \node[dot] (n8) at (0.5,-1.8) {};
    \node[dot] (n9) at (1,-2.2) {};
    \draw[thick] (n3) -- (n6) -- (n7);
    \draw[thick] (n3) -- (n8) -- (n9);

    \node[dot] (n10) at (-0.5,-2.2) {};
    \node[dot] (n11) at (-0.9,-2.5) {};
    \node[dot] (n12) at (0.5,-2.2) {};
    \draw[thick] (n4) -- (n10) -- (n11);
    \draw[thick] (n4) -- (n12);

    \node[dot] (n13) at (-0.5,-1.3) {};
    \node[dot] (n14) at (-1,-1.7) {};
    \draw[thick] (n2) -- (n13) -- (n14); 

    \node[dot] (n15) at (0.5,-1.0) {};
    \node[dot] (n16) at (1,-1.5) {};
    \node[dot] (n17) at (1,-1.0) {};
    \draw[thick] (n2) -- (n15) -- (n16); 
    \draw[thick] (n15) -- (n17);

    \node[dot] (n18) at (0.5,-0.1) {};
    \node[dot] (n19) at (1,-0.4) {};
    \node[dot] (n20) at (1.5,-0.6) {};
    \draw[thick] (root) -- (n18) -- (n19) -- (n20);

    \node[dot] (n21) at (-0.5,-0.6) {};
    \node[dot] (n22) at (-1,-0.8) {};
    \node[dot] (n23) at (-1,-0.4) {};
    \draw[thick] (n1) -- (n21) -- (n22); 
    \draw[thick] (n21) -- (n23); 
    \draw[thick, opacity=0.3] (n22) -- (n23);
    \draw[thick, opacity=0.3] (n22) -- (n14);
    \draw[thick, opacity=0.3] (n22) -- (n13);
    \draw[thick, opacity=0.3] (n21) -- (n13);
    \draw[thick, opacity=0.3] (n13) -- (n6);
    \draw[thick, opacity=0.3] (n14) -- (n7);
    \draw[thick, opacity=0.3] (n14) -- (n6);

    \draw[thick, opacity=0.3] (n6) -- (n10);
    \draw[thick, opacity=0.3] (n7) -- (n10);
    \draw[thick, opacity=0.3] (n7) -- (n11);

    \draw[thick, opacity=0.3] (n20) -- (n17);
    \draw[thick, opacity=0.3] (n19) -- (n15);

    \draw[thick, opacity=0.3] (n16) -- (n8);
    \draw[thick, opacity=0.3] (n15) -- (n8);
    \draw[thick, opacity=0.3] (n12) -- (n8);
    
    \fill[red, opacity=0.2, rounded corners] (0,-1.2) -- (1.9,-2.7) -- (-1.9,-2.7) -- (0,-1.2);
}

\begin{tikzpicture}[scale=1, dot/.style={
        draw, 
        circle, 
        fill, 
        minimum size=1mm, 
        inner sep=0pt
    }]
    \begin{scope}[scale=0.9]
        \drawGraph
    \end{scope}
    \begin{scope}[scale=0.9, shift={(5,0)}]
        \drawGraph
        \draw[thick, green] (n6) -- (n13);
        \draw[thick, green] (n6) -- (n14);
        \draw[thick, green] (n7) -- (n14);
        \draw[thick, green] (n16) -- (n8);
        \draw[thick, green] (n15) -- (n8);
    \end{scope}
    \begin{scope}[scale=0.9, shift={(10,0)}]
        \drawGraph
        \draw[thick, green] (n22) -- (n23);
        \draw[thick, green] (n22) -- (n14);
        \draw[thick, green] (n22) -- (n13);
        \draw[thick, green] (n21) -- (n13);
  
        \draw[thick, green] (n6) -- (n10);
        \draw[thick, green] (n7) -- (n10);
        \draw[thick, green] (n7) -- (n11);
    
        \draw[thick, green] (n20) -- (n17);
        \draw[thick, green] (n19) -- (n15);
    
        \draw[thick, green] (n12) -- (n8);
        \draw[thick, green] (n20) -- (n17);
    \end{scope}
    \begin{scope}[scale=0.9, shift={(15,0)}]
        \drawGraph
        \draw[thick, green] (n3) -- (n6) -- (n7);
        \draw[thick, green] (n3) -- (n8) -- (n9);
        \draw[thick, green] (n4) -- (n10) -- (n11);
        \draw[thick, green] (n4) -- (n12);
        \draw[thick, green] (n2) -- (n13) -- (n14); 
        \draw[thick, green] (n2) -- (n15) -- (n16); 
        \draw[thick, green] (n15) -- (n17);
        \draw[thick, green] (root) -- (n18) -- (n19) -- (n20);
        \draw[thick, green] (n1) -- (n21) -- (n22); 
        \draw[thick, green] (n21) -- (n23); 
    \end{scope}
\end{tikzpicture}}
   \caption{
        The first figure illustrates an example of $\bG$. Edge $e$ is highlighted in red, while all other edges in $T_{\bG}(s)$ are shown in bold black. The set of nodes $A$ is indicated by the red shaded area. The subsequent three figures (from left to right) display in green all edges $e'$ corresponding to cases \ref{case:i}, \ref{case:ii}, and \ref{case:iii}.
    }
   \label{fig:cases}
\end{figure*}

For the edge $e' \notin \pi_{\bG}(s,t)$, we distinguish three possible placements w.r.t. $s$:
\begin{enumerate}[(i)]
    \item $e' \in E_{\bG}(A, V(\bG) \setminus A)$;
    \label{case:i}
    \item $e' \notin T_{\bG}(s)$ and $e' \notin E_{\bG}(A, V(\bG) \setminus A)$;
    \label{case:ii}
    \item $e' \in T_{\bG}(s)$.
    \label{case:iii}
\end{enumerate}
(Refer to \cref{fig:cases} for a visualization of the cases.)
We also consider the symmetric placements of $e'$ w.r.t. $t,$ instead of $s$.
(That is, $e' \in E_{\bG}(A', V(\bG) \setminus A')$; $e' \notin T_{\bG}(t)$ and $e' \notin E_{\bG}(A', V(\bG) \setminus A')$; and $e' \in T_{\bG}(t)$.)
In \cref{sec:placement:i}, \cref{sec:placement:ii} and \cref{sec:placement:iii}, we handle separately three cases, which cover all cases, and depend on where $e'$ is placed w.r.t. $s$ and $t$:
\begin{description}
    \item[\Cref{lem:find_path} in \cref{sec:placement:i}] Edge $e'$ is in placement \ref{case:i} w.r.t. $s$ or $t$;
    \item[\Cref{lem:second_case} in \cref{sec:placement:ii}] Edge $e'$ is in placement \ref{case:ii} w.r.t. $s$ or $t$; and
    \item[\Cref{lemma:tree_case} in \cref{sec:placement:iii}] Edge $e'$ is in placement \ref{case:iii} w.r.t. $s$ and $t$.
\end{description}
This case distinction also needs to be implemented algorithmically. While it is straightforward to find all edges $e$ on the path $\pi_{\bG}(s,t)$ in time $\Oh(m + n \log n)$, we must also, for each $e$, find the corresponding edges $e'$ on the path $\pi_{\bG \setminus e}(s,t)$. Once all $e'$ for a given $e$ are identified, they can be categorized into one of the three cases. The algorithm in \cref{sec:placement:i} will allow us to do this, while handling the first case at the same time. 
\end{proof}

\subsection{Edge $e'$ is in placement \ref{case:i} w.r.t. $s$ or $t$}
\label{sec:placement:i}

We only show how to compute $d_{\bG \setminus e,e'}(s,t)$ for $e'$ that is in placement \ref{case:i} w.r.t. $s$. (The computation when $e'$ is in placement \ref{case:i} w.r.t. $s$ is symmetric.)
To handle this case and also identify all edges on $\pi_{\bG \setminus e}(s,t)$, we can use an algorithm very similar to the classic method in \cite{MMG89} for single-fault replacement paths.

\begin{lemma}\label{lem:find_path}
    In time $\Oh(n^2 + m \log n)$ we can compute all of the following:
    \begin{enumerate}[(a)]
        \item For each $e \in \pi_{\bG}(s,t)$ the path $\pi_{\bG \setminus e}(s,t)$.
        These paths contain at most $|\bigcup_{e \in \pi_{\bG}(s,t)} E(\pi_{\bG \setminus e}(s,t))| = \Oh(n)$ distinct edges.
        \label{lem:find_path:a}
        \item For each $e \in \pi_{\bG}(s,t)$ the values $d_{\bG \setminus e,e'}(s,t)$ for all $e' \in \pi_{\bG \setminus e}(s,t)$ with $e' \notin \pi_{\bG}(s,t), e' \in E_{\bG}(A, V(\bG) \setminus A)$, where $A$ is defined as in \cref{thm:single} for $e$.
        \label{lem:find_path:b}
    \end{enumerate}
\end{lemma}
\begin{proof}
    For sake of brevity, set $E \coloneqq E(\bG)$ and $V \coloneqq V(\bG)$.
    First, we devise two minimization problems for \ref{lem:find_path:a} and \ref{lem:find_path:b}. Then, we provide a single algorithm that addresses both cases together.

    Let $e \in \pi_{\bG}(s,t)$. For the minimization in \ref{lem:find_path:a}, we use \cref{prp:path_decomp} on $\bG, s,t$ and $F=\{e\}$ to get that there is $\{x,y\} \in E$ such that $x \notin A$ and $y \in A$ that satisfy
    \begin{equation}
        \pi_{\bG \setminus e} (s,t)
        = \pi_{\bG}(s,x) \circ \{x,y\} \circ \pi_{\bG \setminus e} (y,t).\label{eq:find_path:1}
    \end{equation}
    Moreover, $\pi_{\bG \setminus e} (y,t)$ is fully contained in $A$. 
    By \cref{prp:same_subtree} we have $\pi_{\bG} (y,t) = \pi_{\bG \setminus e} (y,t)$. Thus, for each $e \in \pi_{\bG}(s,t)$ it suffices to compute $\{x,y\} \in E$ such that $x \notin A$ and $y \in A$ that minimizes $d_{\bG}(s,x) + w_{\bG}(x,y) + d_{\bG} (y,t)$. Once we find such $x,y$, we can output $\pi_{\bG}(s,x) \circ \{x,y\} \circ \pi_{\bG} (y,t)$ as $\pi_{\bG \setminus e}(s,t)$.

    We proceed to give the minimization we need to solve for \ref{lem:find_path:b}.
    We begin by observing that for each $e \in \pi_{\bG \setminus e}(s,t)$ there is a unique $e' \in \pi_{\bG \setminus e}(s,t)$ that satisfies $e' \notin \pi_{\bG}(s,t)$ and $e' \in E_{\bG}(A, V \setminus A)$ because $\pi_{\bG \setminus e}(s,t) = \pi_{\bG}(s,x) \circ \{x,y\} \circ \pi_{\bG} (y,t)$.
    In this last expression, both $d_{\bG}(s,t)$ and $\pi_{\bG \setminus e} (y,t)$ can not contain such $e'$ because 
    $e' \notin T_{\bG}(s)$ and $\pi_{\bG} (y,t)$ is fully contained in $A$, respectively. Consequently, the unique $e'$ is $\{x,y\}$.
     (As a sidenote, notice that this implies $|\bigcup_{e \in \pi_{\bG}(s,t)} E(\pi_{\bG \setminus e}(s,t))| \leq |E(T_{\bG}(s))| + |E(T_{\bG}(s))| + n \leq \Oh(n)$.)
    
    To compute $d_{\bG \setminus e,e'}(s,t)$ for such $e'$ we observe that $T_{\bG}(s) =  T_{\bG\setminus e'}(s)$, meaning $A = T_{\bG \setminus e'}(s,v) = T_{\bG}(s,v)$. This allows us to use again \cref{prp:path_decomp} on $\bG \setminus e', s,t$ and $F=\{e\}$ to get that there are vertices $x',y'$ such that $x' \notin A$ and $y' \in A$ and $\{x',y'\} \neq e'$ that satisfy:
    \begin{equation}
        \pi_{\bG \setminus e,e'} (s,t)
        = \pi_{\bG \setminus e}(s,x') \circ \{x',y'\} \circ \pi_{\bG \setminus e,e'} (y,t). \label{eq:find_path:2}
    \end{equation}
     By $x' \notin A$ we have $e \notin \pi_{\bG \setminus e'}(s,x')$ and thus $\pi_{\bG \setminus e}(s,x') = \pi_{\bG}(s,x')$.
     Moreover, by \cref{prp:same_subtree} we have $\pi_{\bG \setminus e'} (y',t) = \pi_{\bG \setminus e,e'} (y',t)$. This can be simplified further: since $e' \in E_{\bG}(A,V \setminus A)$ and $\pi_{\bG \setminus e'} (y',t)$ is entirely contained in $A$, we have $\pi_{\bG \setminus e'} (y',t) = \pi_{\bG} (y',t)$. Thus, for each $e \in \pi_{\bG}(s,t)$ it suffices to compute $\{x',y'\} \in E \setminus e$ with $x \notin A$ and $y' \in A$ that minimize $d_{\bG}(s,x') + w_{\bG}(x',y') + d_{\bG} (y',t)$.
     This is the same minimization problem as in \ref{lem:find_path:a} except that $\{x',y'\}$ is in $E \setminus e$ instead of $E$.

    To solve \ref{lem:find_path:a} and \ref{lem:find_path:b} simultaneously, we iterate over all $e \in \pi_{\bG}(s,t)$ in order of increasing hop-distance from $s$. During this traversal, we maintain the values $d_{\bG}(s,x) + w_{\bG}(x,y) + d_{\bG}(y,t)$ for all $\{x,y\} \in E \cap E_{\bG}(V \setminus A, A)$ such that $\{x,y\} \neq e$ using a data structure that supports insertion, deletion, and retrieval of the smallest and second-smallest elements in time $\Oh(\log n)$. As previously shown, these minima correspond to the solutions for \ref{lem:find_path:a} and \ref{lem:find_path:b}. When transitioning from an edge $e = \{u,v\}$ to the next edge $\{v,w\}$ on $\pi_{\bG}(s,t)$, we update the data structure by removing values corresponding to $\{x,y\} \in E \cap (T_{\bG}(s,u) \times T_{\bG}(s,v))$ and inserting those corresponding to $\{x,y\} \in E \cap (T_{\bG}(s,v) \times T_{\bG}(s,w))$. Aggregating these operations over all $e \in \pi_{\bG}(s,t)$ yields overall $m$ insertions and deletions, and the claimed running time follows.
\end{proof}

\subsection{Edge $e'$ is in placement \ref{case:ii} w.r.t. $s$ or $t$}
\label{sec:placement:ii}

We again only show how to compute $d_{\bG \setminus e,e'}(s,t)$ for $e'$ that is in placement \ref{case:ii} w.r.t. $s$.
This is the case already presented in \cref{sec:to}, but here we give the computations needed in this case in a slightly generalized setting, allowing us to re-use them in \cref{sec:placement:iii}.

\begin{proposition}\label{lem:centroid}
    Let $\bT$ be a connected subtree of $T_{\bG}(s)$. Further, let $F' \subseteq E(\bG)$
    and $n_{\bT} \coloneqq |V(\bT)|$. 
    If we have precomputed the \SSRP values from $t$ in $\bG$ ahead, then we can compute in time $\Oh(n_{\bT} \cdot (n + |F'|) \cdot \log n_{\bT})$ for each $f \in \bT$ and $f' \in F'$, the value $d_{f,f'}$ defined as
    \begin{align}
        d_{f,f'}
        \coloneqq \min_{x \in \bT \setminus T_{\bG}(s,f), y \in T_{\bG}(s,f)} \big\{ \ d_{\bG}(s,x) + w_{\bG}(x,y) + d_{\bG \setminus f'} (y,t) \ \big\}. \label{eq:centroid}
    \end{align}
\end{proposition}
\begin{proof}
    We solve the problem recursively. To this end, we split $\bT$ into $\bT_1$ and $\bT_2$ using \cref{lem:separator}, and find the corresponding centroid $c$. Next, for $f \in \bT_1$, we rewrite \cref{eq:centroid} as
    \begin{align*}
        d_{e,e'}
        = \min \Big\{ \min_{\substack{x \in \bT_1 \setminus T_{\bG}(s,f)\\ y \in T_{\bG}(s,f)}} \big\{ d_{\bG}(s,x) + w_{\bG}(x,y) + d_{\bG \setminus f'} (y,t) \big\},
        \min_{\substack{x \in \bT_2 \\y \in T_{\bG}(s,f)}} \big\{ d_{\bG}(s,x) + w_{\bG}(x,y) + d_{\bG \setminus f'} (y,t) \big\} \Big\},
    \end{align*}
    where we use that if $x \in \bT_2$ and $f \in \bT_1$, then $x \notin T_{\bG}(s,f)$.
    (notice that this is only true when $f$ is not on the path from the root of $\bT$ to $c$, however if $f$ is on such path then the second term is non-existent, and we can ignore it.)
    To compute the second term, we do the following:
    \begin{itemize}
        \item for each node $y \in \bG$ we compute $h_{y} \coloneqq \min_{x \in \bT_2} \{d_{\bG}(s,x) + w_{\bG}(x,y) \}$ in time $\Oh(n_\bT \cdot n)$ ;
        \item we setup for each $f' \in F'$ a data structure from \cref{prp:range_query} where $y$ has the associated value $h_{y} + d_{\bG \setminus f'} (y,t)$ in time $\Oh(|F'| \cdot n)$.
        \item for each $f \in \bT_1$ and $f' \in F'$ we query $\min_{y \in T_{\bG}(s,f)} \{ h_{y} + d_{\bG \setminus f'} (y,t) \}$ in time $\Oh(|F'| \cdot n_{\bT})$.
    \end{itemize} 
    For the first term, we can recurse on $\bT_1$.

    Similarly, when $f \in \bT_2$ we rewrite \cref{eq:centroid} as
    \begin{align*}
        d_{e,e'}
        = \min \Big\{ \min_{\substack{x \in \bT_2 \setminus T_{\bG}(s,f)\\ y \in T_{\bG}(s,f)}} \big\{ d_{\bG}(s,x) + w_{\bG}(x,y) + d_{\bG \setminus f'} (y,t) \big\},
        \min_{\substack{x \in \bT_1 \\y \in T_{\bG}(s,f)}} \big\{ d_{\bG}(s,x) + w_{\bG}(x,y) + d_{\bG \setminus f'} (y,t) \big\} \Big\},
    \end{align*}
    where we use again $x \in \bT_1$ and $f \in \bT_2$ implies $x \notin T_{\bG}(s,f)$. As before, we can compute the first term by recursing in $\bT_2$ and the first by performing symmetric computations to above.

    For the running time notice that we use time $\Oh(n_{\bT}\cdot (n+ |F'|))$ plus we recurse on $\bT_1$ and $\bT_2$ that satisfy $(n_{\bT_1}-1)+(n_{\bT_2}-1)=(n_{\bT}-1)$. Since with every level of recursion $n_\bT$ shrinks by a constant fraction, we get that the overall running time is $\Oh(n_\bT \cdot (n+ |F'|) \cdot \log n_\bT)$.
\end{proof}

As the proof is very short, we repeat that \cref{lem:centroid} allows us to solve the case when $e' \notin T_{\bG}(s)$ and $e' \notin E_{\bG}(A, V \setminus A)$.

\begin{lemma}\label{lem:second_case}
    Let $\sT$ be the time needed to execute \SSRP from $t$ in $\bG$.
    Then, we can compute in time $\sT + \Oh(n^2 \log n)$ the values $d_{\bG \setminus e,e'}(s,t)$ for all $e \in \pi_{\bG}(s,t)$ and $e' \in \pi_{\bG \setminus e}(s,t)$ that satisfy $e' \notin T_{\bG}(s)$ and $e' \notin E_{\bG}(A, V \setminus A)$.
\end{lemma}
\begin{proof}
    We decompose similarly to \cref{lem:find_path}\ref{lem:find_path:b}.
    Again, we observe that $A = T_{\bG \setminus e'}(s,e) = T_{\bG}(s,e)$ and we use \cref{prp:path_decomp} on $\bG \setminus e', s$ and $F=\{e\}$ to get that there is $\{x,y\} \in E(\bG) \setminus e'$ with $x \notin A$ and $y \in A$ such that:
    \begin{equation}
        \pi_{\bG \setminus e,e'} (s,t)
        = \pi_{\bG \setminus e'}(s,x) \circ \{x,y\} \circ \pi_{\bG \setminus e,e'} (y,t). \label{eq:placement_ii:1}
    \end{equation}
    Similarly, we can simplify $\pi_{\bG \setminus e'}(s,x) = \pi_{\bG}(s,x)$ and $\pi_{\bG \setminus e,e'} (y,t) = \pi_{\bG \setminus e'} (y,t)$. This time, however, $\pi_{\bG \setminus e'} (y,t) = \pi_{\bG} (y,t)$ does not necessarily hold. In turn, we can take $\{x,y\} \in E(\bG)$ instead of $\{x,y\} \in E(\bG) \setminus e'$ because $e' \notin E_{\bG}(A, V \setminus A)$.
    Therefore, we can compute $d_{\bG \setminus e,e'}(s,t)$ by using \cref{lem:centroid} on $\bG, T_{\bG}(s)$ and setting $F$ to $\bigcup_{e \in \pi_{\bG}(s,t)} E(\pi_{\bG \setminus e}(s,t))$ to compute $\pi_{\bG \setminus e,e'} (s,t)$
    in time $\Oh(n^2 \log n)$.
    Since we need to compute $\SSRP$ from $t$ before using \cref{lem:centroid}, the claimed running time follows.
\end{proof}

\subsection{Edge $e'$ is in placement \ref{case:iii} w.r.t. $s$ and $t$}
\label{sec:placement:iii}

Fix $e \in \pi_{\bG}(s,t)$ and assume $e' = \{u',v'\} \notin \pi_{\bG}(s,t)$ is in placement \ref{case:iii} w.r.t. $s$ and $t$.
Assume $u'$ is a parent of $v'$ in $T_{\bG}(s)$.
We claim that $e'$ must be used in the opposite direction in $T_{\bG}(s)$ and $T_{\bG}(t)$, i.e., node $v'$ is a parent of $u'$ in $T_{\bG}(t)$.
Otherwise, $\pi_{\bG \setminus e}(s,t) = \pi_{\bG}(s,u') \circ \pi_{\bG}(u',t)$, where neither $\pi_{\bG}(s,u')$ nor $\pi_{\bG}(u',t)$ contains $e'$, but we assumed that $e' \in \pi_{\bG \setminus e}(s,t)$.
This allows us to define $B = T_{\bG}(s,v')$ and $B' = T_{\bG}(t,u')$.
(Again, we hide the dependency on $e,e'$ in $B,B'$).

Since $A,B$ are both subtrees of $T_{\bG}(s)$, we either have $A \cap B = \emptyset$ or $B \subseteq A$. (Note that $A \subseteq B$ is not possible, as otherwise $u',v'$ would be ancestors of $u,v$, which implies $e' \in \pi_{\bG}(s,t)$.)
Similarly, we either have $A' \cap B' = \emptyset$ or $B' \subseteq A'$.
Next, we prove that at least one of the two pairs of sets $(A,B)$ and $(A',B')$ falls in the former case, i.e., the two sets in the pairs are disjoint. 

\begin{proposition}
    At least one of $A \cap B = \emptyset$ and $A' \cap B' = \emptyset$ must hold.
\end{proposition}
\begin{proof}
    For sake of contradiction, assume that $B \subseteq A$ and $B' \subseteq A'$. This means that $u',v' \in A$ and $u',v' \in A'$. But $A \cap A' = \emptyset$, yielding a contradiction.   
\end{proof}

Thus, we can restrict ourselves to computing $d_{\bG \setminus e,e'}(s,t)$ for all $e,e'$ in this case that satisfy $A \cap B = \emptyset$. (If $A \cap B$ does not hold for $e,e'$, then by swapping the role of $s$ and $t$, we get $A' \cap B' = \emptyset$ and so we can use symmetric computations.)

\begin{proposition}\label{prp:case3_helper}
    Whenever $A \cap B = \emptyset$, the two following hold:
    \begin{enumerate}[(a)]
        \item For any $y \in A$, we have that $\pi_{\bG \setminus e,e'}(y,t) = \pi_{\bG \setminus e'}(y,t)$.
        \label{it:case3_helper:d}
        \item For any $y \in B$, we have $\pi_{\bG \setminus e,e'}(y,t) = \pi_{\bG \setminus e}(y,t)$.
        \label{it:case3_helper:a}
    \end{enumerate}
\end{proposition}
\begin{proof}
    We first prove \ref{it:case3_helper:d}. Since both $t,y \in A$ we have $\pi_{\bG \setminus e'}(s,t) = \pi_{\bG}(s,t)$ and $\pi_{\bG \setminus e'}(s,y) = \pi_{\bG}(s,y)$.
    Thus, $y,t \in T_{\bG \setminus e'}(s, v)$,
    and we can use \cref{prp:same_subtree} on $\bG \setminus e'$, $t$, $v$ and $y,t$ to conclude $\pi_{\bG \setminus e,e'}(y,t) = \pi_{\bG \setminus e'}(y,t)$.

    We proceed to prove \ref{it:case3_helper:a}.
    Since $A \cap B = \emptyset$ and $y \in B$, we have $\pi_{\bG \setminus e}(s,y) = \pi_{\bG}(s,y)$. We obtain that  $\pi_{\bG \setminus e}(s,y)$ uses $e'$ such that $u'$ comes before $v'$ because $\pi_{\bG}(s,y)$ does.
    But also $\pi_{\bG \setminus e}(s,t)$ uses $e'$ such that $u'$ comes before $v'$.
    Thus, $y,t \in T_{\bG \setminus e}(s, v')$
    and we can use \cref{prp:same_subtree} on $\bG \setminus e$, $s$, $v'$ and $y,t$ to conclude $\pi_{\bG \setminus e,e'}(y,t) = \pi_{\bG \setminus e}(y,t)$.
\end{proof}

We are now ready to give \cref{lemma:tree_case}, solving the last case of \cref{thm:single}.

\begin{lemma}\label{lemma:tree_case}
    Let $\sT$ be the time needed to execute \SSRP from $t$ in $\bG$.
    Then, we can compute in time $\sT + \Oh(n^2 \log^2 n)$ the values $d_{\bG \setminus e,e'}(s,t)$ for all $e \in \pi_{\bG}(s,t)$ and $e' \in \pi_{\bG \setminus e}(s,t)$ that satisfy $e' \in T_{\bG}(s)$ and $e' \in T_{\bG}(t)$.
\end{lemma}

\begin{proof}
    We will solve this case by using a recursive subroutine. All calls to this subroutine will have access to the following values that we can compute ahead:
    \begin{enumerate}[(a)]
        \item For each node $y$ in $\bG$ we construct in time $\Oh(n)$ a data structure from \cref{prp:range_query} over $T_\bG(s)$ and where to node $x$ we assign value $d_{\bG}(s,x) + w_{\bG}(x,y)$;
        \label{it:prep:a}
        \item \SSRP values from $t$.
        \label{it:prep:b}
    \end{enumerate}

    To define our recursive subroutine, we first order the nodes on $\pi_{\bG}(s,t)$ (assuming there are $\ell$ of them)
    obtaining
    \[
        v_0 = s, v_1, \ldots, v_{\ell-1}=t.
    \]
    For $i,j \in \fragmentco{0}{\ell}$ with $i < j$, we let $\bG\fragmentco{i}{j}$ be the subgraph of $\bG$ defined as $T_{\bG}(s, v_i) \setminus T_{\bG}(s, v_j)$. If $i > 0$, then we additionally add node $v_{i-1}$ and the edge $\{v_{i-1}, v_i\}$ to it.
    Moreover, for $i \in \fragmentco{0}{\ell}$, we let $\bG\fragmentco{i}{\ell}$ be $T_{\bG}(s, v_i)$. We denote with $n_{i,j}$ the number of vertices in $\bG\fragmentco{i}{j}$. 
    
    This allows us to define the following problem. 
    
    \defproblem
    {$(i,j,k)$-\FRP}
    {Indices $i,j,k\in \fragmentco{0}{\ell}$ such that $i < j < k$, a subgraph $\bH$ of $\bG$ and access to Data~\ref{it:prep:a}, \ref{it:prep:b}.}
    {The values $d_{\bG\setminus e,e'}(s,t)$ for all $e' \in \bG\fragmentco{i}{k}$ and $e \in \bG\fragmentco{k}{j}$.}

    Next, we show how solving $(i,j,k)$-\FRP helps to prove \cref{lemma:tree_case}.

    \begin{claim}\label{claim:tree_case:1}
        Suppose, that for any $i,j,k \in \fragment{1}{\ell}$ such that $i < j < k$, we can solve $(i,j,k)$-\FRP in time $\Oh (n \cdot n_{i,k} \log n_{i,k})$. Then, \cref{lemma:tree_case} holds.
    \end{claim}
    \begin{claimproof}
        We want to devise a recursive procedure that runs in $\Oh (n \cdot n_{i,k} \log^2 n_{i,k})$ for the following variation:
        Given $i,k \in \fragment{1}{\ell}$ such that $i < k$, output $d_{\bG\setminus e,e'}(s,t)$ for all $e,e' \in \bG\fragmentco{i}{k}$. If we can do this, then setting $i=0$ and $k=\ell$ yields the claim (together with pre-computing data  \ref{it:prep:a}, \ref{it:prep:b} in time $\sT + \Oh(n^2$).
      
        In order to solve such instance, we find the smallest $j$ such that $i < j \leq k$ and $n_{i,j} \geq n_{i,k} / 3$. If such $j$ also satisfies $n_{i,j} \leq 2n_{i,k} / 3$, then we solve $(i,j,k)$-\FRP and recurse on $\bG\fragmentco{i}{j}$ and $\bG\fragmentco{j}{k}$, both with size $n_{i,j},n_{j,k} \leq 2n_{i,k} / 3$.
        Otherwise, if $n_{i,j} > 2n_{i,k}/3$, we must have $n_{j-1,j} > n_{i,k}/3$. In this case, we can solve the instances $(i,j-1,k)$-\FRP and $(j-1,j,k)$-\FRP and then recurse on $\bG\fragmentco{i}{j-1}$ and $\bG\fragmentco{j}{k}$ (technically, we recurse on the former and latter only if $j \neq i+1$ and $j \neq k$, respectively).

        Note that in the base case, when $k=i+1$, then there is nothing to compute: if $i = 0$, there is no possible choice of $e$, and if $i > 0$, then there is no possible choice of $e'$ (because any edge other than $e$ will be in the subtree of the only possible $e$ which is $\{v_{i-1},v_{i}\}$). 

        For the running time, notice that we obtain a recursion of the type $\sT(n_{i,k}) \leq \sT(n_{i,j}) + \sT(n_{j,k}) + \Oh(n \cdot  n_{i,k} \log n_{i,K})$ where $(n_{i,k}-1) = (n_{i,j} - 1) + (n_{j,k}+1)$ and $n_{i,j},n_{j,k} \leq 2n_{i,k}/3$. Solving the recursion yields that $\sT(n_{i,k}) = \Oh(n \cdot n_{i,k} \log n_{i,k})$.
    \end{claimproof}

    Thus, in the remaining part of the proof it suffices that we explain how to solve an $(i,j,k)$-\FRP instance in time $\sT_{\APSP}(n_{i,k})$.
    To this end, we use \cref{prp:path_decomp} on $\bG, s$ and $F=\{e,e'\}$ to get that there are vertices $x,y$ such that $x \notin A\cup B$ and $y \in A\cup B$ and that satisfy:
    \begin{equation}
        \pi_{\bG \setminus e,e'} (s,t)
        = \pi_{\bG}(s,x) \circ \{x,y\} \circ \pi_{\bG \setminus e,e'} (y,t). \label{eq:main_decomp_pre}
    \end{equation}
    (Note, to apply this correctly we are also using $A \cap B = \emptyset$.)
    \Cref{eq:main_decomp_pre} yields:
    \begin{equation}
        d_{\bG \setminus e,e'} (s,t)
        = \min_{x \notin A \cup B, y \in A \cup B} \big\{ \ d_{\bG}(s,x) + w_{\bG}(x,y) + d_{\bG \setminus e,e'} (y,t) \ \big\}. \label{eq:main_decomp}
    \end{equation}
    
    In the four following claims, we treat differently four cases that depend on $x$ and $y$:
    \begin{description}
        \item[\cref{claim:tree_case:3}:] $x \in \bG\fragmentco{0}{j}$ and $y \in B$;
        \item[\cref{claim:tree_case:4}:] $x \in \bG\fragmentco{j}{k}$ and $y \in B$;
        \item[\cref{claim:tree_case:5}:] $x \in \bG\fragmentco{0}{j}$ and $y \in A$; and
        \item[\cref{claim:tree_case:6}:] $x \in \bG\fragmentco{j}{k}$ and $y \in A$.
    \end{description}
    Note that this case distinction is indeed complete,
    as $x \in \bG\fragmentco{k}{\ell}$ is not possible because $\bG\fragmentco{k}{\ell} \subseteq A$ for any $e \in E(\bG\fragmentco{j}{k})$. Next, we derive $d_{\bG \setminus e,e'}(s,t)$ for the four specific cases. As the actual case for a pair $e, e'$ is unknown, we take the minimum over all four. Since inapplicable cases strictly provide upper bounds, the minimum over all four yields the correct answer.

    \begin{claim}\label{claim:tree_case:3}
         We can compute all values $d_{\bG \setminus e,e'}(s,t)$ such that $x \in \bG\fragmentco{0}{j}$ and $y \in B$ in time $\Oh(n_{i,k}^2 \log n_{i,k})$.
    \end{claim}

    \begin{claimproof}
        Since $x \in \bG\fragmentco{0}{j}$ and $e \in \bG\fragmentco{j}{k}$ we directly have $x \notin A$. This together with $y \in B$, allows us to simplify \cref{eq:main_decomp} via \cref{prp:case3_helper}\ref{it:case3_helper:a} to get that $d_{\bG \setminus e,e'} (s,t)$ equals to
        \begin{align*}
            &\min_{x \in \bG\fragmentco{0}{j} \setminus B, y \in B} \big\{ \ d_{\bG}(s,x) + w_{\bG}(x,y) + d_{\bG \setminus e} (y,t) \ \big\} \\
            &= \min \Big\{ \min_{\substack{x \in \bG\fragmentco{0}{i}\\y \in B}} \big\{ \ d_{\bG}(s,x) + w_{\bG}(x,y) + d_{\bG \setminus e} (y,t) \ \big\},
            \min_{\substack{x \in \bG\fragmentco{i}{j} \setminus B\\y \in B}} \big\{ \ d_{\bG}(s,x) + w_{\bG}(x,y) + d_{\bG \setminus e} (y,t) \ \big\} \Big\},
        \end{align*}
        where we use that if $B \cap \bG\fragmentco{0}{i} = \emptyset$ since $e' \in \bG\fragmentco{i}{j}$.
        The second term in this last expression can be computed in time $\Oh(n_{i,k}^2)$ by using \cref{lem:centroid} on $\bG$, $\bG\fragmentco{i}{j}$ and $\pi_{\bG}(s,t) \cap E(\bG\fragmentco{j}{k})$ (in the role of $F'$). 
        Conversely, the first term can be computed in the following steps:
        \begin{itemize}
        \item for each node $y \in \bG\fragmentco{i}{j}$ we get $h_{y} \coloneqq \min_{x \in \bG\fragmentco{0}{i}} \{d_{\bG}(s,x) + w_{\bG}(x,y) \}$ by querying the data structure from Data \ref{it:prep:a} in time $\Oh(n_{i,k})$;
        \item we setup for each $e \in \pi_{\bG}(s,t) \cap E(\bG\fragmentco{j}{k})$ a data structure from \cref{prp:range_query} over the nodes $y \in \bG\fragmentco{i}{j}$ where $y$ has the associated value $h_{y} + d_{\bG \setminus e} (y,t)$ in time $\Oh(n_{i,k}^2)$; and
        \item for each $e' \in E(\bG\fragmentco{i}{j}) \setminus \pi_{\bG}(s,t)$ and $e \in \pi_{\bG}(s,t) \cap E(\bG\fragmentco{j}{k})$ we query $\min_{y \in B} \{ h_{y} + d_{\bG \setminus e} (y,t) \}$ in time $\Oh(n_{i,k}^2)$. \claimqedhere
    \end{itemize} 
    \end{claimproof}

    \begin{claim}\label{claim:tree_case:4}
         We can compute all values $d_{\bG \setminus e,e'}(s,t)$ such that $x \in \bG\fragmentco{j}{k}$ and $y \in B$ in time $\Oh(n_{i,k}^2)$.
    \end{claim}

    \begin{claimproof}
        Since $x \in \bG\fragmentco{j}{k}$ and $e' \in \bG\fragmentco{i}{j}$ we directly have $x \notin B$.
        Moreover, since $e \in \bG\fragmentco{j}{k}$, we have
        that $x \in \bG\fragmentco{j}{k} \setminus A$.
        This together with $y \in B$, allows us to simplify \cref{eq:main_decomp} via \cref{prp:case3_helper}\ref{it:case3_helper:a} again to get:
        \begin{align*}
            d_{\bG \setminus e,e'} (s,t)
            = \min_{x \in \bG\fragmentco{k}{j} \setminus A, y \in B} \big\{ \ d_{\bG}(s,x) + w_{\bG}(x,y) + d_{\bG \setminus e} (y,t) \ \big\}.
        \end{align*}
        This last expression can be computed in the following steps:
        \begin{itemize}
            \item For each node $y \in \bG\fragmentco{i}{j}$ and edge $e \in \pi_{\bG}(s,t) \cap E(\bG\fragmentco{j}{k})$ we get $h_{e,y} \coloneqq \min_{x \in \bG\fragmentco{k}{j} \setminus A} \{d_{\bG}(s,x) + w_{\bG}(x,y) \}$ by querying the data structure from Data \ref{it:prep:a} in time $\Oh(n_{i,k}^2)$;
            \item we setup for each $e \in \pi_{\bG}(s,t) \cap E(\bG\fragmentco{j}{k})$ a data structure from \cref{prp:range_query} over the nodes $y \in \bG\fragmentco{i}{j}$ where $y$ has the associated value $h_{y,e} + d_{\bG \setminus e} (y,t)$ in time $\Oh(n_{i,k}^2)$; and
            \item for each $e' \in E(\bG\fragmentco{i}{j}) \setminus \pi_{\bG}(s,t)$ and $e \in \pi_{\bG}(s,t) \cap E(\bG\fragmentco{j}{k})$ we query $\min_{y \in B} \{ h_{y} + d_{\bG \setminus e} (y,t) \}$ in time $\Oh(n_{i,k}^2)$.
            \claimqedhere
        \end{itemize}
    \end{claimproof}    

    \begin{claim}\label{claim:tree_case:5}
        We can compute all values $d_{\bG \setminus e,e'}(s,t)$ such that $x \in \bG\fragmentco{0}{j}$, $y \in \bG\fragmentco{j}{k}$ in time $\Oh(n_{i,k}^2)$.
    \end{claim}

    \begin{claimproof}
        Since $x \in \bG\fragmentco{0}{j}$ and $e \in \bG\fragmentco{j}{k}$ we directly have $x \notin A$. This together with $y \in A$, allows us to simplify \cref{eq:main_decomp} via \cref{prp:case3_helper}\ref{it:case3_helper:d} to get: 
        \begin{align*}
            d_{\bG \setminus e,e'} (s,t)
            = \min_{x \in \bG\fragmentco{0}{j} \setminus B, y \in A} \big\{ \ d_{\bG}(s,x) + w_{\bG}(x,y) + d_{\bG \setminus e'} (y,t) \ \big\}.
        \end{align*}
        This last expression is symmetric to the one in \cref{claim:tree_case:4}. Similarly, we can compute:
        \begin{itemize}
            \item For each node $y \in \bG\fragmentco{j}{k}$ and edge $e' \in E(\bG\fragmentco{i}{j}) \setminus \pi_{\bG}(s,t)$ we get $h_{e',y} \coloneqq \min_{x \in \bG\fragmentco{0}{j} \setminus B} \{d_{\bG}(s,x) + w_{\bG}(x,y) \}$ by querying the data structure from Data \ref{it:prep:a} in time $\Oh(n_{i,k}^2)$;
            \item we setup for each $e' \in E(\bG\fragmentco{i}{j}) \setminus \pi_{\bG}(s,t)$ a data structure from \cref{prp:range_query} over the nodes $y \in \bG\fragmentco{j}{k}$ where $y$ has the associated value $h_{y,e'} + d_{\bG \setminus e'} (y,t)$ in time $\Oh(n_{i,k}^2)$; and
            \item for each $e' \in E(\bG\fragmentco{i}{j}) \setminus \pi_{\bG}(s,t)$ and $e \in \pi_{\bG}(s,t) \cap E(\bG\fragmentco{j}{k})$ we query $\min_{y \in A} \{ h_{y} + d_{\bG \setminus e} (y,t) \}$ in time $\Oh(n_{i,k}^2)$. \claimqedhere
        \end{itemize}
    \end{claimproof}    

    \begin{claim}\label{claim:tree_case:6}
        We can compute all values $d_{\bG \setminus e,e'}(s,t)$ such that $x \in \bG\fragmentco{j}{k}$ and $y \in A$ in time $\Oh(n_{i,k} \cdot n \log n_{i,k})$.
    \end{claim}
    \begin{claimproof}
       Since $x \in \bG\fragmentco{j}{k}$ and $e' \in \bG\fragmentco{i}{j}$ we directly have $x \notin B$. This together with $y \in A$, allows us to simplify \cref{eq:main_decomp} via \cref{prp:case3_helper}\ref{it:case3_helper:d} to get: 
        \begin{align*}
            d_{\bG \setminus e,e'} (s,t)
            = \min_{x \in \bG\fragmentco{j}{k} \setminus A, y \in A} \big\{ \ d_{\bG}(s,x) + w_{\bG}(x,y) + d_{\bG \setminus e'} (y,t) \ \big\}.
        \end{align*}
        This last expression can be computed in time $\Oh(n_{i,k} \cdot n)$ by using \cref{lem:centroid} on $\bG$, $\bG\fragmentco{j}{k}$ and $E(\bG\fragmentco{i}{j}) \setminus \pi_{\bG}(s,t)$ (in the role of $F'$). 
    \end{claimproof}   
    
    This concludes the proof of \cref{lemma:tree_case}.
\end{proof}

\section{Lower Bound}
\label{sec:lb}

As a first step, we demonstrate how to encode a single 2-FRP instance into another instance of a more sparse triangle detection problem, formulated as follows:

\defproblem
{Sparse $(\sqrt{n},n,n)$-Triangle Detection}
{A tripartite graph $\bG$ with $m$ edges s.t. $V(\bG) = A \cup B \cup C$ with $|A| = \Oh(\sqrt{n})$, $|B|,|C| = \Oh(n)$.}
{\yes if there are $a,b,c \in A \times B \times C$ such that $(a,b),(b,c),(c,a) \in E(\bG)$, and \no otherwise.}

In the Sparse $(\sqrt{n},n,n)$-Minimum Weight Triangle Problem, the graph is additionally weighted and we want to find a triangle of minimal weight.

The Sparse $(\sqrt{n},n,n)$-Triangle Detection for $m=n^{1.5}$ was considered in \cite{RodittyW12} who related its complexity to girth approximation in undirected graphs and hypothesized that $n^{2-o(1)}$ time is required even using fast matrix multiplication. It is easy to observe that for any $m$, one can combinatorially reduce BMM to Sparse $(\sqrt{n},n,n)$-Triangle Detection, so that under the BMM Hypothesis $mn^{0.5-o(1)}$ time is required for any $m$ that is a polynomial function of $n$. 

In the All-Edge version of Sparse $(\sqrt n,n,n)$-Triangle detection, one needs to decide for every edge in $B\times C$ whether it is in some triangle. Abboud et al. \cite{fullysparse} showed that this version of the problem (for $m=n^{1.5}$) is equivalent to the so-called Fully Sparse Boolean Matrix multiplication problem.

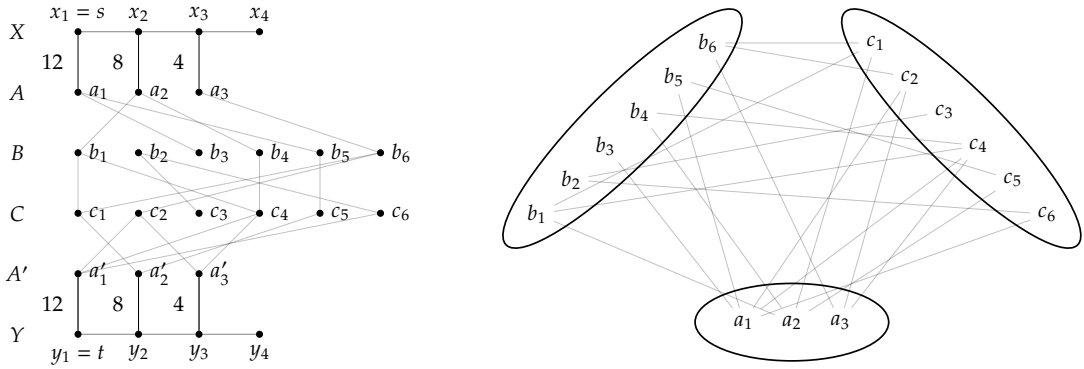
\begin{figure*}[thbp]
   \centering
   \scalebox{0.8}{\begin{tikzpicture}[scale=1, dot/.style={
        draw, 
        circle, 
        fill, 
        minimum size=1mm, 
        inner sep=0pt
    }]

    \foreach \i in {1,...,3} {
      \node[dot, label={right:$a_{\i}$}] (A\i) at (1*\i,-1) {};
    }
    \node at (0,-1) {$A$};

    \foreach \i in {1,...,6} {
      \node[dot, label={right:$b_{\i}$}] (B\i) at (1*\i,-2) {};
    }
    \node at (0,-2) {$B$};

    \foreach \i in {1,...,6} {
      \node[dot, label={right:$c_{\i}$}] (C\i) at (1*\i,-3) {};
    }
    \node at (0,-3) {$C$};

    \foreach \i in {1,...,3} {
      \node[dot, label={right:$a_{\i}'$}] (D\i) at (1*\i,-4) {};
    }
    \node at (0,-4) {$A'$};

    \node[dot, label={above:$x_{1}=s$}] (X1) at (1,0) {};
    \node[dot, label={above:$x_{2}$}] (X2) at (2,0) {};
    \node[dot, label={above:$x_{3}$}] (X3) at (3,0) {};
    \node[dot, label={above:$x_{4}$}] (X4) at (4,0) {};

    \draw (X1) --node[label={left:$12$}] {} (A1);
    \draw (X2) --node[label={left:$8$}] {} (A2);
    \draw (X3) --node[label={left:$4$}] {} (A3);

    \draw[opacity=0.5] (X1) -- (X2) -- (X3) -- (X4); 
    \node at (0,0) {$X$};

    \node[dot, label={below:$y_{1}=t$}] (Y1) at (1,-5) {};
    \node[dot, label={below:$y_{2}$}] (Y2) at (2,-5) {};
    \node[dot, label={below:$y_{3}$}] (Y3) at (3,-5) {};
    \node[dot, label={below:$y_{4}$}] (Y4) at (4,-5) {};

    \foreach \i in {1,...,3} {
       \draw[opacity=0.5] (Y\i) -- (D\i);
    }

    \draw (Y1) --node[label={left:$12$}] {} (D1);
    \draw (Y2) --node[label={left:$8$}] {} (D2);
    \draw (Y3) --node[label={left:$4$}] {} (D3);
    
    \draw[opacity=0.5] (Y1) -- (Y2) -- (Y3) -- (Y4);
    \node at (0,-5) {$Y$};

    \draw[opacity=0.25] (C1) -- (B1);   
    \draw[opacity=0.25] (C3) -- (B2); 
    \draw[opacity=0.25] (C4) -- (B4); 
    \draw[opacity=0.25] (C4) -- (B1); 
    \draw[opacity=0.25] (C5) -- (B5); 
    \draw[opacity=0.25] (C6) -- (B2);
    \draw[opacity=0.25] (C1) -- (B6); 
    \draw[opacity=0.25] (C2) -- (B6); 

    \draw[opacity=0.25] (C2) -- (D1);
    \draw[opacity=0.25] (C4) -- (D1);
    \draw[opacity=0.25] (C6) -- (D1);
    \draw[opacity=0.25] (B5) -- (A1);
    \draw[opacity=0.25] (B3) -- (A1);

    \draw[opacity=0.25] (C1) -- (D2);
    \draw[opacity=0.25] (C5) -- (D2);
    \draw[opacity=0.25] (B1) -- (A2);
    \draw[opacity=0.25] (B4) -- (A2);

    \draw[opacity=0.25] (C4) -- (D3);
    \draw[opacity=0.25] (C2) -- (D3);
    \draw[opacity=0.25] (B6) -- (A3);

    \begin{scope}[scale=0.8, shift={(16,-6)}]
        \begin{scope}[shift={(-3.5,4)},rotate=45]
        \foreach \i in {1,...,6} {
              \node (L\i) at (1*\i-3.5,0) {$b_{\i}$};
            }
            \draw[thick] (0,0) ellipse (3.4 and 0.8);
        \end{scope}
        
        \begin{scope}[shift={(3.5,4)},rotate=-45]
            \foreach \i in {1,...,6} {
              \node (R\i) at (1*\i-3.5,0) {$c_{\i}$};
            }
            \draw[thick] (0,0) ellipse (3.4 and 0.8);
        \end{scope}
        
        \foreach \i in {1,...,3} {
          \node (S\i) at (1*\i-2,0) {$a_{\i}$};
        }
        \draw[thick] (0,0) ellipse (2.0 and 0.8);

        \draw[opacity=0.25] (R1) -- (L1);   
        \draw[opacity=0.25] (R3) -- (L2); 
        \draw[opacity=0.25] (R4) -- (L4); 
        \draw[opacity=0.25] (R4) -- (L1); 
        \draw[opacity=0.25] (R5) -- (L5); 
        \draw[opacity=0.25] (R6) -- (L2);
        \draw[opacity=0.25] (R1) -- (L6); 
        \draw[opacity=0.25] (R2) -- (L6); 

        \draw[opacity=0.25] (R2) -- (S1);
        \draw[opacity=0.25] (R4) -- (S1);
        \draw[opacity=0.25] (R6) -- (S1);
        \draw[opacity=0.25] (L5) -- (S1);
        \draw[opacity=0.25] (L3) -- (S1);

        \draw[opacity=0.25] (R1) -- (S2);
        \draw[opacity=0.25] (R5) -- (S2);
        \draw[opacity=0.25] (L1) -- (S2);
        \draw[opacity=0.25] (L4) -- (S2);

        \draw[opacity=0.25] (R4) -- (S3);
        \draw[opacity=0.25] (R2) -- (S3);
        \draw[opacity=0.25] (L6) -- (S3);
    \end{scope}
    
\end{tikzpicture}}
   \caption{
        An example of the construction of $\bG$ for the Sparse $(\sqrt{n},n,n)$-Triangle Detection Problem.
    }
   \label{fig:centroid}
\end{figure*}

\begin{lemma}\label{lem:lb_emb}
    Given an instance of Sparse $(\sqrt{n},n,n)$-Triangle Detection, we can construct an undirected, unweighted graph $\bG$ with $\Oh(n)$ vertices and $\Oh(m)$ edges in time $\Oh(n + m)$ such that the triangle problem can be solved by reading $\Oh(\sqrt{n})$ outputs of a 2-\FRP computation on $\bG$. Furthermore, if the objective is the minimum weight version rather than detection, a similar construction holds where $\bG$ is assigned rational weights in the range $[1,2]$.
\end{lemma}
We note that our construction actually reduces the potentially harder All-Nodes version of Sparse $(\sqrt{n},n,n)$-Triangle to 2-\FRP.

\begin{proof}
    Let $(A \cup B \cup C, E)$ be an instance of sparse $(\sqrt{n},n,n)$-triangle detection, where $B = \{b_1, \ldots, b_n\}$, $C = \{c_1, \ldots, c_n\}$ and $A = \{a_1, \ldots, a_{\ell}\}$ for some $\ell = \Oh(\sqrt{n})$.

    We construct the graph $\bG$ as follows (see \cref{fig:centroid} for an illustration). The vertex set of $\bG$ is partitioned into $A, B, C, A', X, Y$, where $A, B, C$ are inherited from the input instance. The sets $A', X, Y$ contain copies of the vertices in $A$; specifically, $A'$ (resp. $X$ and $Y$) contains a copy $a_i'$ (resp. $x_i$ and $y_i$) for each $i \in \fragment{1}{\ell}$.
    In $X$ and $Y$ we additionally add two dummy vertices $x_{\ell+1}$ and $y_{\ell+1}$.
    
    The edge set of $\bG$ contains the same edges as the triangle detection instance, with the modification that edges of the type $\{c_i, a_j\}$ in the original instance become edges of the form $\{c_i, a_j'\}$ in $\bG$. For each $i \in \fragment{1}{\ell}$, we insert a path of length $4(n - i + 1)$ from $x_i$ to $a_i$ and from $a_i'$ to $y_i$. Lastly, for each $i \in \fragment{1}{\ell}$, we connect $x_i$ with $x_{i+1}$ and $y_i$ with $y_{i+1}$.

    It is not difficult to see that the constructed graph $\bG$ contains $\Oh(n + \sum_{i=1}^{\ell} 4(n - i + 1)) = \Oh(n)$ nodes and $\Oh(n + \sum_{i=1}^{\ell} 4(n - i + 1) + m) = \Oh(m)$ edges.

    For $i \in \fragment{1}{\ell}$, let $e_i \coloneqq \{x_i, x_{i+1}\}$ and $e_i' \coloneqq \{y_i, y_{i+1}\}$. Further, set $s = x_1$ and $t = y_1$. We claim that for each $i \in \fragment{1}{\ell}$, $d_{\bG \setminus \{e_i, e_i'\}}(s,t) = 11+8n-6i$ if and only if there exist $b_j$ and $c_k$ such that $\{a_i, b_j, c_k\}$ forms a triangle in the input instance (so we can get the final answer from $d_{\bG \setminus \{e_i, e_i'\}}(s,t)$ for all $i \in \fragment{1}{\ell}$).
    To this end, we analyze the cost of a path in $\bG \setminus \{e_i, e_i'\}$ that travels from $s$ to $x_i$, then to $a_i$, takes a shortest path to $a_i'$, and finally proceeds to $y_i$ and $t$. It is straightforward to verify that such a path has length $3 + 4(n-i+1) + 4(n-i+1) + (i-1) + (i-1) = 11+8n-6i$ if a triangle exists (traversing $b_j$ and $c_k$); otherwise, the cost is at least one unit higher. To complete the proof, we show that this is the only path that can achieve a cost of $11+8n-6i$. Indeed, in $\bG \setminus \{e_i, e_i'\}$, any path between $a_{\hi}$ and $a_{\hj}'$ has cost at least $3 + 4(n-\hi+1) + 4(n-\hj+1) + (\hi-1) + (\hj-1) = 11+8n-3\hi-3\hj$ for any $\hi, \hj \in \fragment{1}{i}$. In particular, if $\hi \neq i$ or $\hj \neq i$, then $11+8n-3\hi-3\hj \geq (11+8n-6i) + 3$.

    When handling rational weights, we make $\bG$ weighted as follows. Let $M$ be the largest weight in the input graph; we then assign to each edge $\{a_i,b_j\}$ of weight $w(a_i,b_j)$ in the input graph a weight of $1+w(a_i,b_j)/(M + 1)$ in $\bG$. We do the same for edges $\{b_j,c_k\}$ and $\{c_k, a_i\}$ (which are $\{c_k, a_i'\}$ in $\bG$). All other edges in $\bG$ retain their original weight of one. The proof of correctness is very similar to the unweighted case. It is easy to see that for each $i \in \fragment{1}{\ell}$, the same path we were considering before traversing $a_i,a_i'$ has cost $5+4n-2i+(w(a_i,b_j)+w(b_j,c_k) + w(c_k, a_i))/(M+1)$ for $j,k \in \fragment{1}{n}$ that minimize the sum $w(a_i,b_j)+w(b_j,c_k) + w(c_k, a_i)$. Using the same analysis as before, all paths going through other $a_{\hi}$ and $a_{\hj}'$ have cost at least $(5+4n-2i) + 1$, which is strictly greater than $5+4n-2i+(w(a_i,b_j)+w(b_j,c_k) + w(c_k, a_i))/(M+1)$ because the fractional part is less than one.
\end{proof}

With \cref{lem:lb_emb} at our disposal, \cref{thm:lb} follows almost immediately.

\lb

\begin{proof}
    For \ref{it:lb:i}, we show that any combinatorial algorithm running in $\Oh(mn^{1/2-\varepsilon})$ time implies a subcubic algorithm for triangle detection in a tripartite graph $\bG = (A \cup B \cup C, E)$ with $n$ nodes per partition (which, in turn, implies a subcubic algorithm for Boolean Matrix Multiplication).

    To this end, we partition $A$ into $\Theta(\sqrt{n})$ disjoint sets, each of size $\Oh(\sqrt{n})$, and we partition $E \cap (B \times C)$ into $\Theta(|E|/m)$ disjoint sets, each of size $\Oh(m)$. For each partition $A'$ of $A$ and $E'$ of $E \cap (B \times C)$, we apply \cref{lem:lb_emb} to the subgraph induced by $(A' \cup B \cup C, E' \cup (E \setminus (B \times C)))$ and execute 2-\FRP on the resulting instance. 
    Note that, crucially, $m = \Omega(n^{3/2})$; since $|A'| = \Oh(\sqrt{n})$, we always have $|E \cap (A' \times B)| = \Oh(m)$ and $|E \cap (A' \times C)| = \Oh(m)$. Clearly, this reduction determines whether $\bG$ contains a triangle. The overall runtime is $\Oh(\frac{n^2}{m} \cdot \sqrt{n} \cdot mn^{1/2-\varepsilon}) = \Oh(n^{3-\varepsilon})$, which contradicts the BMM hypothesis.
    
    For \ref{it:lb:ii}, we follow a similar approach, but instead of reducing from triangle detection, we reduce from the negative triangle detection problem.
\end{proof}

\newpage
\bibliographystyle{alphaurl}
\bibliography{main}

@inproceedings{CDWX25,
  author       = {Shucheng Chi and
                  Ran Duan and
                  Benyu Wang and
                  Tianle Xie},
  editor       = {Keren Censor{-}Hillel and
                  Fabrizio Grandoni and
                  Jo{\"{e}}l Ouaknine and
                  Gabriele Puppis},
  title        = {Undirected 3-Fault Replacement Path in Nearly Cubic Time},
  booktitle    = {52nd International Colloquium on Automata, Languages, and Programming,
                  {ICALP} 2025, July 8-11, 2025, Aarhus, Denmark},
  series       = {LIPIcs},
  volume       = {334},
  pages        = {57:1--57:20},
  publisher    = {Schloss Dagstuhl - Leibniz-Zentrum f{\"{u}}r Informatik},
  year         = {2025},
  url          = {https://doi.org/10.4230/LIPIcs.ICALP.2025.57},
  doi          = {10.4230/LIPICS.ICALP.2025.57},
  bibsource    = {dblp computer science bibliography, https://dblp.org}
}

@article{GV19,
    author = {Grandoni, Fabrizio and {Vassilevska Williams}, Virginia},
    title = {Faster Replacement Paths and Distance Sensitivity Oracles},
    year = {2019},
    issue_date = {January 2020},
    publisher = {Association for Computing Machinery},
    address = {New York, NY, USA},
    volume = {16},
    number = {1},
    issn = {1549-6325},
    url = {https://doi.org/10.1145/3365835},
    doi = {10.1145/3365835},
    journal = {ACM Trans. Algorithms},
    month = dec,
    articleno = {15},
    numpages = {25}
}

@article{BFC04,
    author = {Bender, Michael A. and Farach-Colton, Mart\'{\i}n},
    title = {The level ancestor problem simplified},
    year = {2004},
    issue_date = {June 16, 2004},
    publisher = {Elsevier Science Publishers Ltd.},
    address = {GBR},
    volume = {321},
    number = {1},
    issn = {0304-3975},
    url = {https://doi.org/10.1016/j.tcs.2003.05.002},
    doi = {10.1016/j.tcs.2003.05.002},
    journal = {Theor. Comput. Sci.},
    month = jun,
    pages = {5–12},
    numpages = {8}
}

@article{BV94,
    author = {Berkman, Omer and Vishkin, Uzi},
    title = {Finding level-ancestors in trees},
    year = {1994},
    issue_date = {April 1994},
    publisher = {Academic Press, Inc.},
    address = {USA},
    volume = {48},
    number = {2},
    issn = {0022-0000},
    url = {https://doi.org/10.1016/S0022-0000(05)80002-9},
    doi = {10.1016/S0022-0000(05)80002-9},
    journal = {J. Comput. Syst. Sci.},
    month = apr,
    pages = {214–230},
    numpages = {17}
}

@inproceedings{D91,
  author       = {Paul F. Dietz},
  editor       = {Frank K. H. A. Dehne and
                  J{\"{o}}rg{-}R{\"{u}}diger Sack and
                  Nicola Santoro},
  title        = {Finding Level-Ancestors in Dynamic Trees},
  booktitle    = {Algorithms and Data Structures, 2nd Workshop {WADS} '91, Ottawa, Canada,
                  August 14-16, 1991, Proceedings},
  series       = {Lecture Notes in Computer Science},
  volume       = {519},
  pages        = {32--40},
  publisher    = {Springer},
  year         = {1991},
  url          = {https://doi.org/10.1007/BFb0028247},
  doi          = {10.1007/BFB0028247},
  bibsource    = {dblp computer science bibliography, https://dblp.org}
}

@inproceedings{BFC00,
  author       = {Michael A. Bender and
                  Martin Farach{-}Colton},
  editor       = {Gaston H. Gonnet and
                  Daniel Panario and
                  Alfredo Viola},
  title        = {The {LCA} Problem Revisited},
  booktitle    = {{LATIN} 2000: Theoretical Informatics, 4th Latin American Symposium,
                  Punta del Este, Uruguay, April 10-14, 2000, Proceedings},
  series       = {Lecture Notes in Computer Science},
  volume       = {1776},
  pages        = {88--94},
  publisher    = {Springer},
  year         = {2000},
  url          = {https://doi.org/10.1007/10719839\_9},
  doi          = {10.1007/10719839\_9},
  bibsource    = {dblp computer science bibliography, https://dblp.org}
}

@article{NPW01,
    author       = {Enrico Nardelli and
                  Guido Proietti and
                  Peter Widmayer},
    title        = {A faster computation of the most vital edge of a shortest path},
    journal      = {Inf. Process. Lett.},
    volume       = {79},
    number       = {2},
    pages        = {81--85},
    year         = {2001},
    url          = {https://doi.org/10.1016/S0020-0190(00)00175-7},
    doi          = {10.1016/S0020-0190(00)00175-7},
    bibsource    = {dblp computer science bibliography, https://dblp.org}
}

@article{MMG89,
    author = {Malik, K. and Mittal, A. K. and Gupta, S. K.},
    title = {The k most vital arcs in the shortest path problem},
    year = {1989},
    issue_date = {August, 1989},
    publisher = {Elsevier Science Publishers B. V.},
    address = {NLD},
    volume = {8},
    number = {4},
    issn = {0167-6377},
    url = {https://doi.org/10.1016/0167-6377(89)90065-5},
    doi = {10.1016/0167-6377(89)90065-5},
    journal = {Oper. Res. Lett.},
    month = aug,
    pages = {223–227},
    numpages = {5}
}

@techreport{SBK95,
    author = {Schieber, Baruch and Bar-Noy, Amotz and Khuller, Samir},
    title = {The complexity of finding most vital arcs and nodes},
    year = {1995},
    publisher = {University of Maryland at College Park},
    address = {USA}
}

@inproceedings{BBAKCM01,
    author = {Bremler-Barr, Anat and Afek, Yehuda and Kaplan, Haim and Cohen, Edith and Merritt, Michael},
    title = {Restoration by path concatenation: fast recovery of MPLS paths},
    year = {2001},
    isbn = {1581133839},
    publisher = {Association for Computing Machinery},
    address = {New York, NY, USA},
    url = {https://doi.org/10.1145/383962.383980},
    doi = {10.1145/383962.383980},
    booktitle = {Proceedings of the Twentieth Annual ACM Symposium on Principles of Distributed Computing},
    pages = {43–52},
    numpages = {10},
    location = {Newport, Rhode Island, USA},
    series = {PODC '01}
}

@inproceedings{BCFS21,
  author       = {Davide Bil{\`{o}} and
                  Sarel Cohen and
                  Tobias Friedrich and
                  Martin Schirneck},
  editor       = {Petra Mutzel and
                  Rasmus Pagh and
                  Grzegorz Herman},
  title        = {Near-Optimal Deterministic Single-Source Distance Sensitivity Oracles},
  booktitle    = {29th Annual European Symposium on Algorithms, {ESA} 2021, Lisbon,
                  Portugal (Virtual Conference), September 6-8, 2021},
  series       = {LIPIcs},
  volume       = {204},
  pages        = {18:1--18:17},
  publisher    = {Schloss Dagstuhl - Leibniz-Zentrum f{\"{u}}r Informatik},
  year         = {2021},
  url          = {https://doi.org/10.4230/LIPIcs.ESA.2021.18},
  doi          = {10.4230/LIPICS.ESA.2021.18},
  bibsource    = {dblp computer science bibliography, https://dblp.org}
}

@inproceedings{CM20,
  author       = {Shiri Chechik and
                  Ofer Magen},
  editor       = {Artur Czumaj and
                  Anuj Dawar and
                  Emanuela Merelli},
  title        = {Near Optimal Algorithm for the Directed Single Source Replacement
                  Paths Problem},
  booktitle    = {47th International Colloquium on Automata, Languages, and Programming,
                  {ICALP} 2020, Saarbr{\"{u}}cken, Germany (Virtual Conference),
                  July 8-11, 2020},
  series       = {LIPIcs},
  volume       = {168},
  pages        = {81:1--81:17},
  publisher    = {Schloss Dagstuhl - Leibniz-Zentrum f{\"{u}}r Informatik},
  year         = {2020},
  url          = {https://doi.org/10.4230/LIPIcs.ICALP.2020.81},
  doi          = {10.4230/LIPICS.ICALP.2020.81},
  bibsource    = {dblp computer science bibliography, https://dblp.org}
}

@inproceedings{CC19,
  author       = {Shiri Chechik and
                  Sarel Cohen},
  editor       = {Timothy M. Chan},
  title        = {Near Optimal Algorithms For The Single Source Replacement Paths Problem},
  booktitle    = {Proceedings of the Thirtieth Annual {ACM-SIAM} Symposium on Discrete
                  Algorithms, {SODA} 2019, San Diego, California, USA, January 6-9,
                  2019},
  pages        = {2090--2109},
  publisher    = {{SIAM}},
  year         = {2019},
  url          = {https://doi.org/10.1137/1.9781611975482.126},
  doi          = {10.1137/1.9781611975482.126},
  bibsource    = {dblp computer science bibliography, https://dblp.org}
}

@inproceedings{CT24,
  author       = {Shiri Chechik and
                  Tianyi Zhang},
  editor       = {Karl Bringmann and
                  Martin Grohe and
                  Gabriele Puppis and
                  Ola Svensson},
  title        = {Faster Algorithms for Dual-Failure Replacement Paths},
  booktitle    = {51st International Colloquium on Automata, Languages, and Programming,
                  {ICALP} 2024, Tallinn, Estonia, July 8-12, 2024},
  series       = {LIPIcs},
  volume       = {297},
  pages        = {41:1--41:20},
  publisher    = {Schloss Dagstuhl - Leibniz-Zentrum f{\"{u}}r Informatik},
  year         = {2024},
  url          = {https://doi.org/10.4230/LIPIcs.ICALP.2024.41},
  doi          = {10.4230/LIPICS.ICALP.2024.41},
  bibsource    = {dblp computer science bibliography, https://dblp.org}
}

@INPROCEEDINGS{VWWX22,
  author={{Vassilevska Williams}, Virginia  and Woldeghebriel, Eyob and Xu, Yinzhan},
  booktitle={2022 IEEE 63rd Annual Symposium on Foundations of Computer Science (FOCS)}, 
  title={Algorithms and Lower Bounds for Replacement Paths under Multiple Edge Failure}, 
  year={2022},
  volume={},
  number={},
  pages={907-918},
  doi={10.1109/FOCS54457.2022.00090}}

@article{GL09,
    author = {Gotthilf, Zvi and Lewenstein, Moshe},
    title = {Improved algorithms for the k simple shortest paths and the replacement paths problems},
    year = {2009},
    issue_date = {March, 2009},
    publisher = {Elsevier North-Holland, Inc.},
    address = {USA},
    volume = {109},
    number = {7},
    issn = {0020-0190},
    url = {https://doi.org/10.1016/j.ipl.2008.12.015},
    doi = {10.1016/j.ipl.2008.12.015},
    journal = {Inf. Process. Lett.},
    month = mar,
    pages = {352–355},
    numpages = {4}
}

@inproceedings{VW11,
    author = {{Vassilevska Williams}, Virginia},
    title = {Faster replacement paths},
    year = {2011},
    publisher = {Society for Industrial and Applied Mathematics},
    address = {USA},
    booktitle = {Proceedings of the Twenty-Second Annual ACM-SIAM Symposium on Discrete Algorithms},
    pages = {1337–1346},
    numpages = {10},
    location = {San Francisco, California},
    series = {SODA '11}
}

@article{RZ12,
  author       = {Liam Roditty and
                  Uri Zwick},
  title        = {Replacement paths and \emph{k} simple shortest paths in unweighted
                  directed graphs},
  journal      = {{ACM} Trans. Algorithms},
  volume       = {8},
  number       = {4},
  pages        = {33:1--33:11},
  year         = {2012}}

@article{L72,
    ISSN = {00251909, 15265501},
    URL = {http://www.jstor.org/stable/2629357},
    author = {Eugene L. Lawler},
    journal = {Management Science},
    number = {7},
    pages = {401--405},
    publisher = {INFORMS},
    title = {A Procedure for Computing the K Best Solutions to Discrete Optimization Problems and Its Application to the Shortest Path Problem},
    urldate = {2026-01-27},
    volume = {18},
    year = {1972}
}

@article{BW84,
    ISSN = {0030364X, 15265463},
    URL = {http://www.jstor.org/stable/170956},
    author = {Thomas H. Byers and Michael S. Waterman},
    journal = {Operations Research},
    number = {6},
    pages = {1381--1384},
    publisher = {INFORMS},
    title = {Determining All Optimal and Near-Optimal Solutions When Solving Shortest Path Problems by Dynamic Programming},
    urldate = {2026-01-27},
    volume = {32},
    year = {1984}
}

@article{NR01,
    title = {Algorithmic Mechanism Design},
    journal = {Games and Economic Behavior},
    volume = {35},
    number = {1},
    pages = {166-196},
    year = {2001},
    issn = {0899-8256},
    doi = {https://doi.org/10.1006/game.1999.0790},
    url = {https://www.sciencedirect.com/science/article/pii/S089982569990790X},
    author = {Noam Nisan and Amir Ronen},
}

@inproceedings{vsurvey,
  title={On some fine-grained questions in algorithms and complexity},
  author={{Vassilevska Williams}, Virginia},
  booktitle={Proceedings of the ICM},
  volume={3},
  pages={3431--3472},
  year={2018},
  organization={World Scientific}
}

@inproceedings{HS01,
    author = {Hershberger, J. and Suri, S.},
    title = {Vickrey Prices and Shortest Paths: What is an Edge Worth?},
    year = {2001},
    isbn = {0769513905},
    publisher = {IEEE Computer Society},
    address = {USA},
    booktitle = {Proceedings of the 42nd IEEE Symposium on Foundations of Computer Science},
    pages = {252},
    series = {FOCS '01}
}

@article{Y71,
 ISSN = {00251909, 15265501},
 URL = {http://www.jstor.org/stable/2629312},
 author = {Jin Y. Yen},
 journal = {Management Science},
 number = {11},
 pages = {712--716},
 publisher = {INFORMS},
 title = {Finding the K Shortest Loopless Paths in a Network},
 urldate = {2026-01-27},
 volume = {17},
 year = {1971}
}

@inproceedings{R07,
    author = {Roditty, Liam},
    title = {On the K-simple shortest paths problem in weighted directed graphs},
    year = {2007},
    isbn = {9780898716245},
    publisher = {Society for Industrial and Applied Mathematics},
    address = {USA},
    booktitle = {Proceedings of the Eighteenth Annual ACM-SIAM Symposium on Discrete Algorithms},
    pages = {920–928},
    numpages = {9},
    location = {New Orleans, Louisiana},
    series = {SODA '07}
}

@inproceedings{B10,
  author       = {Aaron Bernstein},
  editor       = {Moses Charikar},
  title        = {A Nearly Optimal Algorithm for Approximating Replacement Paths and
                  k Shortest Simple Paths in General Graphs},
  booktitle    = {Proceedings of the Twenty-First Annual {ACM-SIAM} Symposium on Discrete
                  Algorithms, {SODA} 2010, Austin, Texas, USA, January 17-19, 2010},
  pages        = {742--755},
  publisher    = {{SIAM}},
  year         = {2010},
  url          = {https://doi.org/10.1137/1.9781611973075.61},
  doi          = {10.1137/1.9781611973075.61},
  bibsource    = {dblp computer science bibliography, https://dblp.org}
}

@article{WY13,
  author       = {Oren Weimann and
                  Raphael Yuster},
  title        = {Replacement Paths and Distance Sensitivity Oracles via Fast Matrix
                  Multiplication},
  journal      = {{ACM} Trans. Algorithms},
  volume       = {9},
  number       = {2},
  pages        = {14:1--14:13},
  year         = {2013},
  url          = {https://doi.org/10.1145/2438645.2438646},
  doi          = {10.1145/2438645.2438646},
  bibsource    = {dblp computer science bibliography, https://dblp.org}
}

@inproceedings{fullysparse,
  author       = {Amir Abboud and
                  Karl Bringmann and
                  Nick Fischer and
                  Marvin K{\"{u}}nnemann},
  title        = {The Time Complexity of Fully Sparse Matrix Multiplication},
  booktitle    = {Proceedings of the 2024 {ACM-SIAM} Symposium on Discrete Algorithms,
                  {SODA} 2024, Alexandria, VA, USA, January 7-10, 2024},
  pages        = {4670--4703},
  publisher    = {{SIAM}},
  year         = {2024}}

@inproceedings{RodittyW12,
  author       = {Liam Roditty and
                  Virginia {Vassilevska Williams}},
  editor       = {Yuval Rabani},
  title        = {Subquadratic time approximation algorithms for the girth},
  booktitle    = {Proceedings of the Twenty-Third Annual {ACM-SIAM} Symposium on Discrete
                  Algorithms, {SODA} 2012, Kyoto, Japan, January 17-19, 2012},
  pages        = {833--845},
  publisher    = {{SIAM}},
  year         = {2012}}

@article{VWW18,
  author       = {Virginia {Vassilevska Williams} and
                  R. Ryan Williams},
  title        = {Subcubic Equivalences Between Path, Matrix, and Triangle Problems},
  journal      = {J. {ACM}},
  volume       = {65},
  number       = {5},
  pages        = {27:1--27:38},
  year         = {2018},
  url          = {https://doi.org/10.1145/3186893},
  doi          = {10.1145/3186893},
  bibsource    = {dblp computer science bibliography, https://dblp.org}
}

@inproceedings{GPVWX21,
    author       = {Yuzhou Gu and
                  Adam Polak and
                  Virginia {Vassilevska Williams} and
                  Yinzhan Xu},
    editor       = {Nikhil Bansal and
                  Emanuela Merelli and
                  James Worrell},
    title        = {Faster Monotone Min-Plus Product, Range Mode, and Single Source Replacement
                  Paths},
    booktitle    = {48th International Colloquium on Automata, Languages, and Programming,
                  {ICALP} 2021, Glasgow, Scotland (Virtual Conference), July 12-16,
                  2021},
    series       = {LIPIcs},
    volume       = {198},
    pages        = {75:1--75:20},
    publisher    = {Schloss Dagstuhl - Leibniz-Zentrum f{\"{u}}r Informatik},
    year         = {2021},
    url          = {https://doi.org/10.4230/LIPIcs.ICALP.2021.75},
    doi          = {10.4230/LIPICS.ICALP.2021.75},
    bibsource    = {dblp computer science bibliography, https://dblp.org}
}

@inproceedings{CZ24,
    author       = {Shiri Chechik and
                  Tianyi Zhang},
    editor       = {David P. Woodruff},
    title        = {Nearly Optimal Approximate Dual-Failure Replacement Paths},
    booktitle    = {Proceedings of the 2024 {ACM-SIAM} Symposium on Discrete Algorithms,
                  {SODA} 2024, Alexandria, VA, USA, January 7-10, 2024},
    pages        = {2568--2596},
    publisher    = {{SIAM}},
    year         = {2024},
    url          = {https://doi.org/10.1137/1.9781611977912.91},
    doi          = {10.1137/1.9781611977912.91},
    bibsource    = {dblp computer science bibliography, https://dblp.org}
}

@inproceedings{BLM12,
    author = {Baswana, Surender and Lath, Utkarsh and Mehta, Anuradha S.},
    title = {Single source distance oracle for planar digraphs avoiding a failed node or link},
    year = {2012},
    publisher = {Society for Industrial and Applied Mathematics},
    address = {USA},
    booktitle = {Proceedings of the Twenty-Third Annual ACM-SIAM Symposium on Discrete Algorithms},
    pages = {223–232},
    numpages = {10},
    location = {Kyoto, Japan},
    series = {SODA '12}
}

@InProceedings{HLNVW17,
  author =	{Henzinger, Monika and Lincoln, Andrea and Neumann, Stefan and {Vassilevska Williams}, Virginia},
  title =	{{Conditional Hardness for Sensitivity Problems}},
  booktitle =	{8th Innovations in Theoretical Computer Science Conference (ITCS 2017)},
  pages =	{26:1--26:31},
  series =	{Leibniz International Proceedings in Informatics (LIPIcs)},
  ISBN =	{978-3-95977-029-3},
  ISSN =	{1868-8969},
  year =	{2017},
  volume =	{67},
  editor =	{Papadimitriou, Christos H.},
  publisher =	{Schloss Dagstuhl -- Leibniz-Zentrum f{\"u}r Informatik},
  address =	{Dagstuhl, Germany},
  URL =		{https://drops.dagstuhl.de/entities/document/10.4230/LIPIcs.ITCS.2017.26},
  URN =		{urn:nbn:de:0030-drops-81783},
  doi =		{10.4230/LIPIcs.ITCS.2017.26}
}

@inproceedings{HKIM24,
  author       = {Kaito Harada and
                  Naoki Kitamura and
                  Taisuke Izumi and
                  Toshimitsu Masuzawa},
  editor       = {Timothy M. Chan and
                  Johannes Fischer and
                  John Iacono and
                  Grzegorz Herman},
  title        = {A Nearly Linear Time Construction of Approximate Single-Source Distance
                  Sensitivity Oracles},
  booktitle    = {32nd Annual European Symposium on Algorithms, {ESA} 2024, Royal Holloway,
                  London, United Kingdom, September 2-4, 2024},
  series       = {LIPIcs},
  volume       = {308},
  pages        = {65:1--65:18},
  publisher    = {Schloss Dagstuhl - Leibniz-Zentrum f{\"{u}}r Informatik},
  year         = {2024},
  url          = {https://doi.org/10.4230/LIPIcs.ESA.2024.65},
  doi          = {10.4230/LIPICS.ESA.2024.65},
  bibsource    = {dblp computer science bibliography, https://dblp.org}
}

@article{HSB07,
  author       = {John Hershberger and
                  Subhash Suri and
                  Amit M. Bhosle},
  title        = {On the difficulty of some shortest path problems},
  journal      = {{ACM} Trans. Algorithms},
  volume       = {3},
  number       = {1},
  pages        = {5:1--5:15},
  year         = {2007},
  url          = {https://doi.org/10.1145/1219944.1219951},
  doi          = {10.1145/1219944.1219951},
  bibsource    = {dblp computer science bibliography, https://dblp.org}
}

@article{KKP,
  author       = {David R. Karger and
                  Daphne Koller and
                  Steven J. Phillips},
  title        = {Finding the Hidden Path: Time Bounds for All-Pairs Shortest Paths},
  journal      = {{SIAM} J. Comput.},
  volume       = {22},
  number       = {6},
  pages        = {1199--1217},
  year         = {1993},
  url          = {https://doi.org/10.1137/0222071},
  doi          = {10.1137/0222071},
  bibsource    = {dblp computer science bibliography, https://dblp.org}
}

@inproceedings{GVW12,
  author       = {Fabrizio Grandoni and
                  Virginia {Vassilevska Williams}},
  title        = {Improved Distance Sensitivity Oracles via Fast Single-Source Replacement
                  Paths},
  booktitle    = {53rd Annual {IEEE} Symposium on Foundations of Computer Science, {FOCS}
                  2012, New Brunswick, NJ, USA, October 20-23, 2012},
  pages        = {748--757},
  publisher    = {{IEEE} Computer Society},
  year         = {2012},
  url          = {https://doi.org/10.1109/FOCS.2012.17},
  doi          = {10.1109/FOCS.2012.17},
  bibsource    = {dblp computer science bibliography, https://dblp.org}
}

@article{BK13,
  author       = {Surender Baswana and
                  Neelesh Khanna},
  title        = {Approximate Shortest Paths Avoiding a Failed Vertex: Near Optimal
                  Data Structures for Undirected Unweighted Graphs},
  journal      = {Algorithmica},
  volume       = {66},
  number       = {1},
  pages        = {18--50},
  year         = {2013},
  url          = {https://doi.org/10.1007/s00453-012-9621-y},
  doi          = {10.1007/S00453-012-9621-Y},
  bibsource    = {dblp computer science bibliography, https://dblp.org}
}

@article{BCHL20,
  author       = {Surender Baswana and
                  Keerti Choudhary and
                  Moazzam Hussain and
                  Liam Roditty},
  title        = {Approximate Single-Source Fault Tolerant Shortest Path},
  journal      = {{ACM} Trans. Algorithms},
  volume       = {16},
  number       = {4},
  pages        = {44:1--44:22},
  year         = {2020},
  url          = {https://doi.org/10.1145/3397532},
  doi          = {10.1145/3397532},
  bibsource    = {dblp computer science bibliography, https://dblp.org}
}

@inproceedings{BCGLPP18,
  author       = {Davide Bil{\`{o}} and
                  Keerti Choudhary and
                  Luciano Gual{\`{a}} and
                  Stefano Leucci and
                  Merav Parter and
                  Guido Proietti},
  editor       = {Rolf Niedermeier and
                  Brigitte Vall{\'{e}}e},
  title        = {Efficient Oracles and Routing Schemes for Replacement Paths},
  booktitle    = {35th Symposium on Theoretical Aspects of Computer Science, {STACS}
                  2018, Caen, France, February 28 - March 3, 2018},
  series       = {LIPIcs},
  volume       = {96},
  pages        = {13:1--13:15},
  publisher    = {Schloss Dagstuhl - Leibniz-Zentrum f{\"{u}}r Informatik},
  year         = {2018},
  url          = {https://doi.org/10.4230/LIPIcs.STACS.2018.13},
  doi          = {10.4230/LIPICS.STACS.2018.13},
  bibsource    = {dblp computer science bibliography, https://dblp.org}
}

@inproceedings{BG04,
    author = {Bhosle, Amit and Gonzalez, Teofilo},
    year = {2004},
    month = {09},
    pages = {},
    booktitle = {Proc. 16th IASTED International Conference on Parallel and Distributed
Computing and System (PDCS)},
    title = {Replacement Paths for Pairs of Shortest Path Edges in Directed Graphs}
}

@article{W18,
    author       = {R. Ryan Williams},
    title        = {Faster All-Pairs Shortest Paths via Circuit Complexity},
    journal      = {{SIAM} J. Comput.},
    volume       = {47},
    number       = {5},
    pages        = {1965--1985},
    year         = {2018},
    url          = {https://doi.org/10.1137/15M1024524},
    doi          = {10.1137/15M1024524},
    bibsource    = {dblp computer science bibliography, https://dblp.org}
}

@inproceedings{BW25,
    author       = {Greg Bodwin and
                  Lily Wang},
    editor       = {Yossi Azar and
                  Debmalya Panigrahi},
    title        = {Improved Shortest Path Restoration Lemmas for Multiple Edge Failures:
                  Trade-offs Between Fault-tolerance and Subpaths},
    booktitle    = {Proceedings of the 2025 Annual {ACM-SIAM} Symposium on Discrete Algorithms,
                  {SODA} 2025, New Orleans, LA, USA, January 12-15, 2025},
    pages        = {5245--5262},
    publisher    = {{SIAM}},
    year         = {2025},
    url          = {https://doi.org/10.1137/1.9781611978322.179},
    doi          = {10.1137/1.9781611978322.179},
    bibsource    = {dblp computer science bibliography, https://dblp.org}
}

\end{document}